\documentclass[11pt]{article}
\usepackage{amsfonts}
\usepackage{amsmath}
\usepackage{lipsum}
\usepackage{amssymb}
\usepackage{amsmath}
\usepackage{amsthm}
\usepackage{multirow}
\usepackage{booktabs}
\usepackage[dvipsnames]{xcolor} 
\usepackage{xr}
\usepackage[letterpaper, margin=1.25in]{geometry}
\usepackage{setspace}

\usepackage[authoryear]{natbib}
\renewcommand{\theequation}{\thesection.\arabic{equation}}
\newcommand{\myref}[2]{\hyperref[#1]{#2}}
\numberwithin{equation}{section}

\usepackage{latexsym}
\usepackage{graphicx} 
\usepackage{enumerate}
\usepackage{setspace}
\usepackage[toc,title,titletoc,header]{appendix}
\usepackage[bottom]{footmisc} 
\usepackage[pdfborder={0 0 0}]{hyperref}
\hypersetup{
 colorlinks=true,linkcolor=magenta, citecolor=blue, filecolor=magenta, 
 linkbordercolor={0 1 1}, citebordercolor={1 0 0},
}

\usepackage{lineno}
\setpagewiselinenumbers

\allowdisplaybreaks

\newtheorem{theorem}{Theorem}
\newtheorem{lemma}{Lemma}
\newtheorem{assumption}{Assumption}
\newtheorem{definition}{Definition}
\newtheorem{example}{Example}
\newtheorem{remark}{Remark}
\newcounter{assumptionM}
\newcounter{assumptionA}
\def\theassumptionM{M.\arabic{assumptionM}}
\def\theassumptionA{A.\arabic{assumptionA}}
\begin{document}
	\relax
	\hypersetup{pageanchor=false}
	\hypersetup{pageanchor=true}
\author{%
\begin{tabular*}{\textwidth}{@{\extracolsep{\fill}}cc@{}}
\begin{minipage}[t]{0.46\textwidth}
\centering
Federico A. Bugni\\
Department of Economics\\
Northwestern University\\
\href{mailto:federico.bugni@northwestern.edu}{\texttt{federico.bugni@northwestern.edu}}
\end{minipage}
&
\begin{minipage}[t]{0.46\textwidth}
\centering
Mengsi Gao\\
Department of Economics\\
USC\\
\href{mailto:mengsi.gao@usc.edu}{\texttt{mengsi.gao@usc.edu}}
\end{minipage}
\\[2.5cm]
\begin{minipage}[t]{0.46\textwidth}
\centering
Filip Obradovi{\'c}\\
Department of Economics\\
UCLA\\
\href{mailto:obradovicfilip@ucla.edu}{\texttt{obradovicfilip@ucla.edu}}
\end{minipage}
&
\begin{minipage}[t]{0.46\textwidth}
\centering
Amilcar Velez\\
Department of Economics\\
Cornell University\\
\href{mailto:amilcare@cornell.edu}{\texttt{amilcare@cornell.edu}}
\end{minipage}
\end{tabular*}%
}

\title{\vspace{-1.7cm}On the Power Properties of Inference for Parameters\\ with Interval Identified Sets\thanks{Corresponding author: \href{mailto:federico.bugni@northwestern.edu}{federico.bugni@northwestern.edu}. We thank the Co-Editors, Arthur Lewbel and Patrik Guggenberger, and three anonymous referees for their comments and suggestions, which have greatly improved the manuscript. We also thank Eric Auerbach, Ivan Canay, Joel Horowitz, Chuck Manski, and seminar participants at Northwestern for their helpful comments on this paper. Obradovi{\'c} and Velez gratefully acknowledge financial support from the Robert Eisner Memorial Fellowship.
}\vspace{0.5cm} 
}

\maketitle

\vspace{-0.5in}
\thispagestyle{empty} 

\begin{spacing}{1.2}

\begin{abstract}
This paper studies the power properties of confidence intervals (CIs) for a partially-identified parameter of interest with an interval identified set. We assume the researcher has bounds estimators needed to construct the CIs proposed by \cite{imbens/manski:2004}, \cite{stoye:2009}, and \cite{stoye:2020}, denoted by $CI_{\alpha}^{1}$, $CI_{\alpha}^{2}$, $CI_{\alpha}^{3}$, and $CI_{\alpha}^{4}$. We also assume these bounds estimators are ``ordered'': the lower bound estimator is less than or equal to the upper bound estimator. This setup arises in economic applications involving missing data and treatment effects.

Under these conditions, we establish two results. First, we show that $CI_{\alpha}^{1}$ and $CI_{\alpha}^{2}$ are equally powerful, and both dominate $CI_{\alpha}^{3}$ and $CI_{\alpha}^{4}$. Second, we consider a favorable situation in which there are two possible bounds estimators to construct these CIs, and one is more efficient than the other. One would expect that the more efficient bounds estimator yields more powerful inference. We prove that this desirable result holds for $CI_{\alpha}^{1}$ and $CI_{\alpha}^{2}$, but not necessarily for $CI_{\alpha}^{3}$ or $CI_{\alpha}^{4}$. In summary, within the class of models considered, $CI_{\alpha}^{1}$ and $CI_{\alpha}^{2}$ have identical power properties, and both compare favorably to $CI_{\alpha}^{3}$ or $CI_{\alpha}^{4}$.
\end{abstract}
\end{spacing}



\thispagestyle{empty}  
\newpage

\section{Introduction}

This paper contributes to the literature on inference in partially-identified econometric models. Our setup is as in \cite{imbens/manski:2004} and \cite{stoye:2009} where the econometric model implies that, for a data distribution $P$, the real-valued parameter of interest $\theta_{0}(P)$ belongs to an interval identified set $[ \theta_{l}(P), \theta_{u}(P) ]$. We focus on the case in which the endpoints of the interval do not cross, i.e., $\theta_{l}(P) \leq \theta_{u}(P)$, so the identified set is non-empty.  

We assume the researcher can implement asymptotically valid inference for the partially-identified parameter based on asymptotically normal estimators of the identified set's endpoints. That is, we assume the availability of a pair of estimators $( \hat{\theta}_{l},\hat{\theta}_{u}) $ such that
\begin{equation*}
    \sqrt{N}
    \left( \begin{array}{c}
~\hat{\theta}_{l}-\theta _{l}(P)\\\hat{\theta}_{u}-\theta _{u}(P) 
    \end{array}~\right)
    ~\overset{d}{ \to }~\mathcal{N}(  {\bf 0}_{2\times 1}, \Sigma ( P) ) 
\end{equation*}
uniformly in a suitable set of distributions, along with a uniformly consistent estimator of $\Sigma (P) $. Furthermore, we assume that the bounds estimators are ``ordered'', in the sense that $\hat{\theta}_{l}\leq \hat{\theta}_{u}$ with probability one. We refer to these conditions as the ``ordered bounds setup'' or OBS.

Under our conditions, the researcher has several natural options for constructing a confidence interval (CI, henceforth) for $\theta _{0}( P) $ with a confidence level of $(1-\alpha)$. The first option is the CI proposed by \cite{imbens/manski:2004}, which we denote by $CI_{\alpha }^{1}$. Two additional options were introduced by \cite{stoye:2009}, and we denote them by $CI_{\alpha }^{2}$ and $CI_{\alpha }^{3}$. A final option was recently introduced by \cite{stoye:2020}, and we denote it by $CI_{\alpha }^{4}$. 

One notable aspect is that $CI_{\alpha}^{1}$ is straightforward to implement, whereas $CI_{\alpha}^{2}$, $CI_{\alpha}^{3}$ and $CI_{\alpha }^{4}$ are relatively more complex. \cite{stoye:2009} shows that $CI_{\alpha}^{1}$ and $CI_{\alpha}^{2}$ rely on the so-called superefficiency condition, whereas $CI_{\alpha}^{3}$ dispenses with this requirement at the expense of introducing a tuning parameter that must be selected in practice. See Section \ref{sec:CI} for a detailed description of these CIs. \cite{stoye:2009} shows that $CI_{\alpha }^{1}$, $CI_{\alpha }^{2}$, and $CI_{\alpha }^{3}$ are all uniformly asymptotically valid and exact. In addition, \cite{stoye:2020} shows the uniform asymptotic validity of $CI_{\alpha }^{4}$. 
This implies that the four CIs are equivalent in terms of uniform coverage of parameter values in the identified set, all of which are valid candidates for $\theta_{0}(P)$.

The previously mentioned results do not discuss the statistical power of inference for the four CIs. That is, they are silent about the ability of these CIs to rule out parameters outside of the identified set, which are not valid candidates for $\theta_{0}(P)$. In fact, to our knowledge, the literature has not compared these CIs in terms of power in the context of OBS. Our first contribution is to conduct this comparison. To this end, we study the limiting coverage probabilities of the four CIs for all possible parameter sequences outside the identified set. A higher limiting coverage rate for these parameters corresponds to lower power. We formally show that $CI_{\alpha }^{1}$ and $CI_{\alpha }^{2}$ are equally powerful, and both dominate $CI_{\alpha }^{3}$ and $CI_{\alpha }^{4}$. This is a favorable result from a practical standpoint, as $CI_{\alpha}^{1}$ is straightforward to implement, and $CI_{\alpha}^{1}$ and $CI_{\alpha}^{2}$ are free from tuning parameter choices, unlike $CI_{\alpha}^{3}$.

For our second result, we consider the advantageous situation in which the researcher has not one but two pairs of estimators that satisfy OBS, and one of them is known to be more efficient than the other, in the sense of having smaller diagonal elements of their asymptotic covariance matrices.\footnote{This is weaker than the standard definition of efficiency, where the difference of the asymptotic covariance matrices of the efficient and inefficient estimators, respectively, is negative semidefinite.} While one could implement asymptotically exact CIs for $\theta _{0}( P) $ using either one of these pairs of estimators, it is reasonable to expect that inference based on the more efficient pair is preferable. In particular, one would expect that the more efficient estimator would result in more powerful inference. We formally demonstrate that this result generally holds for $CI_{\alpha }^{1}$ and $CI_{\alpha }^{2}$, but not necessarily for $CI_{\alpha }^{3}$ or $CI_{\alpha }^{4}$. Specifically, for $CI_{\alpha }^{1}$ and $CI_{\alpha }^{2}$, inference based on the more efficient bounds estimators is always at least as powerful as inference based on the less efficient estimators. In contrast, for $CI_{\alpha }^{3}$ and $CI_{\alpha }^{4}$, the reverse can occur. That is, inference based on the more efficient estimators can be strictly less powerful than inference based on the less efficient estimators. We explain this counterintuitive and undesirable phenomenon and provide the conditions under which it arises.

One can summarize both results in our paper by stating that, within our OBS framework, $CI_{\alpha}^{1}$ and $CI_{\alpha}^{2}$ have identical power properties, and both compare favorably to $CI_{\alpha}^{3}$ and $CI_{\alpha}^{4}$.

Our econometric framework with OBS is motivated by a general class of econometric problems involving missing data; see \cite{manski:1989,manski:1990,manski:1994,manski:1995}. We illustrate our setup with the standard missing data problem in Example \ref{ex:OBS_Ex}. Beyond this canonical example, our analysis is also motivated by the treatment effects problem studied in \cite{bugni/gao/obradovic/velez:2024a}. 
In that paper, we investigate inference for treatment effect parameters such as the average treatment effect (ATE) in a randomized controlled trial (RCT) with imperfect compliance. In this context, the ATE is partially identified, with its identified set being an interval. We also propose consistent and asymptotically normal estimators of these bounds that satisfy the OBS. Accordingly, we can implement asymptotically valid and exact inference using either $CI_{\alpha}^{1}$, $CI_{\alpha}^{2}$, $CI_{\alpha}^{3}$, or $CI_{\alpha}^{4}$. Our first contribution demonstrates that $CI_{\alpha}^{1}$ and $CI_{\alpha}^{2}$ are equally powerful, and both are more powerful than $CI_{\alpha}^{3}$ and $CI_{\alpha}^{4}$. Based on this finding, we choose to conduct inference using either $CI_{\alpha}^{1}$ or $CI_{\alpha}^{2}$, instead of $CI_{\alpha}^{3}$ or $CI_{\alpha}^{4}$.

Within the empirical application in \cite{bugni/gao/obradovic/velez:2024a}, we have two possible implementations of our bounds estimators for the ATE: we can estimate treatment probabilities using sample frequencies or exact probabilities (known in an RCT). By similar arguments to those in \cite{hahn:1998} and \cite{hirano/imbens/ridder:2003}, the bounds estimators that use sample analogs are shown to be more efficient than those using exact probabilities. Intuitively, one would expect that the more efficient implementation of the bounds estimator produces a more powerful inference of the partially-identified parameter value. Our second contribution shows that this intuitive result holds for $CI_{\alpha}^{1}$ and $CI_{\alpha}^{2}$, but may fail to hold for $CI_{\alpha}^{3}$ or $CI_{\alpha}^{4}$. That is, it is possible for the more efficient bounds estimator to result in less powerful inference when $CI_{\alpha}^{3}$ or $CI_{\alpha}^{4}$ is used. Based on the two contributions, we suggest using $CI_{\alpha}^{1}$ or $CI_{\alpha}^{2}$ under OBS. 

The econometric model described by our identified set $[\theta_l(P), \theta_u(P)]$ is a partially-identified model that is both simple and empirically relevant. Within the OBS framework, the methods proposed by \cite{imbens/manski:2004} and \cite{stoye:2009} are particularly relevant to this model. Nonetheless, there is a substantial literature on inference in partially-identified models, with a significant focus on moment (in)equality models; see \citet{tamer:2010, canay/shaikh:2017, ho/rosen:2017, molinari:2020,canay/illanes/velez:2026} for comprehensive reviews. If the endpoints of the identified set, $\theta_l(P)$ and $\theta_u(P)$, are specified as functions of expectations of observed random variables (e.g., linear functions or ratios of expectations), our model could be viewed as a special case within the class of moment (in)equality models. However, this is not a necessary assumption in our framework, as we do not require $\theta_l(P)$ and $\theta_u(P)$ to be linked to expectations. Thus, our framework is not inherently part of the moment (in)equality literature.

The rest of the paper is organized as follows. Section \ref{sec:setup} describes the econometric model. Section \ref{sec:OBS} presents our setup (i.e., OBS), and Section \ref{sec:CI} details the four CIs. We present our main results in Section \ref{sec:Main}, which is divided into two subsections. Section \ref{sec:Power1} compares the power of inference across the CIs. Section \ref{sec:Power2} compares the power of inference for each CI when two bounds estimators are available, with one being more efficient than the other. 
Section \ref{sec:conclusions} concludes. The appendix contains most proofs and intermediate results. Proofs related to $CI_{\alpha }^{4}$ are provided in an online supplement.

\section{Setup}\label{sec:setup}

For a data distribution $P$, the real-valued parameter of interest is denoted by $\theta _{0}(P)$. The econometric model indicates that $\theta _{0}(P)$ belongs to an interval identified set $\Theta _{I}(P) = [\theta _{l}(P), \theta _{u}(P)]$. Moreover, we assume that the identified set is non-empty, i.e., $\theta _{l}(P) \leq \theta _{u}(P)$.

We consider a confidence interval $CI_{\alpha }$ that covers $\theta _{0}(P)$ with a minimum prespecified coverage probability of $(1-\alpha)$ as the sample size $N$ increases. Furthermore, we require this coverage condition to be satisfied uniformly for all parameters in the identified set $\Theta _{I}(P)$ and for all probability distributions $P$ in a suitable space $\mathcal{P}$. Specifically, we require our CI to be {\it uniformly asymptotically valid} and, if possible, {\it uniformly asymptotically exact}, which we define next.

The confidence interval $CI_{\alpha }$ for $\theta _{0}(P) \in \Theta _{I}(P)$ is {\it uniformly asymptotically valid} if it satisfies the following property:
\begin{equation}
\underset{N \to \infty }{\lim\inf}~\inf_{P\in \mathcal{P}}~\inf_{\theta \in \Theta _{I}(P) }~P\left( \theta \in CI_{\alpha }\right) ~~\geq~~ 1-\alpha.
\label{eq:covGoal}
\end{equation}
Moreover, $CI_{\alpha }$ is {\it uniformly asymptotically exact} if \eqref{eq:covGoal} holds with equality, i.e.,
\begin{equation}
\underset{N \to \infty }{\lim\inf}~\inf_{P\in \mathcal{P}}~\inf_{\theta \in \Theta _{I}(P) }~P\left( \theta \in CI_{\alpha }\right) ~~=~~ 1-\alpha.
\label{eq:covGoal_exact}
\end{equation}
By definition, a CI that is uniformly asymptotically exact is also uniformly asymptotically valid.

The remainder of this section is organized as follows. Section \ref{sec:OBS} specifies our main assumptions, referred to as OBS. Section \ref{sec:CI} describes the four CIs considered in this paper. Under OBS, these four CIs are uniformly asymptotically exact.

\subsection{Ordered bounds setup (OBS)}\label{sec:OBS}

Throughout our paper, we assume that the researcher can construct estimators of the bounds of the identified set and its limiting distribution that satisfy the following condition.

\begin{definition}[{\it OBS} or Ordered Bounds Setup]\label{def:setup}
We say that the estimator $(\hat{\theta}_{l},\hat{\theta} _{u},\hat{\sigma}_{l},\hat{\sigma}_{u},\hat{\rho}) \in \mathbb{R} \times \mathbb{R} \times \mathbb{R}_{+} \times \mathbb{R}_{+} \times [-1,1]$ satisfies OBS with parameters $({\theta }_{l}(P),{\theta }_{u}(P),{\sigma }_{l}(P),{\sigma } _{u}(P),{\rho }(P))$ and set of distributions $\mathcal{P}\equiv \mathcal{P}(\underline{\sigma }^{2},\bar{\sigma}^{2},\overline{\Delta})$ if the following conditions are satisfied:
\begin{enumerate}[(a)]
\item $( \hat{\theta}_{l},\hat{\theta}_{u})$ is ``ordered'', i.e., $P(\hat{ \theta}_{l}\leq \hat{\theta}_{u}) = 1$.

\item $(\hat{\theta}_{l},\hat{\theta}_{u})$ is uniformly asymptotically normal, i.e.,
\begin{equation*}
    \sqrt{N}\left( 
\begin{array}{c}
\hat{\theta}_{l}-\theta _{l}(P) \\ 
\hat{\theta}_{u}-\theta _{u}(P)
\end{array}
\right) ~\overset{d}{ \to }~\mathcal{N}\left( \mathbf{0}_{2\times 1},\left( 
\begin{tabular}{cc}
$\sigma _{l}(P)^2$ & $\rho (P)\sigma _{l}(P)\sigma _{u}(P)$ \\ 
$\rho (P)\sigma _{l}(P)\sigma _{u}(P)$ & $\sigma _{u}(P)^2$
\end{tabular}
\right) \right)
\end{equation*}
uniformly in $P\in\mathcal{P}$. Moreover, assume that $\sigma _{l}(P)^2,\sigma _{u}(P)^2 \in \left[ \underline{\sigma }^{2},\bar{\sigma}^{2}\right] $ with $0<\underline{\sigma }^{2}\leq \bar{\sigma}^{2}<\infty $ and $\theta _{u}(P)-\theta _{l}(P)\leq \overline{\Delta}<\infty$.

\item $(\hat{\sigma}_{l},\hat{\sigma}_{u},\hat{\rho})$ is uniformly consistent for $({\sigma}_{l}(P),{\sigma}_{u}(P),{\rho}(P))$, i.e.,
\begin{equation*}
    \left( \hat{\sigma}_{l},\hat{\sigma}_{u},\hat{\rho}\right) ~\overset{p}{ \to }~\left( \sigma _{l}(P),\sigma _{u}(P),\rho (P)\right) 
\end{equation*}
uniformly in $P\in\mathcal{P}$.
\end{enumerate}
\end{definition}

The set of distributions $\mathcal{P}$ in Definition \ref{def:setup} encodes high-level assumptions that facilitate our asymptotic analysis. Condition (a) occurs frequently in economic applications where bound estimators are constructed as sample analogs of ordered population bounds, so that the ordering is preserved by construction. It arises, for example, in models with treatment effects or missing data that are not assumed to be missing at random; see, e.g., \cite{manski:1990, horowitz/manski:2000}. Importantly, condition (a) is closely related to existing approaches in the literature. In particular, \citet[Lemma 3]{stoye:2009} shows that, under conditions (b)-(c), condition (a) implies the superefficiency condition stated in \citet[Assumption 3]{stoye:2009}.\footnote{The superefficiency condition requires the existence of a sequence $\{a_N\}_{N \in \mathbb{N}}$ such that $a_N \to 0$, $a_N \sqrt{N} \to \infty$, and $\sqrt{N}| (\hat{\theta}_u - \hat{\theta}_l) - (\theta_{u}(P_N) - \theta_l(P_N)) | \overset{p}{\to} 0$ for any sequence $\{P_N \in \mathcal{P}\}_{N \in \mathbb{N}}$ with $\theta_{u}(P_N) - \theta_l(P_N) \leq a_N$.} As discussed in \cite{stoye:2009}, this condition is sufficient for the inference results in \cite{imbens/manski:2004}, who impose a stronger assumption.

Conditions (b)-(c) are high-level requirements on the asymptotic distribution of the bound estimators and on the estimation of the parameters of the corresponding limiting distribution. They coincide exactly with Assumptions 1(i)-(ii) in \cite{imbens/manski:2004} and \cite{stoye:2009}. In applications, one would typically impose low-level conditions (e.g., i.i.d.\ sampling and bounded moments) to establish conditions (b)-(c) in Definition \ref{def:setup} for a suitable estimator $(\hat{\theta}_l,\hat{\theta}_u,\hat{\sigma}_l,\hat{\sigma}_u,\hat{\rho})$, typically constructed as sample analogs of the corresponding population quantities. We leave $\mathcal{P}$ unspecified to maintain generality.

The previous discussion implies that the OBS framework satisfies all of the assumptions required by \cite{imbens/manski:2004} and \cite{stoye:2009}. Thus, an estimator $(\hat{\theta}_{l}, \hat{\theta}_{u}, \hat{\sigma}_{l}, \hat{\sigma}_{u}, \hat{\rho})$ that satisfies OBS can be used to implement any of the CIs proposed in these papers to achieve uniformly asymptotically exact inference. These CIs are reviewed in the next section.

The following example illustrates OBS in the canonical missing-data model, inspired by \cite{manski:1989}.

\begin{example} 
\label{ex:OBS_Ex}
Consider the following missing-data problem. Let $\{(Y_{i},Z_{i})\}_{i=1}^{N}$ be an i.i.d.\ sample from a distribution $P\in \mathcal{P}$, with $Y_{i}\in [\underline{Y},\overline{Y}]$, where $\underline{Y}$ and $\overline{Y}$ are known constants with $\underline{Y}\leq \overline{Y}$, and $ Z_{i}\in \{0,1\}$ is a binary variable that indicates whether $Y_{i}$ is observed ($Z_{i}=1$) or not ($Z_{i}=0$). Further, assume that $E_{P}[ Y_{i}|Z_{i}=1] =\mu (P)$, $V_{P}[ Y_{i}|Z_{i}=1] =\sigma ^{2}(P) $ with $\sigma ^{2}(P)\geq \delta_1 >0$. Finally, assume that $P( Z_{i}=0) =\pi (P) \in [0,\delta_2]$ with $\delta_2<1$.

The parameter of interest is $\theta =E_{P}[Y_{i}]$. By elementary arguments:
\begin{equation*}
\theta ~=~E_{P}[Y_{i}]~\in ~[E_{P}[Y_{i}Z_{i}+\underline{Y}(1-Z_{i})],~E_{P}[Y_{i}Z_{i}+\overline{Y}(1-Z_{i})]]. 
\end{equation*}
Then, the sharp identified set for $\theta $ is
\begin{equation*}
\Theta _{I}(P)~=~[\theta _{l}(P),~\theta _{u}(P)], 
\end{equation*}
where $\theta _{l}(P)=E_{P}[Y_{i}Z_{i}+\underline{Y}(1-Z_{i})]$ and $\theta _{u}(P)=E_{P}[Y_{i}Z_{i}+\overline{Y}(1-Z_{i})]$. Also, 
\begin{align*}
\theta _{l}(P)& ~=~E_{P}[Y_{i}Z_{i}+\underline{Y}(1-Z_{i})]=\mu (P)( 1-\pi (P)) +\underline{Y}\pi (P), \\
\theta _{u}(P)& ~=~E_{P}[Y_{i}Z_{i}+\overline{Y}(1-Z_{i})]=\mu (P)( 1-\pi (P)) +\overline{Y}\pi (P).
\end{align*}

In this context, it is natural to propose the following sample analog estimators for the bounds: 
\begin{equation}
(\hat{\theta}_{l},\hat{\theta}_{u})~=~\left( \frac{1}{N} \sum_{i=1}^{N}Y_{i}Z_{i}+\underline{Y}(1-Z_{i}),\frac{1}{N} \sum_{i=1}^{N}Y_{i}Z_{i}+\overline{Y}(1-Z_{i})\right). 
\label{eq:estimator_bounds_example}
\end{equation}
Note that $\underline{Y}\leq \overline{Y}$ implies $\hat{\theta}_{l}\leq \hat{\theta}_{u}$, i.e., Definition \ref{def:setup}(a) holds. Since the observations are bounded, the triangular array central limit theorem implies 
\begin{align*}
 \sqrt{N}(\hat{\theta}_{l}-\theta _{l}(P),\hat{\theta}_{u}-\theta _{u}(P)) 
 ~\overset{d}{ \to }~\mathcal{N}( \mathbf{0}_{2\times 1},\Sigma (P) ) ,
\end{align*}
uniformly in $P$, where $\Sigma (P)$ is the following matrix 
\begin{equation*}
(1-\pi (P))\left[ 
\begin{array}{cc}
\sigma ^{2}(P)+\pi (P)(\mu (P)-\underline{Y})^{2} & 
\sigma ^{2}(P)+ \pi (P)(\mu (P)-\underline{Y})(\mu (P) -\overline{Y}) \\ 
\sigma ^{2}(P)+\pi (P)(\mu (P)-\underline{Y})(\mu (P) -\overline{Y}) & \sigma ^{2}(P)+\pi (P)(\mu (P)-\overline{Y})^{2}
\end{array}
\right]. 
\end{equation*}
From here, we deduce that Definition \ref{def:setup}(b) holds with $\underline{\sigma } ^{2}=(1-\delta _{2})\delta _{1}>0$, $\bar{\sigma}^{2}=5(\overline{Y} - \underline{Y})^2/4<\infty$, and $ \bar{\Delta}=\overline{Y} - \underline{Y}<\infty$. Finally, standard asymptotic arguments imply that Definition \ref{def:setup}(c) holds for sample analogs, i.e., $\hat{\sigma}_{l}$ and $\hat{\sigma}_{u}$ based on
\begin{align*}
\hat{\sigma}_{l}^{2}& ~=~\frac{1}{N}\sum_{i=1}^{N}(Y_{i}Z_{i}+\underline{Y}(1-Z_{i}))^{2}-\Big(\frac{1}{N}\sum_{i=1}^{N}(Y_{i}Z_{i}+\underline{Y}(1-Z_{i}))\Big)^{2} \\
\hat{\sigma}_{u}^{2}& ~=~\frac{1}{N}\sum_{i=1}^{N}(Y_{i}Z_{i}+\overline{Y}(1-Z_{i}))^{2}-\Big(\frac{1}{N}\sum_{i=1}^{N}(Y_{i}Z_{i}+\overline{Y}(1-Z_{i}))\Big)^{2},
\end{align*}
and $\hat{\rho}=\hat{\sigma}_{lu}/(\hat{\sigma}_{l}\hat{\sigma}_{u})$ with
\begin{equation*}
\hat{\sigma}_{lu}~=~\frac{1}{N}\sum_{i=1}^{N}\left( 
\begin{array}{c}
(Y_{i}Z_{i}+\underline{Y}(1-Z_{i}))\times \\ 
(Y_{i}Z_{i}+\overline{Y}(1-Z_{i}))
\end{array}
\right) -\left( 
\begin{array}{c}
\frac{1}{N}\sum_{i=1}^{N}(Y_{i}Z_{i}+\underline{Y}(1-Z_{i}))\times \\ 
\frac{1}{N}\sum_{i=1}^{N}(Y_{i}Z_{i}+\overline{Y}(1-Z_{i}))
\end{array}
\right) . 
\end{equation*}
\end{example}

We conclude this section by noting that there are well-known examples of interval-identified sets that do not satisfy OBS. In particular, the ``ordered'' condition in Definition \ref{def:setup}(a) can be restrictive in certain cases. To illustrate this, Example \ref{ex:NoOBS_Ex} in Appendix \ref{sec:OBSfails} presents a missing-data problem in a linear regression setting where the bounds estimators are not ``ordered'', causing OBS to fail. It is important to note that our results depend critically on OBS, and in particular on the requirement that the bounds be ``ordered''; without it, the results of our paper do not apply.

\subsection{Confidence intervals (CIs)}\label{sec:CI}

This paper considers four CIs. Our first confidence interval is $CI_{\alpha }^{1}$, originally proposed by \citet[Section 4]{imbens/manski:2004}, and revisited by \cite{stoye:2009}. Given an estimator $(\hat{\theta}_{l},\hat{\theta}_{u},\hat{\sigma}_{l},\hat{\sigma}_{u},\hat{\rho})$ and a confidence level $(1-\alpha)$, $CI_{\alpha }^{1}$ is defined as follows: 
\begin{equation}
CI_{\alpha }^{1}~=~\left[ ~\hat{\theta}_{l}-\frac{\hat{\sigma}_{l}c^{1}}{\sqrt{N}},~\hat{\theta}_{u}+\frac{\hat{\sigma}_{u}c^{1}}{\sqrt{N}}~\right] ,
\label{eq:CI1}
\end{equation}
where $c^{1}$ solves 
\begin{equation}
\Phi \left( c^{1}+\frac{\sqrt{N}(\hat{\theta}_{u}-\hat{\theta}_{l})}{\max \{\hat{\sigma}_{l},\hat{\sigma}_{u}\}}\right) ~-~\Phi \left( -c^{1}\right) ~=~1-\alpha .  \label{eq:CI1_problem}
\end{equation}
Provided that $\max \{\hat{\sigma}_{l},\hat{\sigma}_{u}\}>0$, it follows that $c^{1}$ is uniquely determined by \eqref{eq:CI1_problem}.\footnote{Under OBS, $\max \{\hat{\sigma}_{l},\hat{\sigma}_{u}\}>0$ occurs with probability approaching one.} Under OBS, \citet[Lemma 4]{imbens/manski:2004} implies that $CI_{\alpha }^{1}$ is uniformly asymptotically valid, while \citet[Proposition 1]{stoye:2009} shows that $CI_{\alpha }^{1}$ is uniformly asymptotically exact. Importantly, the implementation of $CI_{\alpha}^{1}$ does not require estimating the correlation coefficient $\hat{\rho}$, whereas the other methods do.

Our second confidence interval is $CI_{\alpha }^{2}$, developed by \cite{stoye:2009}. Given a generic estimator $(\hat{\theta}_{l},\hat{\theta}_{u},\hat{\sigma}_{l},\hat{\sigma}_{u},\hat{\rho})$ and a confidence level $(1-\alpha)$, $CI_{\alpha }^{2}$ is defined as follows: 
\begin{equation}
CI_{\alpha }^{2}~=~\left[ ~\hat{\theta}_{l}-\frac{\hat{\sigma}_{l}c_{l}^{2}}{\sqrt{N}},~\hat{\theta}_{u}+\frac{\hat{\sigma}_{u}c_{u}^{2}}{\sqrt{N}}~\right] ,\label{eq:CI2}
\end{equation}
where $(c_{l}^{2},c_{u}^{2})$ solve:
\begin{align}
& \min_{c_{l},c_{u}\in \mathbb{R}}~(\hat{\sigma}_{l}c_{l}+\hat{\sigma}_{u}c_{u})~\text{s.t.}~  \notag \\
& P\Big(~\Big\{-c_{l}\leq z_{1}\Big\}~\cap ~\Big\{\hat{\rho}z_{1}\leq c_{u}+\tfrac{\sqrt{N}(\hat{\theta}_{u}-\hat{\theta}_{l})}{\hat{\sigma}_{u}}+\sqrt{1-\hat{\rho}^{2}}z_{2}\Big\}~\Big|~\hat{\theta}_{l},\hat{\theta}_{u},\hat{\sigma}_{l},\hat{\sigma}_{u},\hat{\rho}~\Big)~\geq ~1-\alpha ~~\text{and}~~ \notag \\
& P\Big(~\Big\{-c_{l}-\tfrac{\sqrt{N}(\hat{\theta}_{u}-\hat{\theta}_{l})}{\hat{\sigma}_{l}}+\sqrt{1-\hat{\rho}^{2}}z_{2}\leq \hat{\rho}z_{1}\Big\}~\cap ~\Big\{z_{1}\leq c_{u}\Big\}~\Big|~\hat{\theta}_{l},\hat{\theta}_{u},\hat{\sigma}_{l},\hat{\sigma}_{u},\hat{\rho}~\Big)~\geq ~1-\alpha,
\label{eq:CI2_problem}
\end{align}
where $(z_{1},z_{2})\sim \mathcal{N}(\mathbf{0}_{2\times 1},\mathbf{I}_{2\times 2})$. It is unclear to us whether \eqref{eq:CI2_problem} has a unique solution.\footnote{\citet[Page 1305]{stoye:2009} states that typically $(c_{l}^{2},c_{u}^{2})$ is uniquely determined by the fact that both constraints in \eqref{eq:CI2_problem} hold with equality.} Be that as it may, our formal arguments will only require the researcher to choose $ (c_{l}^{2},c_{u}^{2})$ arbitrarily whenever \eqref{eq:CI2_problem} has multiple solutions.

As explained in \citet[Page 1305]{stoye:2009}, $CI_{\alpha }^{2}$ calibrates the critical values $(c_{l}^{2},c_{u}^{2})$ taking into account that the underlying problem is bivariate. In this sense, $CI_{\alpha }^{2}$ is viewed as an improvement upon $CI_{\alpha }^{1}$. In fact, \citet[Page 1305]{stoye:2009} argues that $CI_{\alpha }^{2}$ is the shortest CI with correct nominal size. \citet[Proposition 2]{stoye:2009} shows that $CI_{\alpha }^{2}$ is uniformly asymptotically exact under OBS.

Our third confidence interval is $CI_{\alpha }^{3}$, also developed by \cite{stoye:2009}. Unlike $ CI_{\alpha }^{1}$ and $CI_{\alpha }^{2}$, $CI_{\alpha }^{3}$ was introduced as a CI that does not require the superefficiency condition (\citet[Assumption 3]{stoye:2009}) for its validity.  
As already explained, the superefficiency condition is guaranteed under our OBS assumptions, but may fail in other contexts. To implement $CI_{\alpha }^{3}$, the researcher must specify a preassigned tuning parameter sequence of constants $\left\{ b_{N}\right\} _{N\geq 1}$ that satisfies $ b_{N}  \to 0$ and $b_{N}\sqrt{N}  \to \infty $. These types of sequences are commonly utilized for moment selection in the moment (in)equality literature; e.g., see \cite{andrews/soares:2010,bugni:2010,bugni:2015}. For example, these papers suggest sequences such as $b_{N}=\ln N/\sqrt{N}$, $b_{N}=\sqrt{ \ln \ln N}/\sqrt{N}$, or $b_{N}=N^{-c}$ for any $c\in (0,1/2)$.

Given an estimator $(\hat{\theta}_{l},\hat{\theta}_{u},\hat{\sigma}_{l}, \hat{\sigma}_{u},\hat{\rho})$, a confidence level $(1-\alpha)$, and a sequence $\left\{ b_{N}\right\} _{N\geq 1}$, $CI_{\alpha }^{3}$ is defined as follows: 
\begin{equation}
CI_{\alpha }^{3}~=~\left\{ 
\begin{tabular}{ll}
$\Big[ ~\hat{\theta}_{l}-\dfrac{\hat{\sigma}_{l}c_{l}^{3}}{\sqrt{N}},~\hat{\theta}_{u}+\dfrac{\hat{\sigma}_{u}c_{u}^{3}}{\sqrt{N}}~\Big] $ & if $\hat{\theta}_{l}-\dfrac{\hat{\sigma}_{l}c_{l}^{3}}{\sqrt{N}}\leq \hat{\theta}_{u}+\dfrac{\hat{\sigma}_{u}c_{u}^{3}}{\sqrt{N}}$ \\ 
$\emptyset $ & otherwise,
\end{tabular}
\right.   \label{eq:CI3}
\end{equation}
where $(c_{l}^{3},c_{u}^{3})$ solve
\begin{align}
& \min_{c_{l},c_{u}\in \mathbb{R}}~(\hat{\sigma}_{l}c_{l}+\hat{\sigma}_{u}c_{u})~\text{s.t.}~  \notag \\
& \left. P\left( 
\begin{array}{c}
\Big\{\hat{\rho}z_{1}\leq c_{u}+\tfrac{\sqrt{N}(\hat{\theta}_{u}-\hat{\theta}_{l})1[(\hat{\theta}_{u}-\hat{\theta}_{l})>b_{N}]}{\hat{\sigma}_{u}}+\sqrt{1-\hat{\rho}^{2}}z_{2}\Big\} \\ 
~\cap ~\big\{-c_{l}\leq z_{1}\big\}
\end{array}
\right\vert \hat{\theta}_{l},\hat{\theta}_{u},\hat{\sigma}_{l},\hat{\sigma}_{u},\hat{\rho}\right) \geq 1-\alpha ~\text{and}  \notag \\
& \left. P\left( 
\begin{array}{c}
\Big\{-c_{l}-\tfrac{\sqrt{N}(\hat{\theta}_{u}-\hat{\theta}_{l})1[(\hat{\theta}_{u}-\hat{\theta}_{l})>b_{N}]}{\hat{\sigma}_{l}}+\sqrt{1-\hat{\rho}^{2}}z_{2}\leq \hat{\rho}z_{1}\Big\} \\ 
~\cap ~\big\{z_{1}\leq c_{u}\big\}
\end{array}
\right\vert \hat{\theta}_{l},\hat{\theta}_{u},\hat{\sigma}_{l},\hat{\sigma} _{u},\hat{\rho}\right) \geq 1-\alpha ,  \label{eq:CI3_problem}
\end{align}
where $(z_{1},z_{2})\sim \mathcal{N}(\mathbf{0}_{2\times 1},\mathbf{I}_{2\times 2})$. As with $CI_\alpha^2$, it is unclear whether \eqref{eq:CI3_problem} is guaranteed to have a unique solution. In any case, our arguments only require the researcher to choose $(c_{l}^{3},c_{u}^{3})$ arbitrarily whenever \eqref{eq:CI3_problem} has multiple solutions. \citet[Proposition 3]{stoye:2009} shows that $CI_{\alpha }^{3}$ is uniformly asymptotically exact under OBS.

\begin{remark}\label{rem:robustness}
As already noted, a central advantage of $CI_{\alpha }^{3}$ is that it does not require the superefficiency condition for validity. Since the OBS framework satisfies the superefficiency condition, one might reasonably conclude that $CI_{\alpha }^{3}$ is not needed in this context. Moreover, $CI_{\alpha }^{3}$ is as computationally costly as $CI_{\alpha }^{2}$ and requires an additional tuning parameter sequence. Our forthcoming power results further clarify the costs associated with the robustness properties of $CI_{\alpha }^{3}$.
\end{remark}

Our fourth confidence interval was introduced recently by \cite{stoye:2020}. Given an estimator $(\hat{\theta}_{l},\hat{\theta}_{u},\hat{\sigma}_{l}, \hat{\sigma}_{u},\hat{\rho})$ and a confidence level $(1-\alpha)$, $CI_{\alpha }^{4}$ is defined as follows: 
\begin{equation}
CI_{\alpha }^{4}~=~\Big[ \hat{\theta}_{l}-\frac{\hat{\sigma}_{l}c^4 }{\sqrt{N}},\hat{\theta}_{u}+\frac{\hat{\sigma}_{u}c^4 }{\sqrt{N}}\Big] \cup \Big[ \hat{\theta}^{\ast }-\frac{\hat{\sigma}^{\ast }\Phi^{-1} ( 1-\alpha /2) }{\sqrt{N}},\hat{\theta}^{\ast }+\frac{\hat{\sigma}^{\ast }\Phi^{-1} ( 1-\alpha /2) }{\sqrt{N}}\Big],
\label{eq:CI4}
\end{equation}
where $c^4$ is the unique value of $c$ that solves
\begin{equation*}
\inf_{\Delta \geq 0}P\Bigg( 
\begin{array}{c}
\{ z_{1}( \hat \rho ) -\Delta -c\leq 0\leq z_{2}( \hat \rho ) +c\}~~ \cup  \\ 
\{ \vert z_{1}( \hat \rho ) +z_{2}( \hat \rho ) -\Delta \vert \leq \sqrt{2+2{\hat \rho}}\Phi ^{-1}( 1-\alpha /2) \} 
\end{array}
\Bigg|\hat \rho \Bigg) ~=~1-\alpha ,
\end{equation*}
where $z( \hat\rho ) =( z_{1}(\hat\rho ) ,z_{2}( \hat\rho ) ) \sim \mathcal{N}( \mathbf{0}_{2\times 1},[ 1,\hat\rho ;\hat\rho ,1] ) $,
\begin{align*}
\hat{\theta}^{\ast } ~=~\frac{\hat{\theta}_{l}\hat{\sigma}_{u}+\hat{\theta}_{u}\hat{\sigma}_{l}}{\hat{\sigma}_{l}+\hat{\sigma}_{u}}~~~~\text{and}~~~~
\hat{\sigma}^{\ast } ~=~\frac{\hat{\sigma}_{l}\hat{\sigma}_{u}\sqrt{2+2\hat{\rho}}}{\hat{\sigma}_{l}+\hat{\sigma}_{u}}.
\end{align*} 
\citet[Theorem 1]{stoye:2020} proves that $CI_{\alpha}^{4}$ is uniformly asymptotically valid under OBS. \cite{stoye:2020} argues that $CI_{\alpha}^{4}$ has several advantages relative to existing approaches in the literature: it delivers desirable coverage properties even under model misspecification, is never empty or excessively short, requires no tuning parameters, and is computationally trivial to implement.

\begin{remark}
As we explain in the introduction, our econometric framework does not necessarily correspond to a moment (in)equality model. For this reason, the CIs proposed by the moment (in)equality literature are not included among the inference methods we consider.
\end{remark}

\section{Main results}\label{sec:Main}

Our goal is to compare the power of inference based on CIs for the partially-identified parameter $\theta_{0}(P) \in \Theta_{I}(P)$. By the duality between CIs and hypothesis testing, we can investigate the power of an inference method based on a CI by deriving its limiting coverage probability for a parameter value $\theta$ outside of $\Theta_{I}(P)$. With an interval identified set $\Theta_{I}(P) = [\theta_l(P), \theta_u(P)]$, this means that either $\theta < \theta_l(P)$ or $\theta > \theta_u(P)$. Since $\theta$ does not belong to the identified set $\Theta_{I}(P)$, it cannot be the true parameter value $\theta_{0}(P)$. Thus, a lower limiting coverage probability for $\theta$ is equivalent to higher statistical power against the (incorrect) null hypothesis $H_0: \theta_{0}(P) = \theta$.

Based on the previous discussion, we compare the limiting coverage probabilities for parameter values outside $\Theta_{I}(P)$ across various CIs. Importantly, our analysis allows the parameter value and the data distribution to vary arbitrarily with the sample size. That is, we consider all possible sequences $\{(P_{N}, \theta_{N}) \in \mathcal{P} \times \Theta_{I}(P_{N})^{c}\}_{N\in \mathbb{N}}$. Thus, our results include power analysis for fixed alternatives, i.e., $\theta_{N} = \bar{\theta} \not\in \Theta_{I}(P_{N})$, as well as local alternatives, i.e., $\theta_{N} \uparrow \theta_{l}(P_{N})$ or $\theta_{N} \downarrow \theta_{u}(P_{N})$.

\subsection{Power comparison across CIs}\label{sec:Power1}

Our only result in this section is Theorem \ref{thm:CIcomparison}. This result compares the limiting coverage probability of the four CIs for sequences of parameters outside $\Theta_{I}(P)$.

\begin{theorem}[Comparison across CIs]\label{thm:CIcomparison} 
Let $(\hat{\theta}_{l},\hat{\theta}_{u},\hat{\sigma}_{l},\hat{\sigma}_{u},\hat{\rho})$ be an estimator that satisfies OBS  in Definition \ref{def:setup} with parameter $({\theta } _{l}(P),{\theta }_{u}(P),{\sigma }_{l}(P),{\sigma }_{u}(P),{\rho }(P))$ and set $\mathcal{P}$. Assume that $\alpha \in (0,0.5)$. Then,

\begin{enumerate}[(a)]
\item For any sequence $\{(P_{N},\theta _{N})\in \mathcal{P}\times \Theta _{I}(P_{N})^{c}\}_{N\in \mathbb{N}}$, 
\begin{equation}
\lim _{N  \to \infty}~\big(P_{N}(\theta _{N}\in CI_{\alpha}^{1})-P_{N}(\theta _{N}\in CI_{\alpha}^{2})\big)~=~ 0 .
\label{eq:CI4main1}
\end{equation}

\item For any sequence $\{(P_{N},\theta _{N})\in \mathcal{P}\times \Theta _{I}(P_{N})^{c}\}_{N\in \mathbb{N}}$,
\begin{equation}
\underset{N \to \infty}{\lim  \inf }~\big(P_{N}(\theta _{N}\in CI_{\alpha}^{3})-P_{N}(\theta _{N}\in CI_{\alpha}^{j})\big)~\geq ~0~~~\text{for}~j=1,2. 
\label{eq:CI4main2}
\end{equation}
Furthermore, \eqref{eq:CI4main2} holds strictly for suitable sequences of $\{(P_{N},\theta _{N})\in \mathcal{P}\times \Theta _{I}(P_{N})^{c}\}_{N\in \mathbb{N}}$.

\item For any sequence $\{(P_{N},\theta _{N})\in \mathcal{P}\times \Theta _{I}(P_{N})^{c}\}_{N\in \mathbb{N}}$,
\begin{equation}
\underset{N \to \infty}{\lim  \inf }~\big(P_{N}(\theta _{N}\in CI_{\alpha}^{4})-P_{N}(\theta _{N}\in CI_{\alpha}^{j}) \big)~\geq ~0~~~\text{for}~j=1,2,3. \label{eq:CI4main4}
\end{equation}
Furthermore, \eqref{eq:CI4main4} holds strictly for suitable sequences of $\{(P_{N},\theta _{N})\in \mathcal{P}\times \Theta _{I}(P_{N})^{c}\}_{N\in \mathbb{N}}$.
\end{enumerate}
\end{theorem}

Theorem \ref{thm:CIcomparison} consists of three parts. Part (a) states that $CI_{\alpha}^{1}$ and $CI_{\alpha}^{2}$ are equivalent in terms of power: for all sequences of parameters such that $\theta_{N} \not\in \Theta_{I}(P_{N})$, the difference in their coverage rates converges to zero. To explain this result, it is useful to split the analysis into two mutually exclusive cases: ``short'' identified sets and ``long'' identified sets. We say that the identified set is ``short'' if the length of the identified set is $O(1/\sqrt{N})$, and we say that the identified set is ``long'' otherwise. For long identified sets, $CI_{\alpha}^{1}$ and $CI_{\alpha}^{2}$ asymptotically treat the inference problem as one-sided. This implies that both CIs agree on using the $(1-\alpha)$-quantile of the normal distribution as the critical value, leading to identical limiting coverage rates. For short identified sets, the ordered nature of the bounds forces a degenerate asymptotic distribution; otherwise, the estimators would cross with positive probability.\footnote{See Lemma \ref{lem:near1} for a precise statement of this result.} This degeneracy in the asymptotic distribution implies that $CI_{\alpha}^{1}$ and $CI_{\alpha}^{2}$ also agree on the critical values, resulting in the same limiting coverage rates. Combining the two cases, the CIs have identical limiting coverage rates.

Part (b) compares the power of $CI_{\alpha}^{3}$ with the first two CIs for all sequences of parameters $\theta_{N} \not\in \Theta_{I}(P_{N})$. Notably, our conclusions hold regardless of the choice of $\{b_{N}\}_{N \in \mathbb{N}}$ used in the implementation of $CI_{\alpha}^{3}$, provided that $b_{N} \to 0$ and $b_{N}\sqrt{N}  \to \infty$. The result shows that $CI_{\alpha}^{1}$ and $CI_{\alpha}^{2}$ are weakly more powerful than $CI_{\alpha}^{3}$, and that there exist sequences ${(P_N,\theta_N)}_{N \in \mathbb{N}}$ for which this relationship holds strictly. This allows us to conclude that $CI_{\alpha}^{1}$ and $CI_{\alpha}^{2}$ dominate $CI_{\alpha}^{3}$ in terms of power. The argument considers the two cases presented in the previous paragraph. For long identified sets, $CI_{\alpha}^{1}$ and $CI_{\alpha}^{2}$ asymptotically treat the inference problem as one-sided (using the smallest critical values consistent with validity), and are therefore weakly more powerful than $CI_{\alpha}^{3}$. In turn, for short identified sets, $CI_{\alpha}^{3}$ asymptotically treats the inference problem as two-sided, using the $(1-\alpha/2)$-quantile of the normal distribution as the critical value. This critical value is larger than those used by $CI_{\alpha}^{1}$ and $CI_{\alpha}^{2}$, implying that $CI_{\alpha}^{3}$ is weakly less powerful than $CI_{\alpha}^{1}$ and $CI_{\alpha}^{2}$.

Finally, part (c) shows that $CI_{\alpha}^{4}$ is dominated in power by the other three CIs for all sequences of parameters $\theta_{N} \notin \Theta_{I}(P_{N})$. To gain intuition for this result, recall that $CI_{\alpha}^{4}$ is the union of two confidence sets. The first is intended to cover the parameter value under correct specification, while the second is intended to cover the pseudo-true parameter value when the model is misspecified. Our results show that the component of $CI_{\alpha}^{4}$ designed for correct specification is dominated in power by the other CIs.

Theorem \ref{thm:CIcomparison} provides a comprehensive comparison of the relative power properties of the four CIs. It states that, under OBS, $CI_{\alpha}^{1}$ and $CI_{\alpha}^{2}$ are equally powerful, both dominate $CI_{\alpha}^{3}$ and $CI_{\alpha}^{4}$, with $CI_{\alpha}^{3}$ dominating $CI_{\alpha}^{4}$. As already mentioned, we view this as favorable from a practical point of view, since $CI_{\alpha}^{1}$ is straightforward to implement, and neither $CI_{\alpha}^{1}$ nor $CI_{\alpha}^{2}$ requires a tuning parameter sequence. Recall that $CI_{\alpha}^{3}$ was designed to be robust to violations of the superefficiency condition, while $CI_{\alpha}^{4}$ was designed to be robust to model misspecification. In the OBS setting, however, the model is correctly specified, and superefficiency holds automatically. Hence, using $CI_{\alpha}^{3}$ or $CI_{\alpha}^{4}$ incurs a loss of power without providing any compensating robustness benefit.

The proof of Theorem \ref{thm:CIcomparison} is based on auxiliary results that characterize the limiting coverage rates for suitable sequences $\{(P_N, \theta_{N})\}_{N \in \mathbb{N}}$ with $P_N \in \mathcal{P}$ and $\theta_{N} \not\in \Theta_{I}(P_{N})$ for all $N \in \mathbb{N}$. For $CI_{\alpha}^{1}$, $CI_{\alpha}^{2}$, and $CI_{\alpha}^{3}$, these results are presented in Lemmas \ref{lem:C1_limit}, \ref{lem:C2_limit}, and \ref{lem:C3_limit} in the appendix, respectively. The results pertaining to $CI_{\alpha}^{4}$ are relegated to the supplementary appendix. We believe these auxiliary results may be of independent interest beyond this paper.

\subsection{Power comparison across two bounds estimators}\label{sec:Power2}

This section considers a situation in which the researcher has two bounds estimators for constructing CIs for $\theta_{0}(P) \in \Theta_{I}(P)$, with one more efficient than the other. We use the superscripts $E$ and $I$ for the efficient and inefficient estimators, respectively. The next assumption formalizes the setup.

\begin{assumption}\label{ass:1} 
Assume the following conditions:
\begin{enumerate}[(a)]
\item $(\hat{\theta}_{l}^{E},\hat{\theta}_{u}^{E},\hat{\sigma}_{l}^{E},\hat{\sigma}_{u}^{E},\hat{\rho}^{E})$ satisfies OBS with parameters $({\theta }_{l}(P),{\theta }_{u}(P),{\sigma }_{l}^{E}(P),{\sigma }_{u}^{E}(P),{\rho }^{E}(P))$ and set $\mathcal{P}$. 

\item $(\hat{\theta}_{l}^{I},\hat{\theta}_{u}^{I},\hat{\sigma}_{l}^{I},\hat{\sigma}_{u}^{I},\hat{\rho}^{I})$ satisfies OBS with parameters $({\theta }_{l}(P),{\theta }_{u}(P),{\sigma }_{l}^{I}(P),{\sigma }_{u}^{I}(P),{\rho }^{I}(P))$ and set $\mathcal{P}$.

\item $(\hat{\theta}_{l}^{E},\hat{\theta}_{u}^{E})$ is more efficient than $(\hat{\theta}_{l}^{I},\hat{\theta}_{u}^{I})$ in the sense that, for all $P\in \mathcal{P}$, 
\begin{equation}
    \sigma _{l}^{E}(P)\leq \sigma _{l}^{I}(P)~~~\text{and}~~~\sigma _{u}^{E}(P)\leq\sigma _{u}^{I}(P).
    \label{eq:sigmaComparison}
\end{equation}
\end{enumerate}
\end{assumption}

Assumptions \ref{ass:1}(a)-(b) imply that both $(\hat{\theta}_{l}^{E}, \hat{\theta}_{u}^{E})$ and $(\hat{\theta}_{l}^{I}, \hat{\theta}_{u}^{I})$ are asymptotically normal estimators of the bounds of the identified set (i.e., $(\theta_{l}(P), \theta_{u}(P))$), and that we consistently estimate their limiting covariance matrix. Moreover, these results hold uniformly for all distributions $P \in \mathcal{P}$. This means that either estimator can be used to construct CIs that are uniformly asymptotically exact. Assumption \ref{ass:1}(c) specifies the sense in which $(\hat{\theta}_{l}^{E}, \hat{\theta}_{u}^{E})$ is more efficient than $(\hat{\theta}_{l}^{I}, \hat{\theta}_{u}^{I})$: for both the lower and upper bounds, the asymptotic variance of each efficient estimator is less than or equal to that of the inefficient estimator.  We note that this condition weakens the usual relative asymptotic efficiency criterion, which states that, for all $P\in \mathcal{P}$,
{\small\begin{equation*}
    \left( 
\begin{tabular}{cc}
$\sigma _{l}^{E}(P)^2$ & $\rho ^{E}(P)\sigma _{l}^{E}(P)\sigma _{u}^{E}(P)$ \\ 
$\rho ^{E}(P)\sigma _{l}^{E}(P)\sigma _{u}^{E}(P)$ & $\sigma _{u}^{E}(P)^2$
\end{tabular}
\right) -\left( 
\begin{tabular}{cc}
$\sigma _{l}^{I}(P)^2$ & $\rho ^{I}(P)\sigma _{l}^{I}(P)\sigma _{u}^{I}(P)$ \\ 
$\rho ^{I}(P)\sigma _{l}^{I}(P)\sigma _{u}^{I}(P)$ & $\sigma _{u}^{I}(P)^2$
\end{tabular}
\right) 
\end{equation*}}
is a negative semidefinite matrix.

We seek to compare the power of inference based on CIs constructed from inefficient and efficient bounds estimators. As discussed in Section \ref{sec:Main}, we compare their power by contrasting the limiting coverage rates of the corresponding CIs for sequences of parameter values outside the identified set. Since these parameters lie outside the identified set, a lower limiting coverage probability corresponds to higher statistical power against the (incorrect) null hypotheses. 

Our first result in this section compares the limiting coverage rates of inference based on $CI_{\alpha}^{1}$ when implemented with the efficient and inefficient estimators.

\begin{theorem}[Power comparison for $CI_{\alpha}^{1}$]\label{thm:CI1}
Let $\alpha \in (0,0.5)$ and Assumption \ref{ass:1} hold. Define $CI_{\alpha}^{1,E}$ and $CI_{\alpha}^{1,I}$ as in \eqref{eq:CI1} with $(\hat{\theta}_{l}^{E}, \hat{\theta}_{u}^{E}, \hat{\sigma}_{l}^{E}, \hat{\sigma}_{u}^{E}, \hat{\rho}^{E})$ and $(\hat{\theta}_{l}^{I}, \hat{\theta}_{u}^{I}, \hat{\sigma}_{l}^{I}, \hat{\sigma}_{u}^{I}, \hat{\rho}^{I})$, respectively. For any sequence $\{(P_{N}, \theta_{N}) \in \mathcal{P} \times \Theta_{I}(P_{N})^{c} \}_{N \in \mathbb{N}}$,
\begin{equation}
\underset{N \to \infty}{\lim \inf}~ \big( P_{N}(\theta_{N} \in CI_{\alpha}^{1,I}) - P_{N}(\theta_{N} \in CI_{\alpha}^{1,E})\big) ~~\geq ~~0.
\label{eq:CI1main}
\end{equation}
Furthermore, \eqref{eq:CI1main} holds strictly for suitable sequences of $\{(P_{N}, \theta_{N}) \in \mathcal{P} \times \Theta_{I}(P_{N})^{c} \}_{N \in \mathbb{N}}$.
\end{theorem}

In simple terms, Theorem \ref{thm:CI1} shows that $CI_{\alpha}^{1}$ based on the more efficient bounds estimator (i.e., $CI_{\alpha}^{1,E}$) is more powerful than when it is based on the less efficient bounds estimator (i.e., $CI_{\alpha}^{1,I}$). While both CIs satisfy the coverage goal in \eqref{eq:covGoal} with equality, Theorem \ref{thm:CI1} demonstrates that the former dominates the latter in terms of power.

Our second result compares the limiting coverage rates of inference based on $CI_{\alpha}^{2}$ when implemented with the efficient and inefficient estimators.

\begin{theorem}[Power comparison for $CI_{\alpha}^{2}$]\label{thm:CI2}
Let $\alpha \in (0,0.5)$ and Assumption \ref{ass:1} hold. Define $CI_{\alpha}^{2,E}$ and $CI_{\alpha}^{2,I}$ as in \eqref{eq:CI2} with $(\hat{\theta}_{l}^{E}, \hat{\theta}_{u}^{E}, \hat{\sigma}_{l}^{E}, \hat{\sigma}_{u}^{E}, \hat{\rho}^{E})$ and $(\hat{\theta}_{l}^{I}, \hat{\theta}_{u}^{I}, \hat{\sigma}_{l}^{I}, \hat{\sigma}_{u}^{I}, \hat{\rho}^{I})$, respectively. For any sequence $\{(P_{N}, \theta_{N}) \in \mathcal{P} \times \Theta_{I}(P_{N})^{c} \}_{N \in \mathbb{N}}$,
\begin{equation}
\underset{N  \to \infty}{\lim \inf} ~ \big(P_{N}(\theta_{N} \in CI_{\alpha}^{2,I}) - P_{N}(\theta_{N} \in CI_{\alpha}^{2,E})\big) ~~\geq ~~0.
\label{eq:CI2main}
\end{equation}
Furthermore, \eqref{eq:CI2main} holds strictly for suitable sequences of $\{(P_{N}, \theta_{N}) \in \mathcal{P} \times \Theta_{I}(P_{N})^{c}\}_{N \in \mathbb{N} }$.
\end{theorem} 

Our takeaway from Theorem \ref{thm:CI2} is analogous to that of Theorem \ref{thm:CI1}: $CI_{\alpha}^{2}$ based on the more efficient bounds estimator (i.e., $CI_{\alpha}^{2,E}$) is more powerful than when it is based on the less efficient bounds estimator (i.e., $CI_{\alpha}^{2,I}$). Both CIs satisfy the coverage goal in \eqref{eq:covGoal} with equality, but the former is more powerful than the latter.

The results of Theorems \ref{thm:CI1} and \ref{thm:CI2}, while novel, may not seem surprising. After all, it is natural to expect that a more efficient implementation of inference leads to greater statistical power. However, this intuition need not hold for our two remaining CIs, as shown by the next two results. We begin with the case of $CI_{\alpha}^{3}$.

\begin{theorem}[Power comparison for $CI_{\alpha }^{3}$]\label{thm:CI3}
Let $\alpha \in (0,0.5)$ and Assumption \ref{ass:1} hold. Define $CI_{\alpha }^{3,E}$ and $CI_{\alpha }^{3,I}$ as in \eqref{eq:CI3} with $(\hat{\theta}_{l}^{E},\hat{\theta}_{u}^{E},\hat{\sigma}_{l}^{E},\hat{\sigma}_{u}^{E},\hat{\rho}^{E})$ and $(\hat{\theta}_{l}^{I},\hat{\theta}_{u}^{I},\hat{\sigma}_{l}^{I},\hat{\sigma}_{u}^{I},\hat{\rho}^{I})$, respectively, implemented with the same sequence of constants $\{b_N\}_{N\in \mathbb{N}}$ that satisfies $b_N \to 0$ and $b_N \sqrt{N} \to\infty$. Then, there exist sequences of $\{(P_{N},\theta _{N})\in \mathcal{P}\times \Theta _{I}(P_{N})^{c}\}_{N\in \mathbb{N}}$  such that  
\begin{equation}
\lim _{N  \to \infty} P_{N}(\theta _{N}\in CI_{\alpha }^{3,I})~>~
\lim _{N  \to \infty} P_{N}(\theta _{N}\in CI_{\alpha }^{3,E})  \label{eq:CI3main}
\end{equation}
and there also exist other sequences of $\{(P_{N},\theta _{N})\in \mathcal{P}\times \Theta _{I}(P_{N})^{c}\}_{N\in \mathbb{N}}$ such that  
\begin{equation}
\lim _{N  \to \infty} P_{N}(\theta _{N}\in CI_{\alpha }^{3,I})~<~\lim _{N  \to \infty} P_{N}(\theta _{N}\in CI_{\alpha }^{3,E}).  \label{eq:CI3main2}
\end{equation}
\end{theorem}

Theorem \ref{thm:CI3} shows that power rankings between efficient and inefficient implementations of $CI_{\alpha}^{3}$ may vary with the underlying parameters of the problem. We now explain this phenomenon. In the context of OBS, the main difference between $CI_{\alpha}^{2}$ and $CI_{\alpha}^{3}$ is that the latter includes the indicator term $1[(\hat{\theta}_{u}-\hat{\theta}_{l})>b_{N}]$ in \eqref{eq:CI3_problem}. This term adjusts the critical values depending on whether the estimated identified set is ``long'' or ``short'' relative to the tuning parameter sequence $\{b_N\}_{N\geq 1}$. The proof of Theorem \ref{thm:CI3} focuses on short identified sets, under which the indicator term is zero.\footnote{A similar argument can be constructed for long identified sets. The proof restricts attention to the short case for brevity; the corresponding arguments are available upon request.} In this case, $CI_{\alpha}^{3}$ uses an asymptotic critical value corresponding to a two-sided testing problem, equal to the $(1-\alpha/2)$-quantile of the standard normal distribution. The coverage of $CI_{\alpha}^{3}$ increases with the normalized size of the identified set, given by $\mu/\sigma$, and decreases with the normalized distance from the identified set, captured by $\Psi_l/\sigma$ or $\Psi_u/\sigma$. See part (e) of Lemma \ref{lem:C3_limit} for the explicit expression. As the variance of the bounds estimator (i.e., $\sigma^2$) decreases, two opposing forces emerge. On the one hand, a smaller $\sigma$ increases $\mu/\sigma$, making the identified set appear longer relative to sampling uncertainty and thereby increasing coverage. On the other hand, a smaller $\sigma$ also increases $\Psi_l/\sigma$ and $\Psi_u/\sigma$, making the parameter appear farther from the identified set and thereby decreasing coverage. If the first effect dominates, we obtain the counterintuitive finding that a more efficient estimator leads to higher coverage. If the second effect dominates instead, we get the more intuitive result that a more efficient estimator reduces coverage.

At this point, one may wonder why the previous phenomenon does not arise for $CI_{\alpha}^{1}$ or $CI_{\alpha}^{2}$. The key difference is that, unlike $CI_{\alpha}^{3}$, these procedures do not use the $(1-\alpha/2)$-quantile of the normal distribution when the identified set is short. Instead, their critical values depend on the normalized size of the identified set through $G(\mu/\sigma)$, as shown in Lemma \ref{lem:C2_limit}. This dependence eliminates the phenomenon described above. A formal justification is provided in Lemma \ref{lem:tildeH}.

The possibility that a more efficient estimator leads to lower power arises from the way the critical value is constructed in $CI_{\alpha}^{3}$ to ensure robustness to violations of the superefficiency condition. In this sense, it can be interpreted as a cost of the robustness properties of $CI_{\alpha}^{3}$. When superefficiency holds automatically under OBS, this robustness is unnecessary, and $CI_{\alpha}^{3}$ incurs this cost without providing any compensating robustness benefit.

Our last result compares the power of efficient and inefficient implementations of $CI_{\alpha}^{4}$.

\begin{theorem}[Power comparison for $CI_{\alpha }^{4}$]\label{thm:CI4}
Let $\alpha \in (0,0.5)$ and Assumption \ref{ass:1} hold. Define $CI_{\alpha }^{4,E}$ and $CI_{\alpha }^{4,I}$ as in \eqref{eq:CI4} with $(\hat{\theta}_{l}^{E},\hat{\theta}_{u}^{E},\hat{\sigma}_{l}^{E},\hat{\sigma}_{u}^{E},\hat{\rho}^{E})$ and $(\hat{\theta}_{l}^{I},\hat{\theta}_{u}^{I},\hat{\sigma}_{l}^{I},\hat{\sigma}_{u}^{I},\hat{\rho}^{I})$, respectively. Then, there exist sequences of $\{(P_{N},\theta _{N})\in \mathcal{P}\times \Theta _{I}(P_{N})^{c}\}_{N\in \mathbb{N}}$  such that  
\begin{equation}
\lim _{N  \to \infty} P_{N}(\theta _{N}\in CI_{\alpha }^{4,E})~<~\lim _{N  \to \infty} P_{N}(\theta _{N}\in CI_{\alpha }^{4,I})  \label{eq:orderbounds1}
\end{equation}
and there also exist other sequences of $\{(P_{N},\theta _{N})\in \mathcal{P}\times \Theta _{I}(P_{N})^{c}\}_{N\in \mathbb{N}}$ such that  
\begin{equation}
\lim _{N  \to \infty} P_{N}(\theta _{N}\in CI_{\alpha }^{4,E})~>~\lim _{N  \to \infty} P_{N}(\theta _{N}\in CI_{\alpha }^{4,I}).  \label{eq:orderbounds2}
\end{equation}
\end{theorem}

Theorem \ref{thm:CI4} parallels Theorem \ref{thm:CI3} for $CI_{\alpha}^{4}$: the relative power of efficient and inefficient implementations depends on the underlying parameters. The intuition follows the same lines as for Theorem \ref{thm:CI3}. As in the analysis of $CI_{\alpha}^{3}$, we interpret the counterintuitive power properties in Theorem \ref{thm:CI4} as the cost associated with the desirable properties of $CI_{\alpha}^{4}$ discussed in \citet{stoye:2020}.

\section{Conclusions}\label{sec:conclusions}

This paper studies the power properties of CIs for a partially-identified parameter of interest with an interval identified set. We assume that the researcher has bounds estimators to construct the CIs proposed by \cite{imbens/manski:2004}, \cite{stoye:2009} and \cite{stoye:2020}, denoted by $CI_{\alpha }^{1}$,  $CI_{\alpha }^{2}$, $CI_{\alpha }^{3}$, and $CI_{\alpha }^{4}$. We also assume these estimators are ``ordered'' in the sense that the estimator of the lower bound is less than or equal to the estimator of the upper bound. We refer to this as the ``ordered bounds setup'' or OBS.

Under our conditions, the literature has established that $CI_{\alpha}^{1}$, $CI_{\alpha}^{2}$, $CI_{\alpha}^{3}$, and $CI_{\alpha}^{4}$ are all uniformly asymptotically valid. However, the literature has not investigated the power of inference associated with these CIs. That is, it does not assess the ability of these CIs to rule out parameters outside the identified set, which, by definition, are invalid candidates for $\theta_{0}(P)$.

In this context, this paper makes two contributions. Our first contribution is to compare the coverage probabilities of the four CIs across all possible parameter sequences that do not belong to the identified set. A higher coverage rate for these parameters translates into lower power. We formally show that $CI_{\alpha}^{1}$ and $CI_{\alpha}^{2}$ are equally powerful, and both dominate $CI_{\alpha}^{3}$ and $CI_{\alpha}^{4}$, with $CI_{\alpha}^{3}$ dominating $CI_{\alpha}^{4}$.

For our second contribution, we consider a favorable situation in which the researcher has two pairs of estimators for these CIs, with one pair known to be more efficient than the other. In this context, it is reasonable to expect that inference based on the more efficient pair leads to more powerful inference. We formally demonstrate that this conclusion holds for $CI_{\alpha}^{1}$ and $CI_{\alpha}^{2}$, but not necessarily for $CI_{\alpha}^{3}$ or $CI_{\alpha}^{4}$.

In the context of OBS, our findings indicate that $CI_{\alpha}^{1}$ and $CI_{\alpha}^{2}$ are equally powerful, and both are preferable to $CI_{\alpha}^{3}$ and $CI_{\alpha}^{4}$. This conclusion is valuable from a practical standpoint for two reasons. First, $CI_{\alpha}^{1}$ is straightforward to implement, whereas some of the other methods are not. Second, implementing $CI_{\alpha}^{3}$ requires a carefully calibrated tuning parameter sequence, whereas the other methods do not. It is important to note, however, that these conclusions depend critically on the OBS condition, and in particular on the requirement that the bounds be ``ordered''; in settings where this condition fails, the results of this paper need not apply.

\renewcommand{\theequation}{\Alph{section}-\arabic{equation}}
\begin{appendix} 

\section{Appendix}

This appendix uses ``CMT'' to denote the ``continuous mapping theorem'', ``UHC'' to denote ``upper hemicontinuous'', ``LHC'' to denote ``lower hemicontinuous'', and ``s.t.'' to denote ``subject to''. Also, we define $\bar{\mathbb{R}} = [-\infty, \infty]$, $\bar{\mathbb{R}}_{+} = [0,\infty]$, $\mathbb{R}_{+} = [0,+\infty)$, $\mathbb{R}_{++} = (0,\infty)$, and $\mathbb{R}_{--} = (-\infty,0)$. Unless specified otherwise, all limits are taken as $N \to \infty$.

\subsection{Proof of theorems}\label{sec:proofs}

\begin{proof}[Proof of Theorem \ref{thm:CIcomparison}]
\underline{Part (a).} We prove the result by contradiction. That is, suppose \eqref{eq:CI4main1} fails. Then, ${\lim \sup}_{N \to \infty} (P_{N}(\theta _{N}\in CI_{\alpha}^{1})-P_{N}(\theta _{N}\in CI_{\alpha}^{2}))\neq 0$ or ${\lim \inf}_{N \to \infty} (P_{N}(\theta _{N}\in CI_{\alpha}^{1})-P_{N}(\theta _{N}\in CI_{\alpha}^{2}))\neq 0$. In either case, we can find a subsequence $\{k_{N}\}_{N\in \mathbb{N}}$ s.t. 
\begin{equation}
\lim_{N \to \infty} (P_{k_{N}}(\theta _{k_{N}}\in CI_{\alpha}^{1})-P_{k_{N}}(\theta _{k_{N}}\in CI_{\alpha}^{2})) ~ \neq ~0. \label{eq:CIcomparison_1}
\end{equation}
The proof is completed by showing that \eqref{eq:CIcomparison_1} cannot hold.

By possibly taking a further subsequence, 
\begin{align}
&\left(
\begin{array}{c}
\theta _{l}(P_{k_{N}}),\theta _{u}(P_{k_{N}}),\sigma _{l}(P_{k_{N}}),\sigma _{u}(P_{k_{N}}),\rho (P_{k_{N}}),\\
\sqrt{k_{N}}(\theta _{u}(P_{k_{N}})-\theta _{l}(P_{k_{N}})), \sqrt{k_{N}}(\theta _{l}(P_{k_{N}})-\theta _{k_{N}}),\sqrt{k_{N}}(\theta _{k_{N}}-\theta _{u}(P_{k_{N}}))
\end{array}
\right) \notag\\
&  \to (\theta _{l},\theta _{u},\sigma _{l},\sigma _{u},\rho ,\mu ,\Psi _{l},\Psi _{u})\in \bar{\mathbb{R}}\times \bar{\mathbb{R}}\times [\underline{\sigma },\overline{\sigma }]\times [\underline{ \sigma },\overline{\sigma }]\times [-1,1]\times \bar{\mathbb{R}} _{+}\times \bar{\mathbb{R}}\times \bar{\mathbb{R}}.
\label{eq:CIcomparison_2}
\end{align}
We then divide the argument into four exhaustive cases, depending on the possible values of $(\mu ,\Psi _{l},\Psi _{u})$. In this regard, note that $\theta _{N}\in \Theta _{I}(P_{N})^{c}$ implies that either (i) $\sqrt{k_{N}}(\theta _{l}(P_{k_{N}})-\theta _{k_{N}}) >0$ or (ii) $\sqrt{k_{N}}(\theta _{k_{N}}-\theta _{u}(P_{k_{N}}))>0$. By taking limits, we conclude that either (i) $\Psi _{l}\geq 0$ or (ii) $\Psi _{u}\geq 0$. The proof is completed by showing that none of the following exhaustive cases satisfy \eqref{eq:CIcomparison_1}.

\noindent{Case 1:}  $\mu =\infty $ and $\Psi _{l}\geq 0$. Then, consider the following derivation:
\begin{align*}
  \lim_{N \to \infty}  P_{k_{N}}(\theta _{k_{N}}\in CI_{\alpha}^{1})~\overset{(1)}{=}~\Phi (\Phi ^{-1}(1-\alpha )-\Psi _{l}/\sigma_{l})  ~\overset{(2)}{=}~\lim_{N \to \infty}  P_{k_{N}}(\theta _{k_{N}}\in CI_{\alpha}^{2}),
\end{align*}
where (1) holds by \eqref{eq:CIcomparison_2}, $F_1(\mu,\sigma_l,\sigma_u) = \Phi^{-1}(1-\alpha)$, and part (a) of Lemma \ref{lem:C1_limit}, and (2) by part (a) of Lemma \ref{lem:C2_limit}. This equation implies that \eqref{eq:CIcomparison_1} fails.

\noindent{Case 2:}  $\mu =\infty $ and $\Psi _{u}\geq 0$. This case is analogous to case 1 except that we replace parts (a) of Lemmas \ref{lem:C1_limit} and \ref{lem:C2_limit} with parts (b) of these results.

\noindent{Case 3:} $\mu \in \mathbb{R} _{+}$ and $\Psi _{l}\geq 0$. By $\mu \in \mathbb{R} _{+}$, it follows that $\theta _{u}(P_{k_{N}})-\theta _{l}(P_{k_{N}}) \to 0$. By this and Lemma \ref{lem:near1}, we get $\rho =1$ and $\sigma _{l}=\sigma _{u}$. We set $\sigma=\sigma _{l}=\sigma _{u}$. Then, consider the following derivation:
\begin{align*}
  \lim_{N \to \infty}  P_{k_{N}}(\theta _{k_{N}}\in CI_{\alpha}^{1})~&\overset{(1)}{=}~\Phi \left( (\Psi_{l} +\mu )/\sigma +G(\mu /\sigma )\right) -\Phi \left( \Psi_{l} /\sigma -G(\mu /\sigma )\right)  \\
  ~&\overset{(2)}{=}~\lim_{N \to \infty}  P_{k_{N}}(\theta _{k_{N}}\in CI_{\alpha}^{2}),
\end{align*}
where (1) holds by \eqref{eq:CIcomparison_2}, $F_1(\mu, \sigma,\sigma) = G(\mu/\sigma)$, and part (a) of Lemma \ref{lem:C1_limit}, and (2) by part (c) of Lemma \ref{lem:C2_limit}. This equation implies that \eqref{eq:CIcomparison_1} fails.

\noindent{Case 4:} $\mu \in \mathbb{R} _{+}$ and $\Psi _{u}\geq 0$. This case is analogous to case 3 except that we replace part (a) of Lemma \ref{lem:C1_limit} with part (b), and part (c) of Lemma \ref{lem:C2_limit} with part (d).

\underline{Part (b).} We prove the result by contradiction. That is, suppose that ${\lim  \inf }_{N \to \infty}(P_{N}(\theta _{N}\in CI_{\alpha}^{3})-P_{N}(\theta _{N}\in CI_{\alpha}^{j}))<0$ for some $j=1,2$. Then, we can find a subsequence $\{k_{N}\}_{N\in \mathbb{N}}$ s.t. 
\begin{equation}
\lim_{N \to \infty} (P_{k_{N}}(\theta _{k_{N}}\in CI_{\alpha }^{3})-P_{k_{N}}(\theta _{k_{N}}\in CI_{\alpha }^{j})) ~<~0~~\text{ for some }j=1,2.  \label{eq:CIcomparison_2B}
\end{equation}
The proof is completed by showing that \eqref{eq:CIcomparison_2B} cannot hold.

By possibly taking a further subsequence, 
\begin{align}
& \left( 
\begin{array}{c}
\theta _{l}(P_{k_{N}}),\theta _{u}(P_{k_{N}}),\sigma _{l}(P_{k_{N}}),\sigma _{u}(P_{k_{N}}),\rho (P_{k_{N}}), \\ 
\sqrt{k_{N}}(\theta _{u}(P_{k_{N}})-\theta _{l}(P_{k_{N}})),\sqrt{k_{N}} (\theta _{l}(P_{k_{N}})-\theta _{k_{N}}),\sqrt{k_{N}}(\theta _{k_{N}}-\theta _{u}(P_{k_{N}}))
\end{array}
\right) \notag\\
&  \to  (\theta _{l},\theta _{u},\sigma _{l},\sigma _{u},\rho ,\mu ,\Psi _{l},\Psi _{u})\in \bar{\mathbb{R}}\times \bar{\mathbb{R}}\times [ \underline{\sigma },\overline{\sigma }]\times [ \underline{ \sigma },\overline{\sigma }]\times [ -1,1]\times \bar{\mathbb{R}} _{+}\times \bar{\mathbb{R}}\times \bar{\mathbb{R}}.
\label{eq:CIcomparison_3}
\end{align}
We then divide the argument into four exhaustive cases, depending on the possible values of $(\mu ,\Psi _{l},\Psi _{u})$. As before, we note that $\theta _{N}\in \Theta _{I}(P_{N})^{c}$ implies that either (i) $\Psi _{l}\geq 0$ or (ii) $\Psi _{u}\geq 0$. The proof is completed by showing that none of the following exhaustive cases satisfy \eqref{eq:CIcomparison_2B}.

\noindent {Case 1:} $\mu =\infty $ and $\Psi _{l}\geq 0$. Then, consider the following derivation for $j=1,2$,
\begin{equation*}
\lim_{N \to  \infty }P_{k_{N}}(\theta _{k_{N}}\in CI_{\alpha }^{j}) ~\overset{(1)}{=} ~\Phi (\Phi ^{-1}(1-\alpha )-\Psi _{l}/\sigma _{l})~\overset{(2)}{\leq} \lim_{N \to  \infty }P_{k_{N}}(\theta _{k_{N}}\in CI_{\alpha }^{3}),
\end{equation*}
where (1) holds by \eqref{eq:CIcomparison_3}, $F_1(\mu,\sigma_l,\sigma_u) = \Phi^{-1}(1-\alpha)$, and part (a) of Lemma \ref{lem:C1_limit} and Lemma \ref{lem:C2_limit}, and (2) by part (a) of Lemma \ref{lem:C3_limit}. This equation implies that \eqref{eq:CIcomparison_2B} fails.

\noindent {Case 2:} $\mu =\infty $ and $\Psi _{u}\geq 0$. This case is analogous to case 1 except that we replace parts (a) of Lemmas \ref{lem:C1_limit} and \ref{lem:C2_limit} with parts (b) of these results, and part (a) of Lemma \ref{lem:C3_limit} with part~(f).

\noindent {Case 3:} $\mu \in \mathbb{R}_{+}$ and $\Psi _{l}\geq 0$. By $\mu \in \mathbb{R}_{+}$ and Lemma \ref{lem:near1}, we get $\rho =1$ and $\sigma _{l}=\sigma _{u}$. We set $\sigma =\sigma _{l}=\sigma _{u}$. 

As a preliminary derivation, note that 
\begin{equation}
\Phi ^{-1}(1-\alpha /2)~\geq ~G(\mu /\sigma ). \label{eq:CIcomparison_4}
\end{equation}
To see this, note that $\Phi (c+\mu /\sigma )-\Phi (-c)$ is strictly increasing in $c$, and 
\begin{align*}
\Phi (\Phi ^{-1}(1-\alpha /2)+\mu /\sigma )-\Phi (-\Phi ^{-1}(1-\alpha /2))&~\overset{(1)}{\geq }~\Phi (\Phi ^{-1}(1-\alpha /2))-\Phi (-\Phi^{-1}(1-\alpha /2)) \\
& ~=~1-\alpha  \\
& ~\overset{(2)}{=}~\Phi (G(\mu /\sigma )+\mu /\sigma )-\Phi (-G(\mu /\sigma)),
\end{align*}
where (1) holds by $\mu /\sigma \geq 0$, and (2) by definition of $G(y)$. 

For $j=1,2$, we then have the following derivation:
\begin{align}
\lim_{N\to  \infty }P_{k_{N}}(\theta _{k_{N}} \in CI_{\alpha }^{3}) &~\overset{(1)}{=}~\Phi \left( (\Psi _{l}+\mu )/\sigma +\Phi ^{-1}(1-\alpha /2)\right) -\Phi \left( \Psi _{l}/\sigma -\Phi ^{-1}(1-\alpha /2)\right) \notag \\
&~\overset{(2)}{\geq }~\Phi \left( (\Psi _{l}+\mu )/\sigma +G(\mu /\sigma )\right) -\Phi \left( \Psi _{l}/\sigma -G(\mu /\sigma )\right)   \notag \\
&~\overset{(3)}{=}~\lim_{N\to  \infty }P_{k_{N}}(\theta _{k_{N}} \in CI_{\alpha }^{j}), \label{eq:CIcomparison_5}
\end{align}
where (1) holds by part (e) of Lemma \ref{lem:C3_limit}, (2) by \eqref{eq:CIcomparison_4} and that $\Phi ( (\Psi _{l}+\mu )/\sigma +c) -\Phi ( \Psi _{l}/\sigma -c) $ is increasing in $c$, and (3) by part (a) of Lemma \ref{lem:C1_limit}, part (c) of Lemma \ref{lem:C2_limit}, and $F_{1}(\mu ,\sigma ,\sigma )=G(\mu /\sigma )$. Note that this equation implies that \eqref{eq:CIcomparison_2B} fails. 

\noindent {Case 4:} $\mu \in \mathbb{R}_{+}$ and $\Psi _{u}\geq 0$. This case is analogous to case 3 except that we replace part (a) of Lemma \ref{lem:C1_limit} with part (b), part (c) of Lemma \ref{lem:C2_limit} with part (d), and part (e) of Lemma \ref{lem:C3_limit} with part (j).

To conclude the proof of this part, it suffices to verify \eqref{eq:CI4main2} holds strictly for a suitably chosen sequence $\{(P_{N},\theta _{N})\in \mathcal{P}\times \Theta _{I}(P_{N})^{c}\}_{N\in \mathbb{N}}$. We can consider two cases: ${\lim \sup } _{N \to \infty }\sqrt{N}(\theta _{u}(P_{N})-\theta _{l}(P_{N}))<\infty$ or  ${\lim \sup } _{N \to \infty }\sqrt{N}(\theta _{u}(P_{N})-\theta _{l}(P_{N}))=\infty$. For brevity, we focus on the first case. 

To this end, consider Example \ref{ex:OBS_Ex} with $\underline{Y}=0$, $\overline{Y}=1$, $\{Y_{i}|Z_{i}=1\}\sim Be(1/2)$, and $Z_{i}\sim Be(1-\pi (P_{N}))$, with $\pi (P_{N}) = a_1 /\sqrt{N} \downarrow 0$ for any $a_1 \in (0,\infty)$, which leads to ${\lim \sup } _{N \to \infty }\sqrt{N}(\theta _{u}(P_{N})-\theta _{l}(P_{N}))  <\infty$. Consider coverage of $\theta _{N}=\theta _{l}(P_{N})-a_2/ \sqrt{N}\in\Theta _{I}(P_{N})^{c}$ for any $a_2>0$ using $CI_{\alpha }^{3}$ implemented with $b_{N}=(\ln N)/\sqrt{N}$. Then, \eqref{eq:CIcomparison_2} and \eqref{eq:CIcomparison_3} hold with $(\theta _{l},\theta _{u},\sigma ,\sigma ,\rho ,\mu ,\Psi _{l},\Psi _{u}) = (1/2, 1/2, 1/2, 1/2, 1, a_1, a_2, - a_1 -a_2)$. For $j=1,2$, we then obtain 
\begin{align}
  \lim_{N \to \infty}  P_{N}(\theta _{N}\in CI_{\alpha}^{j})~&\overset{(1)}{=}~\Phi \left( 2(a_2 +a_1 )+G(2a_1 )\right) -\Phi \left( 2 a_2 -G(2a_1  )\right) \notag \\
  ~&\overset{(2)}{<}~\Phi \left( 2( a_2 +a_1 )+\Phi ^{-1}(1-\alpha /2)\right) -\Phi \left( 2 a_2 -\Phi ^{-1}(1-\alpha /2)\right)\notag\\
  ~&\overset{(3)}{=}~ \lim_{N \to \infty} P_{N}(\theta _{N}\in CI_{\alpha}^{3}),\label{eq:CIcomparison_20}
\end{align}
where (1) holds by part (a) of Lemmas \ref{lem:C1_limit} and \ref{lem:C2_limit}, (2) by repeating the derivation in \eqref{eq:CIcomparison_4}, and (3) by part (e) of Lemma \ref{lem:C3_limit}. For concreteness, $a_1=1$, $a_2=1$, and $\alpha=0.05$ give $\lim_{N\to\infty}P_N(\theta_N\in CI_\alpha^{j})=0.36<0.48=\lim_{N\to\infty}P_N(\theta_N\in CI_\alpha^{3})$ for $j=1,2$. Note that \eqref{eq:CIcomparison_20} implies that \eqref{eq:CI4main2} holds strictly.

\underline{Part (c).} These results are shown in the supplementary appendix.  
\end{proof}

\begin{proof}[Proof of Theorem \ref{thm:CI1}]
We prove \eqref{eq:CI1main} by contradiction. To this end, assume \eqref{eq:CI1main} fails:
\begin{equation}
\underset{N  \to \infty }{\text{lim inf} }~~\big(P_{N}(\theta _{N}\in CI_{\alpha}^{1,I})-P_{N}(\theta _{N}\in CI_{\alpha}^{1,E})\big)~<~0.  \label{eq:CI1_eq1}
\end{equation}
By possibly taking a subsequence $\{k_{N}\}_{N\geq 1}$, the following sequence converges:
\begin{align}
&\left(
\begin{array}{c}
\theta _{l}(P_{k_{N}}),\theta _{u}(P_{k_{N}}),\sigma _{l}^{E}(P_{{ k_{N}}}),\sigma _{u}^{E}(P_{{k_{N}}}),\rho ^{E}(P_{{k_{N}}}),\sigma _{l}^{I}(P_{{k_{N}}}),\sigma _{u}^{I}(P_{{k_{N}}}),\rho ^{I}(P_{{k_{N}} }),\\
\sqrt{k_{N}}(\theta _{u}(P_{k_{N}})-\theta _{l}(P_{k_{N}})), \sqrt{k_{N}}(\theta _{l}(P_{k_{N}})-\theta _{k_{N}}),\sqrt{k_{N}}(\theta _{k_{N}}-\theta _{u}(P_{k_{N}}))
\end{array}
\right) \notag \\
&  \to (\theta _{l},\theta _{u},\sigma^{E} _{l},\sigma^{E} _{u},\rho^{E},\sigma^{I} _{l},\sigma^{I} _{u},\rho^{I} ,\mu ,\Psi _{l},\Psi _{u}), \notag\\
&\in \bar{\mathbb{R}}\times \bar{\mathbb{R}}\times [\underline{\sigma },\overline{\sigma }]\times [\underline{ \sigma },\overline{\sigma }]\times [-1,1]\times [\underline{\sigma },\overline{\sigma }]\times [\underline{ \sigma },\overline{\sigma }]\times [-1,1]\times \bar{\mathbb{R}} _{+}\times \bar{\mathbb{R}}\times \bar{\mathbb{R}}. \label{eq:seqCI1}
\end{align}
By Assumption \ref{ass:1}, $\sigma _{l}^{E}\leq \sigma _{l}^{I}$ and $\sigma _{u}^{E}\leq \sigma _{u}^{I}$. We then divide the argument into four exhaustive cases, depending on the possible values of $(\mu ,\Psi _{l},\Psi _{u})$. As in the proof of Theorem \ref{thm:CIcomparison}, we note that $\theta _{N}\in \Theta _{I}(P_{N})^{c}$ implies that either (i) $\Psi _{l}\geq 0$ or (ii) $\Psi _{u}\geq 0$. The proof is completed by showing that none of the following exhaustive cases satisfy \eqref{eq:CI1_eq1}.

\noindent{Case 1:} $\mu =\infty $ and $\Psi _{l}\geq 0$. Then, 
\begin{align*}
\lim_{N  \to \infty} P_{N}(\theta _{N} \in CI_{\alpha}^{1,E})&~\overset{(1)}{=}~\Phi (\Phi ^{-1}(1-\alpha )-\Psi _{l}/\sigma _{l}^{E}) \\
&~\overset{(2)}{\leq }~\Phi (\Phi ^{-1}(1-\alpha )-\Psi _{l}/\sigma _{l}^{I}) \\
&~\overset{(3)}{=}~\lim_{N  \to \infty} P_{N}(\theta _{N} \in CI_{\alpha}^{1,I}),
\end{align*}
as desired, where (1) and (3) hold by part (a) of Lemma \ref{lem:C1_limit}, and (2) by $\Psi _{l}\geq 0$ and $\sigma _{l}^{E}\leq \sigma _{l}^{I}$. 

\noindent{Case 2:} $\mu =\infty $ and $\Psi _{u}\geq 0$. This case is analogous to case 1 except that we replace part (a) of Lemma \ref{lem:C1_limit} with part (b).

\noindent{Case 3:} $\mu \in \mathbb{R}_{+}$ and $\Psi _{l}\geq 0$. By $\mu \in \mathbb{R}_{+}$ and Lemma \ref{lem:near1}, $\sigma  ^{E}_{l}=\sigma ^{E} _{u}$, $\sigma  ^{I}_{l}=\sigma ^{I} _{u}$, $\rho ^{E} =1$, and $\rho ^{I} =1$. We set $\sigma ^{I}=\sigma _{l}^{I}=\sigma _{u}^{I}$ and $\sigma ^{E}=\sigma _{l}^{E}=\sigma _{u}^{E}$. Then, 
\begin{align*}
\lim_{N  \to \infty} P_{N}(\theta _{N} \in CI_{\alpha}^{1,E})& ~\overset{(1)}{=}~\Phi \left( (\Psi _{l}+\mu )/\sigma ^{E}+G(\mu /\sigma ^{E})\right) -\Phi \left( \Psi _{l}/\sigma ^{E}-G(\mu /\sigma ^{E})\right)  \\
&~\overset{(2)}{=}~H(\sigma ^{E},\mu ,\Psi _{l}) \\
&~\overset{(3)}{\leq }~H(\sigma ^{I},\mu ,\Psi _{l}) \\
&~\overset{(4)}{=}~\Phi \left( (\Psi _{l}+\mu )/\sigma ^{I}+G(\mu /\sigma ^{I})\right) -\Phi \left( \Psi _{l}/\sigma ^{I}-G(\mu /\sigma ^{I})\right) 
\\
&~\overset{(5)}{=}~\lim_{N \to \infty} P_{N}(\theta _{N} \in CI_{\alpha}^{1,I}),
\end{align*}
as desired, where (1) and (5) hold by part (c) of Lemma \ref{lem:C1_limit}, (2) and (4) by $H:[\underline{ \sigma },\overline{\sigma }]\times \mathbb{R}_{+}\times \bar{\mathbb{R}} _{+}  \to \mathbb{R}$ defined in Lemma \ref{lem:tildeH}, and (3) by part (a) of Lemma \ref{lem:tildeH} and $\sigma _{l}^{E}\leq \sigma _{l}^{I}$.

\noindent{Case 4:} $\mu \in \mathbb{R}_{+}$ and $\Psi _{u}\geq 0$. This case is analogous to case 3 except that we replace part (c) of Lemma \ref{lem:C1_limit} with part (d).

To conclude the proof, it suffices to show that the inequality in \eqref{eq:CI1main} holds strictly for suitable sequences of $\{(P_{N},\theta _{N})\in \mathcal{P}\times \Theta _{I}(P_{N})^{c}\}_{N\in \mathbb{N}}$. To this end, consider $\{(P_{N},\theta _{N})\}_{N\in \mathbb{N}}$ that satisfies \eqref{eq:seqCI1} with $\sigma ^{E}<\sigma ^{I}$, and either (i) $ \Psi _{l}>0$ or (ii) $\Psi _{u}>0$. For brevity, focus on (i), but analogous results hold for (ii). 

We now have two cases. If $\mu =\infty $, then the desired strict inequality follows from the derivation in case 1. The difference with this derivation is that the weak inequality now becomes strict under $\Psi _{l}>0$ and $\sigma ^{E}<\sigma ^{I}$. In turn, if $\mu \in \mathbb{R}_{+}$, the desired strict inequality follows from the derivations in case 3. The difference with this derivation is that the weak inequality now becomes strict under $\Psi _{l}>0$, $\sigma ^{E}<\sigma ^{I}$, and part (b) of Lemma \ref{lem:tildeH}.

We now illustrate the case with $\mu =\infty $ in the context of Example \ref{ex:OBS_Ex}. Within this example, set $\underline{Y}=0$, $\overline{Y}=1$, $\{Y_{i}|Z_{i}=1\}\sim Be(1/2)$, and $Z_{i}\sim Be(1-\pi (P_{N}))$, with $\pi (P_{N}) =a_1 \in (0,1)$, which leads to ${\lim \sup } _{N \to \infty }\sqrt{N}(\theta _{u}(P_{N})-\theta _{l}(P_{N})) =\infty$. Our goal is to cover $\theta _{N}=\theta _{l}(P_{N})- a_2/ \sqrt{N}\in\Theta _{I}(P_{N})^{c}$ for some $ a_2 
>0$. In this context, consider two estimators for the bounds. The first one is \eqref{eq:estimator_bounds_example}, which serves as the efficient estimator. The second one is as in \eqref{eq:estimator_bounds_example} but using only the first $\lfloor  a_3 N\rfloor $ observations of the sample for any $ a_3 \in (0,1)$. This will be the inefficient estimator, as it uses a fraction of the available sample. Both bounds estimators can be represented as
\begin{equation}
( \hat{\theta}_{l}( \lambda ) ,\hat{\theta}_{u}( \lambda ) ) ~=~\Big( \frac{1}{\lfloor \lambda N\rfloor } \sum_{i=1}^{\lfloor \lambda N\rfloor }( Y_{i}Z_{i}+\underline{Y}( 1-Z_{i}) ) ,\frac{1}{\lfloor \lambda N\rfloor } \sum_{i=1}^{\lfloor \lambda N\rfloor }( Y_{i}Z_{i}+\overline{Y}(
1-Z_{i}) ) \Big) .  \label{eq:CI1_comp_1}
\end{equation}
For $\lambda =1$, we get $( \hat{\theta}_{l}^{E}, \hat{\theta}_{u}^{E}) =( \hat{\theta}_{l}( 1) ,\hat{ \theta}_{u}( 1) ) $, and for $\lambda = a_3\in (0,1) $, we get $( \hat{\theta}_{l}^{I},\hat{\theta} _{u}^{I}) =( \hat{\theta}_{l}( a_3) ,\hat{ \theta}_{u}( a_3))$. Finally, we assume the estimators for $( \sigma _{l},\sigma _{u},\rho )$ are all obtained by their corresponding sample analogs: $( \hat{\sigma}_{l}^{E},\hat{\sigma}_{u}^{E},\hat{\rho} ^{E}) $ uses the entire sample (i.e., $\lambda =1$), and $( \hat{ \sigma}_{l}^{I},\hat{\sigma}_{u}^{I},\hat{\rho}^{I}) $ relies only on the first $\lfloor  a_3 N\rfloor $ observations of the sample, with the standard-deviation estimators scaled by $\sqrt{N/\lfloor a_3N\rfloor}$.
By repeating arguments in Example \ref{ex:OBS_Ex}, we can verify that $ ( \hat{\theta}_{l}^{E},\hat{\theta}_{u}^{E},\hat{\sigma}_{l}^{E},\hat{ \sigma}_{u}^{E},\hat{\rho}^{E}) $ and $( \hat{\theta}_{l}^{I}, \hat{\theta}_{u}^{I},\hat{\sigma}_{l}^{I},\hat{\sigma}_{u}^{I},\hat{\rho} ^{I}) $ satisfy Assumption \ref{ass:1}. Moreover, we can show that \eqref{eq:seqCI1} holds with
\begin{align*}
    &(\theta _{l},\theta _{u},\sigma^{E} _{l},\sigma^{E} _{u},\rho^{E},\sigma^{I} _{l},\sigma^{I} _{u},\rho^{I} ,\mu ,\Psi _{l},\Psi _{u}) \\
    &=~ \Big(\frac{1-a_1}{2},\frac{1+a_1}{2},
\frac{\sqrt{1-a_1^{2}}}{2},\frac{\sqrt{1-a_1^{2}}}{2},\frac{1-a_1}{1+a_1},
\frac{\sqrt{1-a_1^{2}}}{2\sqrt{ a_3}},\frac{\sqrt{1-a_1^{2}}}{2\sqrt{ a_3}},\frac{1-a_1}{1+a_1},
+\infty,a_{2},-\infty\Big).
\end{align*}
At this point, we can repeat the derivation in case 1 to get the following:
\begin{align*}
\lim_{N  \to \infty} P_{N}(\theta _{N} \in CI_{\alpha}^{1,E})&~\overset{(1)}{=}~\Phi \Big(\Phi ^{-1}(1-\alpha )-2  a_2 /\sqrt{1-a_1^2}\Big) \\
&~\overset{(2)}{<}~\Phi \Big(\Phi ^{-1}(1-\alpha )-2\sqrt{ a_3} a_2 /\sqrt{1-a_1^2}\Big) \\
&~\overset{(3)}{=}~\lim_{N  \to \infty} P_{N}(\theta _{N} \in CI_{\alpha}^{1,I}),
\end{align*}
as desired, where (1) and (3) hold by part (a) of Lemma \ref{lem:C1_limit}, and (2) by $a_3 \in (0,1)$. For concreteness,  $a_1=1/2$, $a_2=1$, $a_3=0.3$, and $\alpha=0.05$ yield $\lim_{N\to\infty}P_N(\theta_N\in CI_{\alpha}^{1,I})=0.69>0.25=\lim_{N\to\infty}P_N(\theta_N\in CI_{\alpha}^{1,E})$.
\end{proof}

\begin{proof}[Proof of Theorem \ref{thm:CI2}] 
This result follows from part (a) of Theorem \ref{thm:CIcomparison} and Theorem \ref{thm:CI1}.
\end{proof}

\begin{proof}[Proof of Theorem \ref{thm:CI3}] To show this result, we construct specific sequences where \eqref{eq:CI3main} and \eqref{eq:CI3main2} can occur. We focus on sequences $\{(P_{N},\theta _{N})\in \mathcal{P}\times \Theta _{I}(P_{N})^{c}\}_{N\in \mathbb{N}}$  s.t.\  
\begin{align}
&\Bigg( 
\begin{array}{c}
\theta _{l}( P_{N}) ,\theta _{u}( P_{N}) ,\sigma _{l}^{E}( P_{N}) ,\sigma _{u}^{E}( P_{N}) ,\rho ^{E}( P_{N}) ,\sigma _{l}^{I}( P_{N}) ,\sigma _{u}^{I}( P_{N}) ,\rho ^{I}( P_{N})  \\ 
\sqrt{N}( \theta _{u}( P_{N}) -\theta _{l}( P_{N}) ) ,\sqrt{N}( \theta _{l}( P_{N}) -\theta _{N}) ,\sqrt{N}( \theta _{N}-\theta _{u}( P_{N}) ) 
\end{array}
\Bigg)  \notag \\
&\to ( \theta _{l},\theta _{u},\sigma _{l}^{E},\sigma _{u}^{E},\rho ^{E},\sigma _{l}^{I},\sigma _{u}^{I},\rho ^{I}, \mu ,\Psi _{l},\Psi _{u}) .
\label{eq:CI3_comp}
\end{align}
with $\Psi _{l}\geq 0$ and $\mu \in \mathbb{R}  _{+}$. By Lemma \ref{lem:near1}, $\rho^{E}=\rho^{I} =1$ and $\sigma _{l}^E=\sigma _{u}^E$, and $\sigma _{l}^{I}=\sigma _{u}^{I}$.

We can construct concrete sequences in the context of Example \ref{ex:OBS_Ex}. In particular, we use Example \ref{ex:OBS_Ex} with $\underline{Y}=0$, $\overline{Y}=1$, $\{Y_{i}|Z_{i}=1\}\sim Be(1/2)$, and $Z_{i}\sim Be(1-\pi (P_{N}))$, where $\pi (P_{N})~=~a_1/\sqrt{N}~\downarrow ~0$ for some $a_1 >0$. We consider coverage of $\theta _{N}=\theta _{l}(P_{N}) - a_2/\sqrt{N}\in\Theta _{I}(P_{N})^{c}$ for some $a_2>0$ and $CI_{\alpha }^{3}$ implemented with the subsample of the data $ \{(Y_{i},Z_{i})\}_{i=1}^{\lfloor \lambda N\rfloor }$ for $\lambda \in (0,1]$. In this context, we consider two estimators for the bounds. The efficient estimator uses the full sample, while the inefficient estimator uses only a fraction $\lambda = a_3 \in (0,1)$ of the sample. By repeating arguments in Theorem \ref{thm:CI1}, we deduce that \eqref{eq:CI3_comp} holds with
\begin{equation*}
     ( \theta _{l},\theta _{u},\sigma _{l}^{E},\sigma _{u}^{E},\rho ^{E},\sigma _{l}^{I},\sigma _{u}^{I},\rho ^{I}, \mu ,\Psi _{l}) ~=~ \Big(\frac{1}{2},\frac{1}{2}, \frac{1}{2},\frac{1}{2},1, \frac{1}{2\sqrt{a_3}},\frac{1}{2\sqrt{a_3}},1,a_1,a_2\Big).
\end{equation*}
Part (e) of Lemma \ref{lem:C3_limit} then yields
\begin{align*}
\lim_{N \to \infty} P_{N}(\theta _{N} \in CI_{\alpha }^{3,E})& ~=~\Phi \big( 2( a_2+a_1)+\Phi ^{-1}( 1-\alpha /2) \big)-
\Phi \big( 2a_2-\Phi ^{-1}( 1-\alpha /2) \big), \\
\lim_{N \to \infty} P_{N}(\theta _{N} \in CI_{\alpha }^{3,I})&~=~\Phi \big( 2\sqrt{a_3}(a_2+a_1 ) +\Phi ^{-1}( 1-\alpha /2) \big)-\Phi \big( 2\sqrt{a_3}a_2-\Phi ^{-1}( 1-\alpha /2) \big).
\end{align*}

We now verify the strict inequalities. To obtain \eqref{eq:CI3main}, set $a_1=a_2=1$, $a_3=0.3$, and $\alpha=0.05$, which give
\begin{align*}
\lim_{N\to \infty} P_{N}(\theta _{N} \in CI_{\alpha }^{3,E})~=~0.48~<~0.81~=~\lim_{N\to \infty} P_{N}(\theta _{N} \in CI_{\alpha }^{3,I}).
\end{align*}
To obtain \eqref{eq:CI3main2}, set $a_1=0.15$, $a_2=0.01$, $a_3=0.3$, and $\alpha=0.05$, which give
\begin{align*}
\lim_{N\to \infty} P_{N}(\theta _{N} \in CI_{\alpha }^{3,E})~=~0.963 ~>~ 0.958~=~\lim_{N\to \infty} P_{N}(\theta _{N} \in CI_{\alpha }^{3,I}).
\end{align*}
This completes the proof.
\end{proof}
 
\begin{proof}[Proof of Theorem \ref{thm:CI4}] This result is shown in the supplementary appendix.
\end{proof}

\subsection{Auxiliary results}

\begin{lemma}[$CI_{\alpha}^{1}$]\label{lem:C1_limit}
Assume $\alpha \in (0,0.5)$ and that $(\hat{\theta}_{l},\hat{\theta}_{u},\hat{\sigma}_{l},\hat{\sigma} _{u},\hat{\rho})$ satisfies OBS with parameter $({\theta }_{l}(P),{\theta } _{u}(P),{\sigma }_{l}(P),{\sigma }_{u}(P),{\rho }(P))$ and set $\mathcal{P}$. Let $F_{1}(\delta ,\sigma _{l},\sigma _{u}): \bar{\mathbb{R}}_{+}\times [\underline{\sigma },\overline{\sigma } ]\times [\underline{\sigma },\overline{\sigma }]  \to \bar{ \mathbb{R}}$ be the unique $c>0$ that solves 
\begin{equation}
\Phi \left( c+{\delta }/{\max \{\sigma _{l},\sigma _{u}\}}\right) -\Phi
(-c)~=~1-\alpha .  \label{eq:CI1_eq2}
\end{equation}
Consider any sequence $\{(P_{N},\theta _{N})\in \mathcal{P}\times \Theta
_{I}(P_{N})^{c}\}_{N\in \mathbb{N}}$ with
\begin{align}
&\left( 
\begin{array}{c}
\theta _{l}(P_{N}),\theta _{u}(P_{N}),\sigma _{l}(P_{N}),\sigma _{u}(P_{N}),\rho (P_{N}),\\
\sqrt{N}(\theta _{u}(P_{N})-\theta _{l}(P_{N})),\sqrt{N}(\theta _{l}(P_{N})-\theta _{N}),\sqrt{N}(\theta _{N}-\theta _{u}(P_{N}))
\end{array}
\right)  \notag \\
&  \to (\theta _{l},\theta _{u},\sigma _{l},\sigma _{u},\rho ,\mu ,\Psi _{l},\Psi _{u})\in \bar{\mathbb{R}}\times \bar{\mathbb{R}}\times [\underline{\sigma },\overline{\sigma }]\times [\underline{ \sigma },\overline{\sigma }]\times [-1,1]\times \bar{\mathbb{R}} _{+}\times \bar{\mathbb{R}}\times \bar{\mathbb{R}}. \label{eq:C1_limit_1}
\end{align}
Under these conditions, we have the following results.
\begin{enumerate}[(a)]
\item If $\Psi _{l} \geq 0$,
\begin{align*}
&\lim_{N  \to \infty} P_{N}(\theta _{N} \in CI_{\alpha}^{1})= \\
&P\left( 
\begin{array}{c}
\{ z_{1}-F_{1}(z_{2}\sigma _{u}\sqrt{1-\rho ^{2}}+z_{1}(\rho \sigma _{u}-\sigma _{l})+\mu ,\sigma _{l},\sigma _{u})\leq -\frac{\Psi _{l} }{\sigma _{l} }\} ~\cap  \\ 
\{ -\frac{\Psi _{l} +\mu }{\sigma _{u}}\leq z_{2}\sqrt{1-\rho ^{2}} +z_{1}\rho +F_{1}(z_{2}\sigma _{u}\sqrt{1-\rho ^{2}}+z_{1}(\rho \sigma _{u}-\sigma _{l})+\mu ,\sigma _{l},\sigma _{u})\} 
\end{array}
\right) ,
\end{align*}
where $(z_{1},z_{2})\sim \mathcal{N}(\mathbf{0}_{2\times 1},\mathbf{I}_{2\times 2})$.
\item If $\Psi _{u} \geq 0$,
\begin{align*}
&\lim_{N  \to \infty} P_{N}(\theta _{N} \in CI_{\alpha}^{1})= \\
& P\left( 
\begin{array}{c}
\{ \frac{\Psi _{u} }{\sigma _{u}}\leq z_{2}\sqrt{1-\rho ^{2}}+z_{1}\rho +F_{1}(z_{2}\sigma _{u}\sqrt{1-\rho ^{2}}+z_{1}(\rho \sigma _{u}-\sigma _{l})+\mu ,\sigma _{l},\sigma _{u})\}  \\ 
\cap ~\{ z_{1}-F_{1}(z_{2}\sigma _{u}\sqrt{1-\rho ^{2}}+z_{1}(\rho \sigma _{u}-\sigma _{l})+\mu ,\sigma _{l},\sigma _{u})\leq \frac{\Psi _{u} +\mu }{ \sigma _{l}}\} 
\end{array}
\right) ,
\end{align*}
where $(z_{1},z_{2})\sim \mathcal{N}(\mathbf{0}_{2\times 1},\mathbf{I}_{2\times 2})$.
\end{enumerate}
\end{lemma}
\begin{proof}
We begin by confirming that $F_{1}(\delta ,\sigma _{l},\sigma _{u})$ is unique and positive. To this end, note that the left-hand side of \eqref{eq:CI1_eq2} is strictly increasing and continuous in $c$. Moreover, at $c=0$, it is equal to $\Phi (\delta /\max \{\sigma _{l},\sigma _{u}\})-\Phi (0)\leq 0.5<1-\alpha $, whereas, as $c\to\infty$, $\lim _{c\to\infty}{\Phi ( c+{\delta }/{\max \{\sigma _{l},\sigma _{u}}\}) -\Phi(-c)}=1>1-\alpha $. This establishes, existence, uniqueness, and positivity.

Next, we show that $F_{1}:\bar{\mathbb{R}}_{+}\times [\underline{ \sigma },\overline{\sigma }]\times [\underline{\sigma },\overline{ \sigma }]  \to \bar{\mathbb{R}}$ is a continuous function of its arguments. To this end, consider any sequence $\{(\delta _{M},\sigma _{M,l},\sigma _{M,u})\}_{M\in \mathbb{N}}$ s.t.\ $\lim_{M\to \infty} (\delta _{M},\sigma _{M,l},\sigma _{M,u}) =(\delta ,\sigma _{l},\sigma _{u})\in \bar{ \mathbb{R}}_{+}\times [\underline{\sigma },\overline{\sigma }]\times [\underline{\sigma },\overline{\sigma }]$. In search of a contradiction, assume that $\{F_{1}(\delta _{M},\sigma _{M,l},\sigma _{M,u})\}_{M \in \mathbb{N}}$ does not converge to $F_{1}(\delta ,\sigma _{l},\sigma _{u})$. 
Then, there exist $\Delta>0$ and a subsequence $\{k_M\}_{M\in\mathbb N}$ such that either $F_1(\delta_{k_M},\sigma_{k_M,l},\sigma_{k_M,u})-F_1(\delta,\sigma_l,\sigma_u)>\Delta$ for all $M\in\mathbb N$ or $F_1(\delta_{k_M},\sigma_{k_M,l},\sigma_{k_M,u})-F_1(\delta,\sigma_l,\sigma_u)<-\Delta$ for all $M\in\mathbb N$. We assume the first case, but the argument is analogous to the other case. Since $\lim_{M\to \infty}  (\delta _{M},\sigma _{M,l},\sigma _{M,u}) = (\delta ,\sigma _{l},\sigma _{u})\in \bar{ \mathbb{R}}_{+}\times [\underline{\sigma } ,\bar{\sigma}]\times [\underline{\sigma } , \bar{\sigma}]$, we have that $\varphi _{k_{M}}=\delta _{k_{M}}/\max \{\sigma _{k_{M},l},\sigma _{k_{M},u}\}  \to \varphi =\delta /\max \{\sigma _{l},\sigma _{u}\}$. Here, there are two cases. Case 1: $\delta <\infty $, and so $ \varphi <\infty $, or case 2: $\delta =\varphi =\infty $. 

\noindent Case 1:  $\delta <\infty $. Then, $\exists M_{1}\in \mathbb{N}$ s.t.\ $|\varphi _{k_{M}}-\varphi |<\Delta /2$ for all $M>M_{1}$. Then, for all $M\geq M_{1}$, we reach the following contradiction.
\begin{align}
1-\alpha & ~\overset{(1)}{=}~\Phi (F_{1}(\delta _{k_{M}},\sigma _{k_{M},l},\sigma _{k_{M},u})+\varphi _{k_{M}})-\Phi (-F_{1}(\delta _{k_{M}},\sigma _{k_{M},l},\sigma _{k_{M},u}))  \notag \\
&~ \overset{(2)}{\geq }~\Phi (\Delta /2+F_{1}(\delta ,\sigma _{l},\sigma _{u})+\varphi )-\Phi (-\Delta -F_{1}(\delta ,\sigma _{l},\sigma _{u})) \notag \\
&~ \overset{(3)}{>}~\Phi (F_{1}(\delta ,\sigma _{l},\sigma _{u})+\varphi )-\Phi (-F_{1}(\delta ,\sigma _{l},\sigma _{u}))  \notag \\
&~ \overset{(4)}{=}~1-\alpha ,
\end{align}
where (1) holds by \eqref{eq:CI1_eq2} with $c=F_{1}(\delta _{k_{M}},\sigma _{k_{M},l},\sigma _{k_{M},u})$, (2) by $F_{1}(\delta _{k_{M}},\sigma _{k_{M},l},\sigma _{k_{M},u})-F_{1}(\delta ,\sigma _{l},\sigma _{u})>\Delta $ and $|\varphi _{k_{M}}-\varphi |<\Delta /2$, and (3) by $\Delta >0$, and (4) by \eqref{eq:CI1_eq2} with $c=F_{1}(\delta ,\sigma _{l},\sigma _{u})$. 

\noindent Case 2: $\delta =\varphi =\infty $.  In this case, $F_{1}(\delta ,\sigma _{l},\sigma _{u})=\Phi ^{-1}(1-\alpha )$ and $\varphi _{k_{M}}  \to \infty $. By this and $\Delta >0$, $\exists M_{1}\in \mathbb{N}$ s.t.\ for all $M>M_{1}$, 
\begin{equation}
\Phi (\Delta +\Phi ^{-1}(1-\alpha )+\varphi _{k_{M}})-\Phi (-\Delta -\Phi ^{-1}(1-\alpha ))~>~1-\alpha .
\label{eq:CI1_eq2c}
\end{equation}
Then, for all $M\geq M_{1}$, we have the following contradiction. 
\begin{align*}
1-\alpha & ~\overset{(1)}{=}~\Phi (F_{1}(\delta _{k_{M}},\sigma _{k_{M},l},\sigma _{k_{M},u})+\varphi _{k_{M}})-\Phi (-F_{1}(\delta _{k_{M}},\sigma _{k_{M},l},\sigma _{k_{M},u}))  \notag \\
& ~\overset{(2)}{\geq }~\Phi (\Delta +F_{1}(\delta ,\sigma _{l},\sigma _{u})+\varphi _{k_{M}})-\Phi (-\Delta -F_{1}(\delta ,\sigma _{l},\sigma _{u}))  \notag \\
& ~\overset{(3)}{=}~\Phi (\Delta +\Phi ^{-1}(1-\alpha )+\varphi _{k_{M}})-\Phi (-\Delta -\Phi ^{-1}(1-\alpha ))\\
& ~\overset{(4)}{>}~1-\alpha ,
\end{align*}
as desired, where (1) holds by \eqref{eq:CI1_eq2} with $c=F_{1}(\delta _{k_{M}},\sigma _{k_{M},l},\sigma _{k_{M},u})$, (2) by $F_{1}(\delta _{k_{M}},\sigma _{k_{M},l},\sigma _{k_{M},u})-F_{1}(\delta ,\sigma _{l},\sigma _{u})>\Delta $, (3) by $F_{1}(\delta ,\sigma _{l},\sigma _{u})=\Phi ^{-1}(1-\alpha )$, and (4) by  \eqref{eq:CI1_eq2c}.

We now divide the rest of the argument into its parts.

\noindent\underline{Part (a).} Since $(\hat{\theta}_{l},\hat{\theta}_{u},\hat{\sigma} _{l},\hat{\sigma}_{u},\hat{\rho})$ satisfies OBS with parameter $({\theta } _{l}(P),{\theta }_{u}(P),{\sigma }_{l}(P),{\sigma }_{u}(P),{\rho }(P))$ and \eqref{eq:C1_limit_1},
\begin{equation}
(\sqrt{N}(\hat{\theta}_{l}-\theta _{l}(P_{N})),\sqrt{N}(\hat{\theta} _{u}-\theta _{u}(P_{N})),\hat{\sigma}_{l},\hat{\sigma}_{u})~\overset{d}{ \to }~(z_{1}\sigma _{l},(z_{2}\sqrt{1-\rho ^{2}}+z_{1}\rho )\sigma _{u},\sigma _{l},\sigma _{u}),\label{eq:C1_limit_2}
\end{equation}
where $(z_{1},z_{2})\sim \mathcal{N}(\mathbf{0}_{2\times 1},\mathbf{I}_{2\times 2})$.
Then,
\begin{align*}
&P_{N}(\theta _{N} \in CI_{\alpha}^{1}) \\
& \overset{(1)}{=}~P_{N}\left( 
\begin{array}{c}
\sqrt{N}(\hat{\theta}_{l}-\theta _{l}(P_{N}))-\hat{\sigma}_{l}F_{1}(\sqrt{N}( \hat{\theta}_{u}-\hat{\theta}_{l}),\hat{\sigma}_{l},\hat{\sigma}_{u})\leq \sqrt{N}(\theta _{N}-\theta _{l}(P_{N}))\leq \\ 
\sqrt{N}(\hat{\theta}_{u}-\theta _{u}(P_{N}))+\sqrt{N}(\theta _{u}(P_{N})-\theta _{l}(P_{N}))+\hat{\sigma}_{u}F_{1}(\sqrt{N}(\hat{\theta} _{u}-\hat{\theta}_{l}),\hat{\sigma}_{l},\hat{\sigma}_{u})
\end{array}
\right) \\
& \overset{(2)}{\to}~P\left( 
\begin{array}{c}
z_{1}\sigma _{l}-\sigma _{l}F_{1}(z_{2}\sigma _{u}\sqrt{1-\rho ^{2}} +z_{1}(\rho \sigma _{u}-\sigma _{l})+\mu ,\sigma _{l},\sigma _{u})\leq -\Psi _{l} \leq \\ 
(z_{2}\sqrt{1-\rho ^{2}}+z_{1}\rho )\sigma _{u}+\mu +\sigma _{u}F_{1}(z_{2}\sigma _{u}\sqrt{1-\rho ^{2}}+z_{1}(\rho \sigma _{u}-\sigma _{l})+\mu ,\sigma _{l},\sigma _{u})
\end{array}
\right) \\
& \overset{(3)}{=}~P\left( 
\begin{array}{c}
\left\{ z_{1}-F_{1}(z_{2}\sigma _{u}\sqrt{1-\rho ^{2}}+z_{1}(\rho \sigma _{u}-\sigma _{l})+\mu ,\sigma _{l},\sigma _{u})\leq -\frac{\Psi _{l} }{\sigma _{l} }\right\} \cap \\ 
\left\{ -\frac{\Psi _{l} +\mu }{\sigma _{u}}\leq z_{2}\sqrt{1-\rho ^{2}} +z_{1}\rho +F_{1}(z_{2}\sigma _{u}\sqrt{1-\rho ^{2}}+z_{1}(\rho \sigma _{u}-\sigma _{l})+\mu ,\sigma _{l},\sigma _{u})\right\}
\end{array}
\right) ,
\end{align*}
as desired, where (1) holds by \eqref{eq:CI1} and \eqref{eq:CI1_problem}, which give ${c}^1=F_{1}( \sqrt{N}(\hat{\theta}_{u}-\hat{\theta}_{l}),\hat{\sigma}_{l},\hat{\sigma} _{u})$, (2) by \eqref{eq:C1_limit_1}, \eqref{eq:C1_limit_2}, the continuity of $F_{1}(\delta ,\sigma _{l},\sigma _{u})$, and the CMT, and (3) by $\sigma _{l},\sigma _{u}\geq \underline{\sigma }>0$.

\noindent\underline{Part (b).} This argument is analogous to part (a), and it is thus omitted.
\end{proof}

\begin{lemma}[$CI_{\alpha}^{2}$]\label{lem:C2_limit}
Assume $\alpha \in (0,0.5)$ and that $(\hat{\theta}_{l},\hat{\theta}_{u},\hat{\sigma}_{l},\hat{\sigma} _{u},\hat{\rho})$ satisfies OBS with parameter $({\theta }_{l}(P),{\theta } _{u}(P),{\sigma }_{l}(P),{\sigma }_{u}(P),{\rho }(P))$ and set $\mathcal{P}$.
Consider any sequence $\{(P_{N},\theta _{N})\in \mathcal{P}\times \Theta _{I}(P_{N})^{c}\}_{N\in \mathbb{N}}$ with
\begin{align}
&\left( 
\begin{array}{c}
\theta _{l}(P_{N}),\theta _{u}(P_{N}),\sigma _{l}(P_{N}),\sigma _{u}(P_{N}),\rho (P_{N}),\\
\sqrt{N}(\theta _{u}(P_{N})-\theta _{l}(P_{N})), \sqrt{N}(\theta _{l}(P_{N})-\theta _{N}),\sqrt{N}(\theta _{N}-\theta _{u}(P_{N}))
\end{array}
\right)  \notag \\
&  \to (\theta _{l},\theta _{u},\sigma _{l},\sigma _{u},\rho ,\mu ,\Psi _{l},\Psi _{u})\in \bar{\mathbb{R}}\times \bar{\mathbb{R}}\times [\underline{\sigma },\overline{\sigma }]\times [\underline{ \sigma },\overline{\sigma }]\times [-1,1]\times \bar{\mathbb{R}} _{+}\times \bar{\mathbb{R}}\times \bar{\mathbb{R}}. \label{eq:C2_limit_1}
\end{align}
Under these conditions, we have the following results.

\begin{enumerate}[(a)]
\item If $\mu =\infty$ and $\Psi _{l} \geq 0$,
\begin{equation*}
\lim_{N  \to \infty} P_{N}(\theta _{N}\in CI_{\alpha}^{2})~=~\Phi (\Phi ^{-1}(1-\alpha )-\Psi _{l} /\sigma _{l}).
\end{equation*}

\item If $\mu =\infty$ and $\Psi _{u} \geq 0$,
\begin{equation*}
\lim_{N  \to \infty} P_{N}(\theta _{N}\in CI_{\alpha}^{2})~=~\Phi (\Phi ^{-1}(1-\alpha )-\Psi _{u} /\sigma _{u}).
\end{equation*}

\item If $\mu \in \mathbb{R}_+$ and $\Psi _{l} \geq 0$, then $\sigma _l = \sigma _u = \sigma$, $\rho = 1$, and 
\begin{equation*}
\lim_{N  \to \infty} P_{N}(\theta _{N}\in CI_{\alpha}^{2})~=~\Phi \left( (\Psi _{l} +\mu )/\sigma +G(\mu /\sigma )\right) -\Phi \left( \Psi_{l} /\sigma -G(\mu /\sigma )\right) ,
\end{equation*}
where ${G}(y):\mathbb{R}_{+}  \to \mathbb{R}_{++}$ is as defined in Lemma \ref{lem:tildeH}.

\item If $\mu \in \mathbb{R}_+$ and $\Psi _{u} \geq 0$, then $\sigma _l = \sigma _u = \sigma$, $\rho = 1$,  and 
\begin{equation*}
\lim_{N  \to \infty} P_{N}(\theta _{N}\in CI_{\alpha}^{2})~=~\Phi \left( (\Psi _{u} +\mu )/\sigma +G(\mu /\sigma )\right) -\Phi \left( \Psi _{u} /\sigma -G(\mu /\sigma )\right) ,
\end{equation*}
where ${G}(y):\mathbb{R}_{+}  \to \mathbb{R}_{++}$ is as defined in Lemma \ref{lem:tildeH}.
\end{enumerate}
\end{lemma}
\begin{proof}
Let $C_{2}:\bar{\mathbb{R}}_{+}\times [\underline{\sigma },\overline{ \sigma }]\times [\underline{\sigma },\overline{\sigma }]\times
[-1,1]  \rightrightarrows {\bar{\mathbb{R}}^{2}}$ be the following correspondence:
\begin{equation}
C_{2}(\delta ,\sigma _{l},\sigma _{u},\rho )~=~\left\{ 
\begin{array}{l}
( c_{l},c_{u}) \in {\bar{\mathbb{R}}^{2}}~~~\text{ s.t.} \\ 
P(\{-c_{l}\leq z_{1}\}\cap \{\rho z_{1}\leq c_{u}+{\delta }/{\sigma _{u}}+\sqrt{1-\rho ^{2}}z_{2}\})~\geq ~1-\alpha ~\text{ and} \\ 
P(\{-c_{l}-{\delta }/{\sigma _{l}}+\sqrt{1-\rho ^{2}}z_{2}\leq \rho z_{1}\}\cap\{z_{1}\leq c_{u}\})~\geq ~1-\alpha ,
\end{array}
\right\} , 
\label{eq:C2_defn}
\end{equation}
where $(z_{1},z_{2})\sim \mathcal{N}(\mathbf{0}_{2\times 1},\mathbf{I}_{2\times 2})$, let $S_2:\bar{\mathbb{R}}_{+}\times [\underline{\sigma },\overline{ \sigma }]\times [\underline{\sigma },\overline{\sigma }]\times [-1,1]  \rightrightarrows {\bar{\mathbb{R}}^{2}}$ be the following minimizer correspondence: 
\begin{equation}
S_{2}(\delta ,\sigma _{l},\sigma _{u},\rho )~=~\underset{( c_{l},c_{u}) \in C_{2}(\delta ,\sigma _{l},\sigma _{u},\rho )}{\text {arg min} }(\sigma _{l}c_{l}+\sigma _{u}c_{u}), 
\label{eq:S2_defn}
\end{equation}
and let $F_{2}=(F_{2,l},F_{2,u}):\bar{\mathbb{R}}_{+}\times [ \underline{\sigma },\overline{\sigma }]\times [\underline{\sigma }, \overline{\sigma }]\times [-1,1]  \to {\bar{\mathbb{R}}^{2}}$ be the choice in $S_{2}$ used to implement $CI_{\alpha}^2$. By \eqref{eq:CI2} and \eqref{eq:CI2_problem}, we have that $c^2=F_{2}(\sqrt{N}(\hat{\theta}_{u}-\hat{\theta} _{l}),\hat{\sigma}_{l},\hat{\sigma}_{u},\hat{\rho})$.

We now study the continuity properties of $F_{2}$. For our purposes, it suffices to determine the limit of $\{F_{2}(\delta _{M},\sigma _{M,l},\sigma _{M,u},\rho _{M})\}_{M\in \mathbb{N} }$ for two types of sequences $\{(\delta _{M},\sigma _{M,l},\sigma _{M,u},\rho _{M})\}_{M\in \mathbb{N} }$: (i) those converging to $(\infty ,\sigma _{l},\sigma _{u},\rho )$ with $ (\sigma _{l},\sigma _{u},\rho )\in [\underline{\sigma },\overline{ \sigma }]\times [\underline{\sigma },\overline{\sigma }]\times [-1,1]$ and (ii) those converging to $(\delta ,\sigma ,\sigma ,1)$ with $(\delta ,\sigma )\in \bar{\mathbb{R}}_{+}\times [\underline{ \sigma },\overline{\sigma }]$.

First, consider $\{F_{2}(\delta _{M},\sigma _{M,l},\sigma _{M,u},\rho _{M})\}_{M\in \mathbb{N} }$ for $\lim_{M\to \infty}  (\delta _{M},\sigma _{M,l},\sigma _{M,u},\rho _{M}) = (\infty ,\sigma _{l},\sigma _{u},\rho )$ and $(\sigma _{l},\sigma _{u},\rho )\in [\underline{\sigma },\overline{\sigma } ]\times [\underline{\sigma },\overline{\sigma }]\times [-1,1]$. It is not hard to verify that 
\begin{align*}
C_{2}(\infty ,\sigma _{l},\sigma _{u},\rho )& ~=~\left\{ \left( c_{l},c_{u}\right) \in {\bar{\mathbb{R}}^{2}}~~\text{s.t.}~~~\Phi (c_{l})\geq 1-\alpha ~~\text{and}~~\Phi (c_{u})\geq 1-\alpha \right\}, \\
S_{2}(\infty ,\sigma _{l},\sigma _{u},\rho )& ~=~\{(\Phi ^{-1}(1-\alpha ),\Phi ^{-1}(1-\alpha ))\},\\
F_{2}(\infty ,\sigma _{l},\sigma _{u},\rho )& ~=~(\Phi ^{-1}(1-\alpha ),\Phi ^{-1}(1-\alpha )).
\end{align*}
Furthermore, part (a) of Lemma \ref{lem:C2_is_continuous} shows that $C_{2}$ is continuous at $(\delta ,\sigma _{l},\sigma _{u},\rho )=(\infty ,\sigma _{l},\sigma _{u},\rho )$. Then, Lemma \ref{lem:Berge} shows that $F_{2}$ is continuous at $(\delta ,\sigma _{l},\sigma _{u},\rho )=(\infty ,\sigma _{l},\sigma _{u},\rho )$, i.e., 
\begin{equation}
\underset{(\delta _{M},\sigma _{M,l},\sigma _{M,u},\rho _{M})  \to (\infty ,\sigma _{l},\sigma _{u},\rho )}{\text{lim} } F_{2}(\delta _{M},\sigma _{M,l},\sigma _{M,u},\rho _{M})~=~(\Phi ^{-1}(1-\alpha ),\Phi ^{-1}(1-\alpha )).
\label{eq:C2_limit_2}
\end{equation}

Second, consider $\{F_{2}(\delta _{M},\sigma _{M,l},\sigma _{M,u},\rho _{M})\}_{M\in \mathbb{N} }$ for $\lim_{M\to \infty}  (\delta _{M},\sigma _{M,l},\sigma _{M,u},\rho _{M}) = (\delta ,\sigma ,\sigma ,1)$ and $(\delta ,\sigma )\in \bar{\mathbb{R}}_{+}\times [\underline{\sigma },\overline{\sigma }]$. It is not hard to verify that 
\begin{equation*}
C_{2}(\delta ,\sigma ,\sigma ,1)~=~\left\{ 
\begin{array}{l}
\left( c_{l},c_{u}\right) \in {\bar{\mathbb{R}}^{2}}~~~\text{ s.t.} \\ 
\Phi (c_{u}+{\delta }/{\sigma })-\Phi (-c_{l})~\geq ~1-\alpha ~~\text{and}~~\Phi (c_{u})-\Phi (-c_{l}-{\delta }/{\sigma })~\geq ~1-\alpha  
\end{array}
\right\} .
\end{equation*}
In turn, Lemma \ref{lem:criticalMin} shows that 
\begin{align*}
S_{2}(\delta ,\sigma ,\sigma ,1)~=\{(G(\delta /\sigma ),G(\delta /\sigma ))\}~\text{ and }~F_{2}(\delta ,\sigma ,\sigma ,1)~=~(G(\delta /\sigma ),G(\delta /\sigma )).
\end{align*}
Furthermore, part (b) of Lemma \ref{lem:C2_is_continuous} shows that $C_{2}$ is continuous at $(\delta ,\sigma _{l},\sigma _{u},\rho )=(\delta ,\sigma ,\sigma ,1)$. Under these conditions, Lemma \ref{lem:Berge} shows that $F_{2}$ is continuous at $(\delta ,\sigma _{l},\sigma _{u},\rho )=(\delta ,\sigma ,\sigma ,1)$, i.e., 
\begin{equation}
\underset{(\delta _{M},\sigma _{M,l},\sigma _{M,u},\rho _{M})  \to (\delta ,\sigma ,\sigma ,1)}{\text{lim} }F_{2}(\delta _{M},\sigma _{M,l},\sigma _{M,u},\rho _{M})~=~(G(\delta /\sigma ),G(\delta /\sigma )).
\label{eq:C2_limit_3}
\end{equation}

We now divide the rest of the argument into its parts.

\noindent\underline{Part (a).} Since $(\hat{\theta}_{l},\hat{\theta}_{u},\hat{\sigma} _{l},\hat{\sigma}_{u},\hat{\rho})$ satisfies OBS with parameter $({\theta } _{l}(P),{\theta }_{u}(P),{\sigma }_{l}(P),{\sigma }_{u}(P),{\rho }(P))$ and \eqref{eq:C2_limit_1},
\begin{equation}
(\sqrt{N}(\hat{\theta}_{l}-\theta _{l}(P_{N})),\sqrt{N}(\hat{\theta} _{u}-\theta _{u}(P_{N})),\hat{\sigma}_{l},\hat{\sigma}_{u}, \hat{\rho})~\overset{d}{   \to }~(z_{1}\sigma _{l},(z_{2}\sqrt{1-\rho ^{2}}+z_{1}\rho )\sigma _{u},\sigma _{l},\sigma _{u},\rho),
\label{eq:C2_limit_4}
\end{equation}
where $(z_{1},z_{2})\sim \mathcal{N}(\mathbf{0}_{2\times 1},\mathbf{I}_{2\times 2})$.
Then,
\begin{align*}
& P_{N}(\theta _{N} \in CI_{\alpha}^{2}) \\
& \overset{(1)}{=}~P_{N}\left( 
\begin{array}{c}
\sqrt{N}(\hat{\theta}_{l}-\theta _{l}(P_{N}))-\hat{\sigma}_{l} F_{2,l}(\sqrt{N }(\hat{\theta}_{u}-\hat{\theta}_{l}),\hat{\sigma}_{l},\hat{\sigma}_{u},\hat{\rho}) \leq \sqrt{N}(\theta _{N}-\theta _{l}(P_{N}))\leq  \\ 
\sqrt{N}(\hat{\theta}_{u}-\theta _{u}(P_{N}))+\sqrt{N}(\theta_{u}(P_{N})-\theta _{l}(P_{N}))+\hat{\sigma}_{u}F_{2,u}(\sqrt{N}(\hat{\theta}_{u}-\hat{\theta}_{l}),\hat{\sigma}_{l},\hat{\sigma}_{u},\hat{\rho})
\end{array}
\right)  \\
& \overset{(2)}{\to}~P(z_{1}\sigma _{l}-\sigma _{l}\Phi ^{-1}(1-\alpha )\leq -\Psi_{l} ) \\
& \overset{(3)}{=}~\Phi (\Phi ^{-1}(1-\alpha )-\Psi _{l}/\sigma _{l}),
\end{align*}
as desired, where (1) holds by \eqref{eq:CI2} and \eqref{eq:CI2_problem}, which give ${c}^2=F_{2}(\sqrt{N}(\hat{\theta}_{u}-\hat{\theta}_{l}),\hat{ \sigma}_{l},\hat{\sigma}_{u},\hat{\rho})$, (2) by \eqref{eq:C2_limit_1}, \eqref{eq:C2_limit_2}, \eqref{eq:C2_limit_4}, $\mu = \infty$, and the CMT, and (3) by $\sigma _{l},\sigma _{u}\geq \underline{\sigma }>0$ and $z_{1}\sim \mathcal{N}(0,1)$.

\noindent\underline{Part (b).} This argument is analogous to part (a) and it is therefore
omitted.

\noindent\underline{Part (c).} Since $(\hat{\theta}_{l},\hat{\theta}_{u},\hat{\sigma} _{l},\hat{\sigma}_{u},\hat{\rho})$ satisfies OBS with parameter $({\theta } _{l}(P),{\theta }_{u}(P),{\sigma }_{l}(P),{\sigma }_{u}(P),{\rho }(P))$ and \eqref{eq:C2_limit_1}, \eqref{eq:C2_limit_4} holds. Moreover, by Lemma \ref{lem:near1}, $\sigma _{l}=\sigma _{u}$ and $\rho =1$. We set $\sigma=\sigma _{l}=\sigma _{u}$. Then,
\begin{align*}
&P_{N}(\theta _{N} \in CI_{\alpha}^{2}) \\
& \overset{(1)}{=}~P_{N}\left( 
\begin{array}{c}
\sqrt{N}(\hat{\theta}_{l}-\theta _{l}(P_{N}))-\hat{\sigma}_{l}F_{2,l}(\sqrt{N }(\hat{\theta}_{u}-\hat{\theta}_{l}),\hat{\sigma}_{l},\hat{\sigma}_{u},\hat{ \rho}) \leq \sqrt{N}(\theta _{N}-\theta _{l}(P_{N}))\leq  \\ 
\sqrt{N}(\hat{\theta}_{u}-\theta _{u}(P_{N}))+\sqrt{N}(\theta _{u}(P_{N})-\theta _{l}(P_{N}))+\hat{\sigma}_{u}F_{2,u}(\sqrt{N}(\hat{\theta} _{u}-\hat{\theta}_{l}),\hat{\sigma}_{l},\hat{\sigma}_{u},\hat{\rho})
\end{array}
\right)  \\
& \overset{(2)}{\to}~P\left( \sigma z_{1}-\sigma G(\mu /\sigma )\leq -\Psi_{l} \leq z_{1}\sigma +\mu +\sigma G(\mu /\sigma )\right)  \\
& \overset{(3)}{=}~\Phi \left( (\Psi_{l} +\mu )/\sigma +G(\mu /\sigma )\right) -\Phi \left( \Psi_{l} /\sigma -G(\mu /\sigma )\right) ,
\end{align*}
as desired, where (1) holds by \eqref{eq:CI2} and \eqref{eq:CI2_problem}, which give ${c}^2=F_{2}(\sqrt{N}(\hat{\theta}_{u}-\hat{\theta}_{l}),\hat{ \sigma}_{l},\hat{\sigma}_{u},\hat{\rho})$, (2) by $\sigma=\sigma _{l}=\sigma _{u}$, $\rho =1$, \eqref{eq:C2_limit_1}, \eqref{eq:C2_limit_3}, \eqref{eq:C2_limit_4}, and the CMT, and (3) by $\sigma \geq \underline{\sigma }>0$ and $z_{1}\sim \mathcal{N}(0,1)$.

\noindent\underline{Part (d).} This argument is analogous to part (c), and it is thus omitted.
\end{proof}

\begin{lemma}[$CI_{\alpha }^{3}$]\label{lem:C3_limit} 
Assume $\alpha \in (0,0.5)$ and that $(\hat{\theta}_{l}, \hat{\theta}_{u},\hat{\sigma}_{l},\hat{\sigma}_{u},\hat{\rho})$ satisfies OBS with parameter $({\theta }_{l}(P),{\theta }_{u}(P),{\sigma }_{l}(P),{ \sigma }_{u}(P),{\rho }(P))$ and set $\mathcal{P}$. Consider any sequence $ \{(P_{N},\theta _{N})\in \mathcal{P}\times \Theta _{I}(P_{N})^{c}\}_{N\in \mathbb{N}}$ with 
\begin{align}
& \left( 
\begin{array}{c}
\theta _{l}(P_{N}),\theta _{u}(P_{N}),\sigma _{l}(P_{N}),\sigma _{u}(P_{N}),\rho (P_{N}),\sqrt{N}(\theta _{u}(P_{N})-\theta _{l}(P_{N})),\\
\sqrt{N}(\theta _{u}(P_{N})-\theta _{l}(P_{N})-b_{N}),\sqrt{N}(\theta _{l}(P_{N})-\theta _{N}),\sqrt{N}(\theta _{N}-\theta _{u}(P_{N})), \\ 
\end{array}
\right) ~\to   \notag \\
& (\theta _{l},\theta _{u},\sigma _{l},\sigma _{u},\rho ,\mu ,\eta ,\Psi _{l},\Psi _{u}) ~\in ~\bar{\mathbb{R}}\times \bar{\mathbb{R}}\times [\underline{\sigma },\overline{\sigma }]\times [\underline{\sigma },\overline{\sigma }]\times [-1,1]\times \bar{\mathbb{R}}_{+}\times \bar{\mathbb{R}}\times \bar{\mathbb{R }}\times \bar{\mathbb{R}}. \label{eq:C3_limit_1}
\end{align}
Finally, assume that $\{b_N\}_{N\in \mathbb{N}}$ satisfies $b_N \to 0$ and $b_N \sqrt{N} \to\infty$. Under these conditions, we have the following results.
\begin{enumerate}[(a)]

\item If $\mu =\infty $ and $\Psi _{l}\geq 0$, then 
\begin{equation*}
\underset{N \to \infty }{\lim \inf} ~P_{N}(\theta _{N}\in CI_{\alpha }^{3})~\geq~\Phi (\Phi ^{-1}(1-\alpha )-\Psi _{l}/\sigma _{l}).
\end{equation*}

\item If $\mu =\infty $, $\eta =0$, and $\Psi _{l}\geq 0$, then $\sigma _{l}=\sigma _{u}=\sigma $, $\rho =1$, and 
\begin{equation*}
\underset{N \to \infty }{\lim \sup} ~P_{N}(\theta _{N}\in CI_{\alpha }^{3})~\leq~\Phi (\Phi ^{-1}(1-\alpha/2 )-\Psi _{l}/\sigma).
\end{equation*}

\item If $\mu =\infty $, $\eta \in (0,\infty ]$, and $\Psi _{l}\geq 0$,
\begin{equation*}
\lim_{N \to \infty }P_{N}(\theta _{N}\in CI_{\alpha }^{3})~=~\Phi (\Phi ^{-1}(1-\alpha )-\Psi _{l}/\sigma _{l}).
\end{equation*}

\item If $\mu =\infty $, $\eta \in [-\infty ,0)$, and $\Psi _{l}\geq 0 $, then $\sigma _{l}=\sigma _{u}=\sigma $, $\rho =1$, and
\begin{equation*}
\lim_{N \to \infty }P_{N}(\theta _{N}\in CI_{\alpha }^{3})~=~\Phi (\Phi ^{-1}(1-\alpha /2)-\Psi _{l}/\sigma ).
\end{equation*}

\item If $\mu \in \mathbb{R}_{+}$ and $\Psi _{l}\geq 0$, then $\sigma _{l}=\sigma _{u}=\sigma $, $\rho =1$, and
\begin{equation*}
    \lim_{N \to \infty }P_{N}(\theta _{N}\in CI_{\alpha }^{3})~=~\Phi \left( (\Psi _{l}+\mu )/\sigma +\Phi ^{-1}(1-\alpha /2)\right) -\Phi \left( \Psi _{l}/\sigma -\Phi ^{-1}(1-\alpha /2)\right),
\end{equation*}
 
\item If $\mu =\infty $ and $\Psi _{u}\geq 0$, then 
\begin{equation*}
\underset{N \to \infty }{\lim \inf} ~P_{N}(\theta _{N}\in CI_{\alpha }^{3})~\geq~\Phi (\Phi ^{-1}(1-\alpha )-\Psi _{u}/\sigma _{u}).
\end{equation*}

\item If $\mu =\infty $, $\eta =0$, and $\Psi _{u}\geq 0$, then $\sigma _{l}=\sigma _{u}=\sigma $, $\rho =1$, and 
\begin{equation*}
\underset{N \to \infty }{\lim \sup} ~P_{N}(\theta _{N}\in CI_{\alpha }^{3})~\leq~\Phi (\Phi ^{-1}(1-\alpha/2 )-\Psi _{u}/\sigma).
\end{equation*}

\item If $\mu =\infty $, $\eta \in (0,\infty ]$, and $\Psi _{u}\geq 0$,
\begin{equation*}
\lim_{N \to \infty }P_{N}(\theta _{N}\in CI_{\alpha }^{3})~=~\Phi (\Phi ^{-1}(1-\alpha )-\Psi _{u}/\sigma _{u}).
\end{equation*}

\item If $\mu =\infty $, $\eta \in [-\infty ,0)$, and $\Psi _{u}\geq 0 $, then $\sigma _{l}=\sigma _{u}=\sigma $, $\rho =1$, and
\begin{equation*}
\lim_{N \to \infty }P_{N}(\theta _{N}\in CI_{\alpha }^{3})~=~\Phi (\Phi ^{-1}(1-\alpha /2)-\Psi _{u}/\sigma ).
\end{equation*}

\item If $\mu \in \mathbb{R}_{+}$ and $\Psi _{u}\geq 0$, then $\sigma _{l}=\sigma _{u}=\sigma $, $\rho =1$, and 
\begin{equation*}
    \lim_{N \to \infty }P_{N}(\theta _{N}\in CI_{\alpha }^{3})~=~\Phi \left( (\Psi _{u}+\mu )/\sigma +\Phi ^{-1}(1-\alpha /2)\right) -\Phi \left( \Psi _{u}/\sigma -\Phi ^{-1}(1-\alpha /2)\right),
\end{equation*}
\end{enumerate}
\end{lemma}
\begin{proof}
Let $C_{3}:\bar{\mathbb{R}}_{+}\times \bar{\mathbb{R}}\times [ \underline{\sigma },\overline{\sigma }]\times [\underline{\sigma }, \overline{\sigma }]\times [-1,1] \rightrightarrows{\bar{\mathbb{R}}^{2}}$ be the following correspondence:
\begin{align}
& C_{3}(\delta _{1},\delta _{2},\sigma _{l},\sigma _{u},\rho )  \notag \\
& =~\left\{ 
\begin{array}{l}
\left( c_{l},c_{u}\right) \in {\bar{\mathbb{R}}^{2}}~~~\text{ s.t.} \\ 
P(\{-c_{l}\leq z_{1}\}~\cap ~\{\rho z_{1}\leq c_{u}+{\delta _{1}1[\delta _{2}>0]}/{\sigma _{u}}+\sqrt{1-\rho ^{2}}z_{2}\})~\geq ~1-\alpha ~\text{and}\\ 
P(\{-c_{l}-{\delta _{1}1[\delta _{2}>0]}/{\sigma _{l}}+\sqrt{1-\rho ^{2}} z_{2}\leq \rho z_{1}\}~\cap ~\{z_{1}\leq c_{u}\})~\geq ~1-\alpha ,
\end{array}
\right\} ,  \label{eq:C3_defn}
\end{align}
where $(z_{1},z_{2})\sim \mathcal{N}(\mathbf{0}_{2\times 1},\mathbf{I}_{2\times 2})$, let $S_{3}:\bar{\mathbb{R}}_{+}\times \bar{\mathbb{R}}\times [ \underline{\sigma },\overline{\sigma }]\times [\underline{\sigma }, \overline{\sigma }]\times [-1,1] \rightrightarrows {\bar{\mathbb{R}}^{2}}$ be the following minimizer correspondence: 
\begin{equation}
S_{3}(\delta _{1},\delta _{2},\sigma _{l},\sigma _{u},\rho )~=~\underset{ \left( c_{l},c_{u}\right) \in C_{3}(\delta _{1},\delta _{2},\sigma _{l},\sigma _{u},\rho )}{\text{arg min}}(\sigma _{l}c_{l}+\sigma _{u}c_{u}), \label{eq:S3_defn}
\end{equation}
and let $F_{3}=(F_{3,l},F_{3,u}):\bar{\mathbb{R}}_{+}\times \bar{\mathbb{R}} \times [\underline{\sigma },\overline{\sigma }]\times [ \underline{\sigma },\overline{\sigma }]\times [-1,1] \to {\bar{ \mathbb{R}}^{2}}$ be the choice in $S_{3}(\delta _{1},\delta _{2},\sigma _{l},\sigma _{u},\rho )$ used to implement $CI_{\alpha}^3$. By \eqref{eq:CI3}, \eqref{eq:CI3_problem}, and Lemma \ref{lem:C3NotEmpty}, we have that $ c^{3}=F_{3}(\sqrt{N}(\hat{\theta}_{u}-\hat{\theta}_{l}),\sqrt{N}(\hat{\theta}_{u}-\hat{\theta}_{l})-\sqrt{N}b_{N},\hat{\sigma}_{l},\hat{\sigma}_{u},\hat{ \rho})$ a.s.

\noindent \underline{Part (a).} We first show that $( c_{l},c_{u}) \in C_{3}(\delta _{1},\delta _{2},\sigma _{l},\sigma _{u},\rho )$ satisfies $c_{l},c_{u}\geq \Phi ^{-1}(1-\alpha )>0$. To show this for $c_{l}$, note that the first constraint in \eqref{eq:C3_defn} implies that 
\begin{equation*}
P(-c_{l}\leq z_{1})~\geq~ P(\{-c_{l}\leq z_{1}\}~\cap ~\{\rho z_{1}\leq c_{u}+ {\delta _{1}1[\delta _{2}>0]}/{\sigma _{u}}+\sqrt{1-\rho ^{2}} z_{2}\})~\geq ~1-\alpha ,
\end{equation*}
which implies $c_{l}\geq \Phi ^{-1}(1-\alpha )$. In turn, note that $\Phi ^{-1}(1-\alpha )>0$ by $\alpha \in (0,0.5)$. A similar argument using the second constraint in \eqref{eq:C3_defn} implies that $c_{u}\geq \Phi ^{-1}(1-\alpha )>0$. Since $F_{3}(\delta _{1},\delta _{2},\sigma _{l},\sigma _{u},\rho )\in S_{3}(\delta _{1},\delta _{2},\sigma _{l},\sigma _{u},\rho )\subseteq C_{3}(\delta _{1},\delta _{2},\sigma _{l},\sigma _{u},\rho )$,
\begin{align}
F_{3,l}(\sqrt{N}(\hat{\theta}_{u}-\hat{\theta}_{l}), \sqrt{N}(\hat{\theta}_{u}-\hat{\theta}_{l})-\sqrt{N}b_{N},\hat{\sigma}_{l}, \hat{\sigma}_{u},\hat{\rho})~&\geq~ \Phi ^{-1}(1-\alpha )\notag\\
F_{3,u}(\sqrt{N}(\hat{\theta}_{u}-\hat{\theta}_{l}), \sqrt{N}(\hat{\theta}_{u}-\hat{\theta}_{l})-\sqrt{N}b_{N},\hat{\sigma}_{l}, \hat{\sigma}_{u},\hat{\rho})~&\geq~ \Phi ^{-1}(1-\alpha ).
\label{eq:C3_LB}
\end{align}

Since $(\hat{\theta}_{l},\hat{\theta}_{u},\hat{\sigma} _{l},\hat{\sigma}_{u},\hat{\rho})$ satisfies OBS with parameter $({\theta } _{l}(P),{\theta }_{u}(P),{\sigma }_{l}(P),{\sigma }_{u}(P),{\rho }(P))$ and \eqref{eq:C3_limit_1},
\begin{equation}
(\sqrt{N}(\hat{\theta}_{l}-\theta _{l}(P_{N})),\sqrt{N}(\hat{\theta} _{u}-\theta _{u}(P_{N})),\hat{\sigma}_{l},\hat{\sigma}_{u}, \hat{\rho})~\overset{d}{\to}~(z_{1}\sigma _{l},(z_{2}\sqrt{1-\rho ^{2}}+z_{1}\rho )\sigma _{u},\sigma _{l},\sigma _{u},\rho),
\label{eq:C3_limit_3}
\end{equation}
where $(z_{1},z_{2})\sim \mathcal{N}(\mathbf{0}_{2\times 1},\mathbf{I}_{2\times 2})$. Then,
\begin{align}
&P_{N}(\theta _{N} \in CI_{\alpha}^{3}) \notag\\
& ~\overset{(1)}{=}~P_{N}\left( 
\begin{array}{c}
\sqrt{N}(\hat{\theta}_{l}-\theta _{l}(P_{N}))-\hat{\sigma}_{l}F_{3,l}(\sqrt{N}(\hat{\theta}_{u}-\hat{\theta}_{l}),\sqrt{N}(\hat{\theta}_{u}-\hat{\theta} _{l}-b_{N}),\hat{\sigma}_{l},\hat{\sigma}_{u},\hat{\rho}) \\ 
\leq \sqrt{N}(\theta _{N}-\theta _{l}(P_{N}))\leq  \sqrt{N}(\hat{\theta}_{u}-\theta _{u}(P_{N}))+\sqrt{N}(\theta _{u}(P_{N})-\theta _{l}(P_{N}))+\\
\hat{\sigma}_{u}F_{3,u}(\sqrt{N}(\hat{\theta} _{u}-\hat{\theta}_{l}),\sqrt{N}(\hat{\theta}_{u}-\hat{\theta} _{l}-b_{N}),\hat{\sigma}_{l},\hat{\sigma}_{u},\hat{\rho})
\end{array}
\right) \notag \\
& ~\overset{(2)}{\geq }~P_{N}\left( 
\begin{array}{c}
\sqrt{N}(\hat{\theta}_{l}-\theta _{l}(P_{N}))-\hat{\sigma}_{l}\Phi ^{-1}(1-\alpha )\leq \sqrt{N}(\theta _{N}-\theta _{l}(P_{N}))\leq  \\ 
\sqrt{N}(\hat{\theta}_{u}-\theta _{u}(P_{N}))+\sqrt{N}(\theta _{u}(P_{N})-\theta _{l}(P_{N}))+\\
\hat{\sigma}_{u}F_{3,u}(\sqrt{N}(\hat{\theta} _{u}-\hat{\theta}_{l}),\sqrt{N}(\hat{\theta}_{u}-\hat{\theta} _{l}-b_{N}),\hat{\sigma}_{l},\hat{\sigma}_{u},\hat{\rho})
\end{array}
\right)  \notag\\
&\overset{(3)}{\to}~P\left( z_{1}\sigma _{l}-\sigma _{l}\Phi^{-1}(1-\alpha )\leq -\Psi_{l}\right)  \notag\\
&\overset{(4)}{=}~\Phi (\Phi ^{-1}(1-\alpha )-\Psi_{l} /\sigma _{l}),\label{eq:limit3_4}
\end{align}
where (1) holds by Lemma \ref{lem:C3NotEmpty}, \eqref{eq:CI3}, and \eqref{eq:CI3_problem}, which give ${c}^3=F_{3}(\sqrt{N}(\hat{\theta}_{u}-\hat{\theta}_{l}),\sqrt{N}(\hat{\theta}_{u}-\hat{\theta}_{l}) - \sqrt{N}b_{N},\hat{ \sigma}_{l},\hat{\sigma}_{u},\hat{\rho})$ a.s., (2) by \eqref{eq:C3_LB}, (3) by \eqref{eq:C3_limit_1}, \eqref{eq:C3_LB}, \eqref{eq:C3_limit_3}, $\mu = \infty$, and the CMT, and (4) by $\sigma_{l} \geq \underline{\sigma }>0$ and $z_{1}\sim \mathcal{N}(0,1)$. The desired result is implied by \eqref{eq:limit3_4}.

\noindent \underline{Part (b).} Note that $\sqrt{N}\left( \theta _{u}(P_{N})-\theta _{l}(P_{N})-b_{N}\right) \to \eta =0$ and $b_{N}\to 0$ implies $\theta _{u}(P_{N})-\theta _{l}(P_{N})\to 0$. Lemma \ref{lem:near1} then implies that \eqref{eq:C3_limit_1} holds with $(\sigma _{l},\sigma _{u},\rho )=(\sigma ,\sigma ,1)$ for some $\sigma \in [\underline{\sigma },\bar{\sigma}]$. 

We now derive the desired limiting result. It suffices to show that, for any $\varepsilon >0$,
\begin{equation}
\limsup_{N\to \infty }P_{N}(\theta _{N}\in CI_{\alpha }^{3})~\leq~ \Phi ( \Phi ^{-1}(1-\alpha /2)+\varepsilon -\Psi _{l}/\sigma ) .
\label{eq:aux0_lemma3b}
\end{equation}
To see why, note that the desired limit follows from \eqref{eq:aux0_lemma3b} as $\varepsilon \downarrow 0$. To this end, we fix $\varepsilon >0$ arbitrarily for the remainder of this part.

Since $(\hat{\theta}_{l}, \hat{\theta}_{u},\hat{\sigma}_{l},\hat{\sigma}_{u},\hat{\rho})$ satisfies OBS with parameter $({\theta }_{l}(P),{\theta }_{u}(P),{\sigma }_{l}(P),{ \sigma }_{u}(P),{\rho }(P))$ and \eqref{eq:C3_limit_1}, 
\begin{equation}
(\sqrt{N}(\hat{\theta}_{l}-\theta _{l}(P_{N})),\sqrt{N}(\hat{\theta} _{u}-\theta _{u}(P_{N})),\hat{\sigma}_{l},\hat{\sigma}_{u},\hat{\rho}) ~\overset{d}{ \to }~(z\sigma ,z\sigma ,\sigma ,\sigma ,1),
\label{eq:aux3_lemma3b}
\end{equation}
where $z \sim \mathcal{N}(0,1)$.

Next, we introduce some definitions. Define $A_{N}=\{(\hat{\theta}_{u}-\hat{\theta}_{l})>b_{N}\}$. Let $C_{0}:[-1,1]\rightrightarrows {\bar{\mathbb{R}}^{2}}$ be the following correspondence: 
\begin{equation*}
C_{0}(\rho )=~\left\{ 
\begin{array}{l}
\left( c_{l},c_{u}\right) \in {\bar{\mathbb{R}}^{2}}~~~\text{ s.t.} \\ 
P(\{-c_{l}\leq z_{1}\}~\cap ~\{\rho z_{1}\leq c_{u}+\sqrt{1-\rho ^{2}} z_{2}\})~\geq ~1-\alpha ~\text{and} \\ 
P(\{-c_{l}+\sqrt{1-\rho ^{2}}z_{2}\leq \rho z_{1}\}~\cap ~\{z_{1}\leq c_{u}\})~\geq ~1-\alpha ,
\end{array}
\right\} ,
\end{equation*}
where $(z_{1},z_{2})\sim \mathcal{N}(\mathbf{0}_{2\times 1},\mathbf{I}_{2\times 2})$, let $S_{0}:[\underline{\sigma },\overline{\sigma }]\times [\underline{\sigma },\overline{\sigma }]\times [-1,1]\rightrightarrows {\bar{\mathbb{R}} ^{2}}$ be the following minimizer correspondence: 
\begin{equation}
S_{0}(\sigma _{l},\sigma _{u},\rho )~=~\underset{\left( c_{l},c_{u}\right) \in C_{0}(\rho )}{\text{arg min}}(\sigma _{l}c_{l}+\sigma _{u}c_{u}),
\label{eq:S4_defn}
\end{equation}
and let $F_{0}=(F_{0,l},F_{0,u}):[\underline{\sigma },\overline{\sigma } ]\times [\underline{\sigma },\overline{\sigma }]\times [ -1,1]\to {\bar{\mathbb{R}}^{2}}$ be the choice in $S_{0}(\sigma _{l},\sigma _{u},\rho )$ used to implement $CI_{\alpha}^3$ when $A_{N}^c$ occurs. By repeating arguments in part (b) of Lemma \ref{lem:C_is_UHC_3}, it follows that $C_{0}$ is continuous at $\rho = 1$. Also note that $S_{0}(\sigma,\sigma ,1 ) = \{(\Phi^{-1}(1-\alpha/2),\Phi^{-1}(1-\alpha/2))\}$. Then, by repeating the arguments of Lemma \ref{lem:Berge2}, we obtain that $F_{0}$ is continuous at $(\sigma ,\sigma ,1)$ and $F_{0}(\sigma ,\sigma ,1) = (\Phi^{-1}(1-\alpha/2),\Phi^{-1}(1-\alpha/2))$. By this and \eqref{eq:aux3_lemma3b},
\begin{equation}
F_{0}(\hat\sigma_{l} ,\hat\sigma_{u} ,\hat{\rho})~\overset{p}{\to}~(\Phi^{-1}(1-\alpha/2),\Phi^{-1}(1-\alpha/2)).
\label{eq:aux3_lemma3c}
\end{equation}

Let $C_2$ be defined as in \eqref{eq:C2_defn}, $S_2$ be defined as in \eqref{eq:S2_defn}, and $F_{2}=(F_{2,l},F_{2,u}):\bar{\mathbb{R}}_{+}\times [ \underline{\sigma },\overline{\sigma }]\times [\underline{\sigma }, \overline{\sigma }]\times [-1,1]  \to {\bar{\mathbb{R}}^{2}}$ be the choice in $S_{2}$ used to implement $CI_{\alpha}^3$ when $A_{N}$ occurs. Recall from Lemma \ref{lem:C2_is_continuous} that $C_2$ in \eqref{eq:C2_defn} is continuous at $(\infty,\sigma ,\sigma ,1)$. Also, note from the definition in \eqref{eq:S2_defn} and Lemmas \ref{lem:tildeH} and \ref{lem:criticalMin}, we obtain $S_{2}(\infty,\sigma ,\sigma ,1) = \{(\Phi^{-1}(1-\alpha),\Phi^{-1}(1-\alpha))\}$. From this and Lemma \ref{lem:Berge}, it follows that $F_{2}$ is continuous at $(\infty,\sigma ,\sigma ,1)$ and $F_{2}(\infty ,\sigma ,\sigma ,1) = (\Phi^{-1}(1-\alpha),\Phi^{-1}(1-\alpha))$. By this and \eqref{eq:aux3_lemma3b},
\begin{equation}
F_{2}(\sqrt{N}(\hat{\theta}_{u}-\hat{\theta}_{l}), \hat\sigma_{l} ,\hat\sigma_{u} ,\hat{\rho})~\overset{p}{\to}~(\Phi^{-1}(1-\alpha),\Phi^{-1}(1-\alpha)).
\label{eq:aux3_lemma3d}
\end{equation}

Next, define $E_{N}=\{ F_{3,l}(\sqrt{N}(\hat{\theta}_{u}-\hat{\theta}_{l}),\sqrt{N}(\hat{\theta}_{u}-\hat{\theta}_{l}-b_{N}),\hat{\sigma}_{l},\hat{\sigma}_{u},\hat{\rho})>\Phi ^{-1}(1-\alpha /2)+\varepsilon \} $. By construction,
\begin{align}
&F_{3,l}(\sqrt{N}(\hat{\theta}_{u}-\hat{\theta}_{l}),\sqrt{N}(\hat{\theta} _{u}-\hat{\theta}_{l}-b_{N}),\hat{\sigma}_{l},\hat{\sigma}_{u},\hat{\rho}) \notag\\
&=F_{2,l}(\sqrt{N}(\hat{\theta}_{u}-\hat{\theta}_{l}),\hat{\sigma}_{l},\hat{ \sigma}_{u},\hat{\rho})\times I\left\{ A_{N}\right\} +F_{0,l}(\hat{\sigma} _{l},\hat{\sigma}_{u},\hat{\rho})\times I\left\{ A_{N}^{c}\right\} .
\label{eq:aux3_lemma3e}
\end{align}
Then,
\begin{align}
P_{N}( E_{N})  &~=~P_{N}( E_{N}\cap A_{N}) +P_{N}( E_{N}\cap A_{N}^{c}) \notag \\
&~\overset{(1)}{=}~\left\{ 
\begin{array}{c}
P_{N}( \{ F_{2,l}(\sqrt{N}(\hat{\theta}_{u}-\hat{\theta}_{l}),\hat{ \sigma}_{l},\hat{\sigma}_{u},\hat{\rho})>\Phi ^{-1}(1-\alpha /2)+\varepsilon \} \cap A_{N})  \notag \\
+P_{N}( \{ F_{0,l}(\hat{\sigma}_{l},\hat{\sigma}_{u},\hat{\rho} )>\Phi ^{-1}(1-\alpha /2)+\varepsilon \} \cap A_{N}^{c}) 
\end{array}
\right\}  \notag \\
&~\leq~ \left\{ 
\begin{array}{c}
P_{N}( F_{2,l}(\sqrt{N}(\hat{\theta}_{u}-\hat{\theta}_{l}),\hat{\sigma} _{l},\hat{\sigma}_{u},\hat{\rho})>\Phi ^{-1}(1-\alpha /2)+\varepsilon )  \\ 
+P_{N}( F_{0,l}(\hat{\sigma}_{l},\hat{\sigma}_{u},\hat{\rho})>\Phi ^{-1}(1-\alpha /2)+\varepsilon ) 
\end{array}
\right\}  \notag \\
&~\overset{(2)}{ \to }~0, \label{eq:aux3_lemma3f}
\end{align}
where (1) holds by \eqref{eq:aux3_lemma3e} and (2) by \eqref{eq:aux3_lemma3c} and \eqref{eq:aux3_lemma3d}. To complete the proof, note that
\begin{align*}
& P_{N}(\theta _{N}\in CI_{\alpha }^{3})  \notag \\
& \overset{(1)}{=}~P_{N}\left( 
\begin{array}{c}
\sqrt{N}(\hat{\theta}_{l}-\theta _{l}(P_{N}))-\hat{\sigma}_{l}F_{3,l}(\sqrt{N }(\hat{\theta}_{u}-\hat{\theta}_{l}),\sqrt{N}(\hat{\theta}_{u}-\hat{\theta} _{l}-b_{N}),\hat{\sigma}_{l},\hat{\sigma}_{u},\hat{\rho}) \\ 
\leq \sqrt{N}(\theta _{N}-\theta _{l}(P_{N}))\leq \sqrt{N}(\hat{\theta} _{u}-\theta _{u}(P_{N}))+\sqrt{N}(\theta _{u}(P_{N})-\theta _{l}(P_{N}))+ \\ 
\hat{\sigma}_{u}F_{3,u}(\sqrt{N}(\hat{\theta}_{u}-\hat{\theta}_{l}),\sqrt{N}(\hat{\theta}_{u}-\hat{\theta}_{l}-b_{N}),\hat{\sigma}_{l},\hat{\sigma}_{u}, \hat{\rho})
\end{array}
\right)   \notag \\
& \leq ~P_{N}(\sqrt{N}(\hat{\theta}_{l}-\theta _{l}(P_{N}))-\hat{\sigma}_{l}F_{3,l}(\sqrt{N }(\hat{\theta}_{u}-\hat{\theta}_{l}),\sqrt{N}(\hat{\theta}_{u}-\hat{\theta} _{l}-b_{N}),\hat{\sigma}_{l},\hat{\sigma}_{u},\hat{\rho}) \leq \sqrt{N}(\theta _{N}-\theta _{l}(P_{N})) )   \notag \\
& \overset{(2)}{\leq} ~P_{N}( \sqrt{N}(\hat{\theta}_{l}-\theta _{l}(P_{N}))-\hat{ \sigma}_{l}(\Phi ^{-1}(1-\alpha /2)+\varepsilon) \leq \sqrt{N} (\theta _{N}-\theta _{l}(P_{N}))) +P_{N}( E_{N}) \\
&\overset{(3)}{\to}~ \Phi ( \Phi ^{-1}(1-\alpha /2)+\varepsilon -\Psi _{l}/\sigma ) ,
\end{align*}
where (1) holds by Lemma \ref{lem:C3NotEmpty}, \eqref{eq:CI3} and \eqref{eq:CI3_problem}, (2) by definition of $E_N$, and (3) by \eqref{eq:aux3_lemma3b} and \eqref{eq:aux3_lemma3f}. This derivation implies \eqref{eq:aux0_lemma3b}, and concludes the proof of this part.

\noindent \underline{Part (c).} We divide the analysis into two cases. 

Case 1: $\eta \in (0,\infty )$. By $\sqrt{N}(\theta _{u}(P_{N})-\theta _{l}(P_{N})-b_{N}) \to \eta \in (0,\infty )$ and $b_{N} \to 0$, we conclude that $\theta _{u}(P_{N})-\theta _{l}(P_{N}) \to 0$. Lemma \ref{lem:near1} then implies $(\sigma _{l}(P_{N}),\sigma _{u}(P_{N}),\rho (P_{N})) \to (\sigma ,\sigma ,1)$ with $\sigma \in [\underline{\sigma },\bar{\sigma}]$. 

We now study the continuity properties of $F_{3}$ for any sequence $\{(\delta _{M,1},\delta _{M,2},\sigma _{M,l},\sigma _{M,u},\rho _{M})\}_{M\in \mathbb{N}}$ with $\lim_{M\to \infty}  (\delta _{M,1},\delta _{M,2},\sigma _{M,l},\sigma _{M,u},\rho _{M}) = (\infty ,\eta ,\sigma ,\sigma ,1)$. By evaluating in \eqref{eq:C3_defn}, we get
\begin{equation*}
C_{3}(\infty ,\eta ,\sigma ,\sigma ,1)~=~\left\{ ( c_{l},c_{u}) \in {\bar{\mathbb{R}}^{2}}~\text{s.t.}~\Phi (c_{l})~\geq ~1-\alpha ~\text{and}~\Phi (c_{u})~\geq ~1-\alpha \right\}. 
\end{equation*}
In turn, Lemma \ref{lem:criticalMin} shows that 
\begin{align*}
S_{3}(\infty ,\eta ,\sigma ,\sigma ,1)~&=~\{(\Phi ^{-1}(1-\alpha ),\Phi ^{-1}(1-\alpha ))\}\\
F_{3}(\infty ,\eta ,\sigma ,\sigma ,1)~&=~(\Phi ^{-1}(1-\alpha ),\Phi ^{-1}(1-\alpha )).
\end{align*}
Furthermore, part (a) of Lemma \ref{lem:C_is_UHC_3} shows that $C_{3}$ is continuous at $(\delta _{1},\delta _{2},\sigma _{l},\sigma _{u},\rho )=(\infty ,\eta ,\sigma ,\sigma ,1)$. Under these conditions, Lemma \ref{lem:Berge2} shows that $F_{3}$ is continuous at $(\delta _{1},\delta _{2},\sigma _{l},\sigma _{u},\rho )=(\infty ,\eta ,\sigma ,\sigma ,1)$, i.e., 
\begin{equation}
\underset{(\delta _{M,1},\delta _{M,2},\sigma _{M,l},\sigma _{M,u},\rho _{M}) \to (\infty ,\eta ,\sigma ,\sigma ,1)}{\text{lim}}F_{3}(\delta _{M,1},\delta _{M,2},\sigma _{M,l},\sigma _{M,u},\rho _{M})~=~(\Phi ^{-1}(1-\alpha ),\Phi ^{-1}(1-\alpha )).\label{eq:limit3_6}
\end{equation}

Since $(\hat{\theta}_{l},\hat{\theta}_{u},\hat{\sigma}_{l},\hat{\sigma}_{u}, \hat{\rho})$ satisfies OBS with parameter $({\theta }_{l}(P),{\theta }_{u}(P),{\sigma }_{l}(P),{\sigma }_{u}(P),{\rho }(P))$ and \eqref{eq:C3_limit_1}, 
\begin{equation}
(\sqrt{N}(\hat{\theta}_{l}-\theta _{l}(P_{N})),\sqrt{N}(\hat{\theta} _{u}-\theta _{u}(P_{N})),\hat{\sigma}_{l},\hat{\sigma}_{u},\hat{\rho}) ~\overset{d}{ \to }~(z\sigma ,z\sigma ,\sigma ,\sigma ,1),\label{eq:limit3_7}
\end{equation}
where $z\sim \mathcal{N}(0,1)$. Then, 
\begin{align}
&P_{N}(\theta _{N}\in CI_{\alpha }^{3})  \notag \\
& \overset{(1)}{=}~P_{N}\left( 
\begin{array}{c}
\sqrt{N}(\hat{\theta}_{l}-\theta _{l}(P_{N}))-\hat{\sigma}_{l}F_{3,l}(\sqrt{N }(\hat{\theta}_{u}-\hat{\theta}_{l}),\sqrt{N}(\hat{\theta}_{u}-\hat{\theta} _{l}-b_{N}),\hat{\sigma}_{l},\hat{\sigma}_{u},\hat{\rho}) \\ 
\leq \sqrt{N}(\theta _{N}-\theta _{l}(P_{N}))\leq \sqrt{N}(\hat{\theta} _{u}-\theta _{u}(P_{N}))+\sqrt{N}(\theta _{u}(P_{N})-\theta _{l}(P_{N}))+ \\ 
\hat{\sigma}_{u}F_{3,u}(\sqrt{N}(\hat{\theta}_{u}-\hat{\theta}_{l}),\sqrt{N}( \hat{\theta}_{u}-\hat{\theta}_{l}-b_{N}),\hat{\sigma}_{l},\hat{\sigma}_{u}, 
\hat{\rho})
\end{array}
\right)   \notag \\
& \overset{(2)}{\to}~\Phi (\Phi ^{-1}(1-\alpha )-\Psi _{l}/\sigma ),
\end{align}
as desired, where (1) holds by Lemma \ref{lem:C3NotEmpty}, \eqref{eq:CI3}, and \eqref{eq:CI3_problem}, which give ${c}^{3}=F_{3}(\sqrt{N}(\hat{\theta}_{u}-\hat{\theta}_{l}),\sqrt{N}(\hat{\theta}_{u}-\hat{\theta} _{l}-b_{N}),\hat{\sigma}_{l},\hat{\sigma}_{u},\hat{\rho})$ a.s., and (2) by \eqref{eq:C3_limit_1}, \eqref{eq:limit3_6}, \eqref{eq:limit3_7}, and the CMT.  

Case 2: $\eta =\infty $. We now study continuity properties of $F_{3}$ for any convergent sequence $\{(\delta _{M,1},\delta _{M,2},\sigma _{M,l},\sigma _{M,u},\rho _{M})\}_{M\in \mathbb{N}}$ with $\lim_{M\to \infty}  (\delta _{M,1},\delta _{M,2},\sigma _{M,l},\sigma _{M,u},\rho _{M}) = (\infty ,\infty ,\sigma _{l},\sigma _{u},\rho )$ with $(\sigma _{l},\sigma _{u},\rho )\in [\underline{\sigma },\overline{ \sigma }]\times [\underline{\sigma },\overline{\sigma }]\times [-1,1]$. By evaluating in \eqref{eq:C3_defn}, we get
\begin{equation*}
C_{3}(\infty ,\infty ,\sigma _{l},\sigma _{u},\rho )~=~\left\{ ( c_{l},c_{u}) \in {\bar{\mathbb{R}}^{2}}~\text{s.t.}~\Phi (c_{l})~\geq ~1-\alpha ~\text{and}~\Phi (c_{u})~\geq ~1-\alpha \right\} . 
\end{equation*}
In turn, Lemma \ref{lem:criticalMin} shows that 
\begin{align*}
S_{3}(\infty ,\infty ,\sigma _{l},\sigma _{u},\rho )~&=~\{(\Phi ^{-1}(1-\alpha ),\Phi ^{-1}(1-\alpha ))\}\\
F_{3}(\infty ,\infty ,\sigma _{l},\sigma _{u},\rho )~&=~(\Phi ^{-1}(1-\alpha ),\Phi ^{-1}(1-\alpha )).
\end{align*}
Furthermore, part (a) of Lemma \ref{lem:C_is_UHC_3} shows that $C_{3}$ is continuous at $(\delta _{1},\delta _{2},\sigma _{l},\sigma _{u},\rho )=(\infty ,\infty ,\sigma _{l},\sigma _{u},\rho )$ with $(\sigma _{l},\sigma _{u},\rho )\in [\underline{\sigma },\overline{\sigma }]\times [ \underline{\sigma },\overline{\sigma }]\times [-1,1]$. Under these conditions, Lemma \ref{lem:Berge2} shows that $F_{3}$ is continuous at $(\delta _{1},\delta _{2},\sigma _{l},\sigma _{u},\rho )=(\infty ,\infty
,\sigma _{l},\sigma _{u},\rho )$, i.e.,
\begin{equation}
\underset{(\delta _{M,1},\delta _{M,2},\sigma _{M,l},\sigma _{M,u},\rho _{M}) \to (\infty ,\infty ,\sigma _{l},\sigma _{u},\rho )}{\text{lim}} F_{3}(\delta _{M,1},\delta _{M,2},\sigma _{M,l},\sigma _{M,u},\rho _{M})~=~(\Phi ^{-1}(1-\alpha ),\Phi ^{-1}(1-\alpha )).\label{eq:limit3_8}
\end{equation}

Since $(\hat{\theta}_{l},\hat{\theta}_{u},\hat{\sigma}_{l},\hat{\sigma}_{u}, \hat{\rho})$ satisfies OBS with parameter $({\theta }_{l}(P),{\theta } _{u}(P),{\sigma }_{l}(P),{\sigma }_{u}(P),{\rho }(P))$ and \eqref{eq:C3_limit_1}, 
\begin{equation}
(\sqrt{N}(\hat{\theta}_{l}-\theta _{l}(P_{N})),\sqrt{N}(\hat{\theta} _{u}-\theta _{u}(P_{N})),\hat{\sigma}_{l},\hat{\sigma}_{u},\hat{\rho}) ~\overset{d}{ \to }~(z_{1}\sigma _{l},(z_{2}\sqrt{1-\rho ^{2}} +z_{1}\rho )\sigma _{u},\sigma _{l},\sigma _{u},\rho ),\label{eq:limit3_9}
\end{equation}
where $(z_{1},z_{2})\sim \mathcal{N}(\mathbf{0}_{2\times 1},\mathbf{I}_{2\times 2})$. Then, 
\begin{align}
& P_{N}(\theta _{N}\in CI_{\alpha }^{3})  \notag \\
& ~\overset{(1)}{=}~P_{N}\left( 
\begin{array}{c}
\sqrt{N}(\hat{\theta}_{l}-\theta _{l}(P_{N}))-\hat{\sigma}_{l}F_{3,l}(\sqrt{N }(\hat{\theta}_{u}-\hat{\theta}_{l}),\sqrt{N}(\hat{\theta}_{u}-\hat{\theta} _{l}-b_{N}),\hat{\sigma}_{l},\hat{\sigma}_{u},\hat{\rho}) \\ 
\leq \sqrt{N}(\theta _{N}-\theta _{l}(P_{N}))\leq \sqrt{N}(\hat{\theta} _{u}-\theta _{u}(P_{N}))+\sqrt{N}(\theta _{u}(P_{N})-\theta _{l}(P_{N}))+ \\ 
\hat{\sigma}_{u}F_{3,u}(\sqrt{N}(\hat{\theta}_{u}-\hat{\theta}_{l}),\sqrt{N}( \hat{\theta}_{u}-\hat{\theta}_{l}-b_{N}),\hat{\sigma}_{l},\hat{\sigma}_{u}, \hat{\rho})
\end{array}
\right)   \notag \\
& ~\overset{(2)}{ \to }~\Phi (\Phi ^{-1}(1-\alpha )-\Psi _{l}/\sigma_{l} ),
\end{align}
as desired, where (1) holds by Lemma \ref{lem:C3NotEmpty}, \eqref{eq:CI3}, and \eqref{eq:CI3_problem}, which give ${c}^{3}=F_{3}(\sqrt{N}(\hat{\theta}_{u}-\hat{\theta}_{l}),\sqrt{N}(\hat{\theta}_{u}-\hat{\theta} _{l}-b_{N}),\hat{\sigma}_{l},\hat{\sigma}_{u},\hat{\rho})$ a.s., and (2) by \eqref{eq:C3_limit_1}, \eqref{eq:limit3_8}, \eqref{eq:limit3_9}, and the CMT.  

\noindent \underline{Part (d).} We now show that $\theta _{u}(P_{N})-\theta _{l}(P_{N}) \to 0$. To see this, note that we have two cases: $\eta \in (-\infty,0)$ or $\eta =-\infty $. If $\eta \in (-\infty,0)$, $\sqrt{N}(\theta _{u}(P_{N})-\theta _{l}(P_{N})-b_{N}) \to \eta \in (-\infty,0)$. By this and $ b_{N} \to 0$, we get that $\theta _{u}(P_{N})-\theta _{l}(P_{N}) \to 0$.  If $\eta =-\infty $, $\theta _{u}(P_{N})-\theta _{l}(P_{N})\leq -1/\sqrt{N}+b_{N}$ for all sufficiently large $N$. By this, $\theta _{u}(P_{N})-\theta _{l}(P_{N})\geq 0$, and $ b_{N} \to 0$, we get that $\theta _{u}(P_{N})-\theta _{l}(P_{N}) \to 0$. Since $\theta _{u}(P_{N})-\theta _{l}(P_{N}) \to 0$, Lemma \ref{lem:near1} then implies $(\sigma _{l}(P_{N}),\sigma _{u}(P_{N}),\rho (P_{N})) \to (\sigma ,\sigma ,1)$ with $\sigma \in [\underline{ \sigma },\bar{\sigma}]$. 

We now study the continuity properties of $F_{3}$ for any sequence $\{(\delta _{M,1},\delta _{M,2},\sigma _{M,l},\sigma _{M,u},\rho _{M})\}_{M\in \mathbb{N}}$ with $\lim_{M\to \infty}  (\delta _{M,1},\delta _{M,2},\sigma _{M,l},\sigma _{M,u},\rho _{M}) = (\infty ,\eta
,\sigma ,\sigma ,1)$. By evaluating in \eqref{eq:C3_defn}, we get
\begin{equation*}
C_{3}(\infty ,\eta ,\sigma ,\sigma ,1)~=~\left\{ ( c_{l},c_{u}) \in {\bar{\mathbb{R}}^{2}}~\text{ s.t.}~P(-c_{l}\leq z_{1}\leq c_{u})~\geq ~1-\alpha \right\}.
\end{equation*}
In turn, Lemma \ref{lem:criticalMin} shows that
\begin{align*}
S_{3}(\infty ,\eta ,\sigma ,\sigma ,1 )~=~\{(\Phi ^{-1}(1-\alpha /2),\Phi ^{-1}(1-\alpha /2))\}\\
F_{3}(\infty ,\eta ,\sigma ,\sigma ,1 )~=~(\Phi ^{-1}(1-\alpha /2),\Phi ^{-1}(1-\alpha /2)).
\end{align*}
Furthermore, part (a) of Lemma \ref{lem:C_is_UHC_3} shows that $C_{3}$ is continuous at $(\delta _{1},\delta _{2},\sigma _{l},\sigma _{u},\rho )=(\infty ,\eta ,\sigma ,\sigma ,1)$. Under these conditions, Lemma \ref{lem:Berge2} shows that $F_{3}$ is continuous at $(\delta _{1},\delta _{2},\sigma _{l},\sigma _{u},\rho )=(\infty ,\eta ,\sigma ,\sigma ,1)$, i.e., 
\begin{equation}
\underset{(\delta _{M,1},\delta _{M,2},\sigma _{M,l},\sigma _{M,u},\rho _{M}) \to (\infty ,\eta ,\sigma ,\sigma ,1)}{\text{lim}}F_{3}(\delta _{M,1},\delta _{M,2},\sigma _{M,l},\sigma _{M,u},\rho _{M})~=~(\Phi ^{-1}(1-\alpha /2),\Phi ^{-1}(1-\alpha /2)).\label{eq:limit3_10}
\end{equation}

Since $(\hat{\theta}_{l},\hat{\theta}_{u},\hat{\sigma}_{l},\hat{\sigma}_{u}, \hat{\rho})$ satisfies OBS with parameter $({\theta }_{l}(P),{\theta } _{u}(P),{\sigma }_{l}(P),{\sigma }_{u}(P),{\rho }(P))$ and \eqref{eq:C3_limit_1}, 
\begin{equation}
(\sqrt{N}(\hat{\theta}_{l}-\theta _{l}(P_{N})),\sqrt{N}(\hat{\theta} _{u}-\theta _{u}(P_{N})),\hat{\sigma}_{l},\hat{\sigma}_{u},\hat{\rho}) \overset{d}{ \to }(z\sigma ,z\sigma ,\sigma ,\sigma ,1),
\label{eq:limit3_11}
\end{equation}
where $z\sim \mathcal{N}(0,1)$. Then, 
\begin{align}
& P_{N}(\theta _{N}\in CI_{\alpha }^{3})  \notag \\
& \overset{(1)}{=}~P_{N}\left( 
\begin{array}{c}
\sqrt{N}(\hat{\theta}_{l}-\theta _{l}(P_{N}))-\hat{\sigma}_{l}F_{3,l}(\sqrt{N }(\hat{\theta}_{u}-\hat{\theta}_{l}),\sqrt{N}(\hat{\theta}_{u}-\hat{\theta} _{l}-b_{N}),\hat{\sigma}_{l},\hat{\sigma}_{u},\hat{\rho}) \\ 
\leq \sqrt{N}(\theta _{N}-\theta _{l}(P_{N}))\leq \sqrt{N}(\hat{\theta} _{u}-\theta _{u}(P_{N}))+\sqrt{N}(\theta _{u}(P_{N})-\theta _{l}(P_{N}))+ \\ 
\hat{\sigma}_{u}F_{3,u}(\sqrt{N}(\hat{\theta}_{u}-\hat{\theta}_{l}),\sqrt{N}(\hat{\theta}_{u}-\hat{\theta}_{l}-b_{N}),\hat{\sigma}_{l},\hat{\sigma}_{u},
\hat{\rho})
\end{array}
\right)   \notag \\
& \overset{(2)}{ \to }~\Phi (\Phi ^{-1}(1-\alpha /2)-\Psi _{l}/\sigma ),
\end{align}
as desired, where (1) holds by Lemma \ref{lem:C3NotEmpty}, \eqref{eq:CI3}, and \eqref{eq:CI3_problem}, which give ${c}^{3}=F_{3}(\sqrt{N}(\hat{\theta} _{u}-\hat{\theta}_{l}),\sqrt{N}(\hat{\theta}_{u}-\hat{\theta} _{l}-b_{N}),\hat{\sigma}_{l},\hat{\sigma}_{u},\hat{\rho})$ a.s., and (2) by \eqref{eq:C3_limit_1}, \eqref{eq:limit3_10}, \eqref{eq:limit3_11}, and the CMT. 

\noindent \underline{Part (e).} By $\sqrt{N}(\theta _{u}(P_{N})-\theta _{l}(P_{N})) \to \mu \in \mathbb{R}_{+}$ and $b_{N} \to 0$, we conclude that $\theta _{u}(P_{N})-\theta _{l}(P_{N}) \to 0$. Lemma \ref{lem:near1} then implies $(\sigma _{l}(P_{N}),\sigma _{u}(P_{N}),\rho (P_{N})) \to (\sigma ,\sigma ,1)$ with $\sigma \in [\underline{\sigma },\bar{\sigma}]$. Also, by $\sqrt{N}b_{N} \to \infty $, we conclude that $\sqrt{N}(\theta _{u}(P_{N})-\theta _{l}(P_{N})-b_{N}) \to -\infty $.

We now study the continuity properties of $F_{3}$ for any sequence  $\{(\delta _{M,1},\delta _{M,2},\sigma _{M,l},\sigma _{M,u},\rho _{M})\}_{M\in \mathbb{N}}$ with $\lim_{M\to \infty}  (\delta _{M,1},\delta _{M,2},\sigma _{M,l},\sigma _{M,u},\rho _{M}) =(\mu ,-\infty ,\sigma ,\sigma ,1)$ with $(\mu  ,\sigma )\in \mathbb{R}_{+}\times [\underline{\sigma },\overline{\sigma }]$. By evaluating in \eqref{eq:C3_defn}, we get
\begin{equation*}
C_{3}(\mu ,-\infty ,\sigma ,\sigma ,1)=\left\{ ( c_{l},c_{u}) \in {\bar{\mathbb{R}}^{2}}~\text{ s.t.}~P(-c_{l}\leq z_{1}\leq c_{u})~\geq ~1-\alpha \right\}.
\end{equation*}
In turn, Lemma \ref{lem:criticalMin} shows that
\begin{align*}
S_{3}(\mu ,-\infty,\sigma ,\sigma ,1 )~=~\{(\Phi ^{-1}(1-\alpha /2),\Phi ^{-1}(1-\alpha /2))\}\\
F_{3}(\mu ,-\infty ,\sigma ,\sigma ,1 )~=~(\Phi ^{-1}(1-\alpha /2),\Phi ^{-1}(1-\alpha /2)).
\end{align*}
Furthermore, part (a) of Lemma \ref{lem:C_is_UHC_3} shows that $C_{3}$ is continuous at $(\delta _{1},\delta _{2},\sigma _{l},\sigma _{u},\rho )=(\mu ,-\infty ,\sigma ,\sigma ,1)$. Under these conditions, Lemma \ref{lem:Berge2} shows that $F_{3}$ is continuous at $(\delta _{1},\delta _{2},\sigma _{l},\sigma _{u},\rho )=(\mu ,-\infty ,\sigma ,\sigma ,1)$, i.e., 
\begin{equation}
\underset{(\delta _{M,1},\delta _{M,2},\sigma _{M,l},\sigma _{M,u},\rho _{M}) \to (\mu ,-\infty ,\sigma ,\sigma ,1)}{\text{lim}}F_{3}(\delta _{M,1},\delta _{M,2},\sigma _{M,l},\sigma _{M,u},\rho _{M})~=~(\Phi ^{-1}(1-\alpha /2),\Phi ^{-1}(1-\alpha /2)).
\label{eq:limit3_12}
\end{equation}
Since $(\hat{\theta}_{l},\hat{\theta}_{u},\hat{\sigma}_{l},\hat{\sigma}_{u}, \hat{\rho})$ satisfies OBS with parameter $({\theta }_{l}(P),{\theta } _{u}(P),{\sigma }_{l}(P),{\sigma }_{u}(P),{\rho }(P))$ and \eqref{eq:C3_limit_1}, 
\begin{equation}
(\sqrt{N}(\hat{\theta}_{l}-\theta _{l}(P_{N})),\sqrt{N}(\hat{\theta} _{u}-\theta _{u}(P_{N})),\hat{\sigma}_{l},\hat{\sigma}_{u},\hat{\rho}) \overset{d}{ \to }(z\sigma ,z\sigma ,\sigma ,\sigma ,1),
\label{eq:limit3_13}
\end{equation}
where $z\sim \mathcal{N}(0,1)$. Then, 
\begin{align*}
& P_{N}(\theta _{N}\in CI_{\alpha }^{3}) \\
& \overset{(1)}{=}~P_{N}\left( 
\begin{array}{c}
\sqrt{N}(\hat{\theta}_{l}-\theta _{l}(P_{N}))-\hat{\sigma}_{l}F_{3,l}(\sqrt{N }(\hat{\theta}_{u}-\hat{\theta}_{l}),\sqrt{N}(\hat{\theta}_{u}-\hat{\theta} _{l}-b_{N}),\hat{\sigma}_{l},\hat{\sigma}_{u},\hat{\rho}) \\ 
\leq \sqrt{N}(\theta _{N}-\theta _{l}(P_{N}))\leq \sqrt{N}(\hat{\theta} _{u}-\theta _{u}(P_{N}))+\sqrt{N}(\theta _{u}(P_{N})-\theta _{l}(P_{N}))+ \\ 
\hat{\sigma}_{u}F_{3,u}(\sqrt{N}(\hat{\theta}_{u}-\hat{\theta}_{l}),\sqrt{N}( \hat{\theta}_{u}-\hat{\theta}_{l}-b_{N}),\hat{\sigma}_{l},\hat{\sigma}_{u}, \hat{\rho})
\end{array}
\right)  \\
& \overset{(2)}{ \to }~P\left( \sigma z-\sigma \Phi ^{-1}(1-\alpha /2)\leq -\Psi _{l}\leq z\sigma +\mu +\sigma \Phi ^{-1}(1-\alpha /2)\right)  \\
& \overset{(3)}{=}~\Phi \left( (\Psi _{l}+\mu )/\sigma +\Phi ^{-1}(1-\alpha /2)\right) -\Phi \left( \Psi _{l}/\sigma -\Phi ^{-1}(1-\alpha /2)\right),
\end{align*}
as desired, where (1) holds by Lemma \ref{lem:C3NotEmpty}, \eqref{eq:CI3}, and \eqref{eq:CI3_problem}, which give ${c}^{3}=F_{3}(\sqrt{N}(\hat{\theta}_{u}-\hat{\theta}_{l}),\sqrt{N}(\hat{\theta}_{u}-\hat{\theta}_{l}-b_{N}),\hat{\sigma}_{l},\hat{\sigma}_{u},\hat{\rho})$ a.s., (2) by \eqref{eq:C3_limit_1}, \eqref{eq:limit3_12}, \eqref{eq:limit3_13}, and the
CMT, and (3) by $\sigma \geq \underline{\sigma }>0$ and $z\sim \mathcal{N}(0,1)$.

\noindent \underline{Parts (f)-(j).} These are analogous to those in parts (a)-(e), respectively, and are thus omitted. 
\end{proof}

\begin{lemma}\label{lem:near1} 
Let $(\hat{\theta}_{l},\hat{\theta}_{u},\hat{\sigma} _{l}, \hat{\sigma}_{u},\hat{\rho})$ satisfy OBS with parameters $({\theta } _{l}(P), {\theta }_{u}(P), {\sigma }_{l}(P), {\sigma }_{u}(P), {\rho }(P))$ and $\mathcal{P}$. Then, for any sequence of distributions $\{ P_{N}\} _{N\geq 1}$ that satisfies 
\begin{align}
\theta _{u}( P_{N}) -\theta _{l}(P_{N}) &~ \to ~ 0,  \label{eq:near2} \\
( \rho ( P_{N}) ,\sigma _{l}( P_{N}) ,\sigma _{u}( P_{N}) ) &~ \to ~ ( \rho ,\sigma _{l},\sigma _{u}) \in [ -1,1] \times [ \underline{\sigma },\overline{ \sigma}] \times [ \underline{\sigma },\bar{\sigma}],  \label{eq:near3}
\end{align}
we have that $\rho=1$ and $\sigma _{l}=\sigma _{u}$.
\end{lemma}
\begin{proof}
By \eqref{eq:near2}, we can construct a subsequence $\{k_{N}\}_{N\geq 1}$ of $\{N\}_{N\geq 1}$ s.t.\ $\sqrt{N}(\theta _{u}(P_{k_{N}})-\theta _{l}(P_{k_{N}})) \to 0$. Construct the subsequence recursively as follows. First, set $k_{1}=1$, and so $P_{k_{1}}=P_{1}$. Next, for any $N\in \mathbb{N}$ with $ N>1$, find $k_{N}\in \mathbb{N}$ s.t.\ $k_{N}>k_{N-1}$ and $|\sqrt{N}(\theta _{u}(P_{k_{N}})-\theta _{l}(P_{k_{N}}))|\leq 1/N$. (This is always possible because of \eqref{eq:near2}). By repeating this construction iteratively, and taking limits as $N \to \infty $, we get $\sqrt{N}(\theta _{u}(P_{k_{N}})-\theta _{l}(P_{k_{N}})) \to 0$. Next, consider the sequence $\{\tilde{P}_{N}\}_{N\geq 1}$ with $\tilde{P}_{N}=P_{k_{N}}$. By construction, we get
\begin{equation}
\sqrt{N}(\theta _{u}(\tilde{P}_{N})-\theta _{l}(\tilde{P}_{N}))~ \to ~0.  \label{eq:near4}
\end{equation}
Since $\{\tilde{P}_{N}\}_{N\geq 1}$ is a subsequence of $\{P_{N}\}_{N\geq 1}$, \eqref{eq:near3} implies that 
\begin{equation}
(\rho (\tilde{P}_{N}),\sigma _{l}(\tilde{P}_{N}),\sigma _{u}(\tilde{P} _{N}))~ \to ~(\rho ,\sigma _{l},\sigma _{u}).  \label{eq:near5}
\end{equation}
To complete the proof, it suffices to show that $\rho =1$ and $\sigma _{l}=\sigma _{u}$.

First, note that OBS applied to $\{\tilde{P}_{N}\}_{N\geq 1}$ and \eqref{eq:near5} imply 
\begin{equation}
\sqrt{N}\left( 
\begin{array}{c}
\hat{\theta}_{l}-\theta _{l}(\tilde{P}_{N}), \\ 
\hat{\theta}_{u}-\theta _{u}(\tilde{P}_{N})
\end{array}%
\right) ~\overset{d}{ \to }~N\left( \mathbf{0}_{2\times 1},\left( 
\begin{tabular}{cc}
$\sigma _{l}^{2}$ & $\rho \sigma _{l}\sigma _{u}$ \\ 
$\rho \sigma _{l}\sigma _{u}$ & $\sigma _{u}^{2}$
\end{tabular}
 \right) \right) .  \label{eq:near6}
\end{equation}

Let $Z_{N}\equiv \sqrt{N}(\hat{\theta}_{l}-\hat{\theta}_{u})$. Then, consider the following derivation. 
\begin{align}
Z_{N}& ~=~\sqrt{N}(\hat{\theta}_{l}-\theta _{l}(\tilde{P}_{N}))-\sqrt{N}(\hat{\theta} _{u}-\theta _{u}(\tilde{P}_{N}))-\sqrt{N}(\theta _{u}(\tilde{P}_{N})-\theta _{l}(\tilde{P}_{N})) \nonumber \\
& ~{\overset{d}{ \to }}~Z\sim \mathcal{N}(0,\sigma _{l}^{2}+\sigma _{u}^{2}-2\rho \sigma _{l}\sigma _{u}).  \label{eq:near7}
\end{align}%
where the convergence holds by \eqref{eq:near4} and \eqref{eq:near6}. By OBS, $\tilde{P}_{N}(Z_{N}\leq 0)=\tilde{P}_{N}(\hat{\theta}_{l}\leq \hat{\theta}_{u})=1$. This and \eqref{eq:near7} then imply that $\sigma _{l}^{2}+\sigma _{u}^{2}-2\rho \sigma _{l}\sigma _{u}=0$ or, equivalently, 
\begin{equation}
(\sigma _{l}-\sigma _{u})^{2}~=~2(\rho -1)\sigma _{l}\sigma _{u}.\label{eq:near9}
\end{equation}
To see why, note that \eqref{eq:near7} and $\sigma _{l}^{2}+\sigma _{u}^{2}-2\rho \sigma _{l}\sigma _{u}>0$ imply $\tilde{P}_{N}(Z_{N}\leq 0) \to P(Z\leq 0)<1$, which contradicts $\tilde{P}_{N}(Z_{N}\leq 0)=1$.

We are now ready to conclude the proof. First, note that \eqref{eq:near9}, $(\sigma _{l}-\sigma _{u})^{2}\geq 0$, and $(\rho ,\sigma _{l},\sigma _{u})\in [-1,1]\times [\underline{\sigma },\overline{\sigma }]\times [\underline{\sigma },\overline{\sigma }]$ imply that $\rho =1$. In turn, \eqref{eq:near9} and $\rho =1$ imply that $(\sigma _{l}-\sigma _{u})^{2}=0$ or, equivalently, $\sigma _{l}=\sigma _{u}$.
\end{proof}

\begin{lemma}\label{lem:tildeH} 
Assume that $\alpha \in (0,1/2)$, and let $\underline{\sigma },\overline{\sigma }>0$ be as in Definition \ref{def:setup}. Let $H :[ \underline{\sigma },\overline{\sigma }] \times {\mathbb{R}}_{+}\times \bar{\mathbb{R}}_{+} \to \mathbb{R}$ be defined as follows: 
\begin{equation}
H( \sigma ,\mu ,\Psi ) ~\equiv~ \Phi \big( \tfrac{\Psi +\mu }{\sigma }+{G}( \tfrac{\mu }{\sigma }) \big) -\Phi\big( \tfrac{\Psi }{\sigma }-{G}( \tfrac{\mu }{\sigma })\big) ,
\label{eq:tildeH}
\end{equation}
where ${G}( y) :\mathbb{R}_{+} \to \mathbb{R}_{++}$ is implicitly defined as the unique $c>0$ that solves 
\begin{equation}
\Phi ( c+y) -\Phi ( -c) ~=~1-\alpha .  \label{eq:tildeG}
\end{equation}
Then, we have the following results.
\begin{enumerate}[(a)]
    \item $H( \sigma ,\mu ,\Psi )$ is weakly increasing in $\sigma$ for all $( \sigma ,\mu ,\Psi ) \in [ \underline{\sigma },\overline{\sigma }] \times {\mathbb{R}}_{+}\times \bar{\mathbb{R}}_{+}$,
    \item $H( \sigma ,\mu ,\Psi )$ is strictly increasing in $\sigma$ for all $( \sigma ,\mu ,\Psi ) \in [\underline{\sigma },\overline{\sigma }]\times {\mathbb{R}}_{+}\times \mathbb{R}_{++}$. 
\end{enumerate}
\end{lemma}
\begin{proof}
We first verify that \eqref{eq:tildeG} has a unique and positive solution. For any $y\in \mathbb{R}_{+}$, let $L(c) \equiv \Phi (c+y) -\Phi (-c) $ denote the left-hand side of \eqref{eq:tildeG}. Note that $L$ is continuous and strictly increasing because $L'(c)=\phi(c+y)+\phi(-c)>0$. Moreover, $L(0) =\Phi (y) -\Phi (0) \leq 0.5<1-\alpha $ and $\lim_{c\to\infty}L(c)=1>1-\alpha$. Thus, by the intermediate value theorem, \eqref{eq:tildeG} has a solution. Strict monotonicity gives uniqueness, and $L(0)<1-\alpha$ implies that the solution is positive.

Note that $H( \sigma ,\mu ,\Psi )$ is differentiable in $\sigma$. Therefore, part (a) follows from verifying that, for all $( \sigma ,\mu ,\Psi ) \in [ \underline{\sigma },\overline{\sigma }] \times {\mathbb{R}}_{+}\times \bar{\mathbb{R}}_{+}$,
\begin{equation}
\frac{\partial H( \sigma ,\mu ,\Psi ) }{\partial \sigma }~\geq~ 0.
\label{eq:tildeH0}
\end{equation}
Moreover, part (b) follows from verifying that \eqref{eq:tildeH0} becomes strict when $\Psi \in (0,\infty)$. 

We divide the proof into cases.

\noindent{Case 1}: $\Psi =\infty$. This case implies that $H( \sigma ,\mu ,\infty ) =0$, and so $\partial H( \sigma,\mu ,\infty ) /\partial \sigma =0$.

\noindent{Case 2}: $\Psi <\infty $. Note that
\begin{equation}
\frac{\partial H( \sigma ,\mu ,\Psi ) }{\partial \sigma }~=~- \frac{\left(
\phi ( \tfrac{\Psi +\mu }{\sigma }+G(\tfrac{\mu }{\sigma })) \left(\left( G^{\prime }(\tfrac{\mu }{\sigma })+1\right) \mu +\Psi \right)   
+\phi ( \tfrac{\Psi }{\sigma }-G(\tfrac{\mu }{\sigma })) \left(G^{\prime }(\tfrac{\mu }{\sigma })\mu -\Psi \right) 
\right)}{\sigma ^{2}}.  \label{eq:tildeH2}
\end{equation}
We complete our proof by showing that the numerator of \eqref{eq:tildeH2} is non-positive or, equivalently, 
\begin{align}
\exp \big( -{( \tfrac{\Psi }{\sigma }-G(\tfrac{\mu }{\sigma })) ^{2}}/{2}\big) ~( \Psi -G^{\prime }( \tfrac{\mu }{\sigma }) \mu ) \geq \exp \big( -{( \tfrac{\Psi +\mu }{\sigma }+G(\tfrac{\mu }{\sigma })) ^{2}}/{2}\big) ~ ( ( G^{\prime }(\tfrac{\mu }{\sigma })+1) \mu +\Psi ) , \label{eq:tildeH4}
\end{align}
and that this becomes strict when $\Psi >0$.

We now divide the analysis into three sub-cases.

\noindent{Case 2.1}: $\mu =0$ and $\Psi =0$. In this case, \eqref{eq:tildeH4} is equivalent to $0\geq 0$.

\noindent{Case 2.2}: $\mu =0$ and $\Psi \in (0,\infty )$. In this case, the strict inequality in \eqref{eq:tildeH4} is equivalent to $( \tfrac{\Psi }{\sigma }-G( 0)) ^{2}< ( \tfrac{\Psi }{\sigma }+G(0)) ^{2}$ which, in turn, is equivalent to $G(0)\tfrac{\Psi }{\sigma }>0$. This follows from $G(0)>0$ and $\sigma \geq \underline{\sigma }>0$.

\noindent{Case 2.3}: $\mu >0$ and $\Psi \in [0,\infty )$. As a preliminary result, we now show that
\begin{equation}
G^{\prime }( \tfrac{\mu }{\sigma }) ~=~\frac{-\phi ( {G}( \tfrac{\mu }{\sigma }) +\tfrac{\mu }{\sigma }) }{\phi ( {G}( \tfrac{\mu }{\sigma }) +\tfrac{\mu }{\sigma }) +\phi ( -{G}( \tfrac{\mu }{\sigma })) }~<~0.  \label{eq:tildeH1}
\end{equation}
To this end, note that for any $y\in \mathbb{R}_{+}$, ${G}( y) $ satisfies $F(y)\equiv \Phi ( G(y)+y) -\Phi ( -G(y)) -(1-\alpha )=0$. Since the relationship holds for all $y\in \mathbb{R}_{+}$, we have that $F^{\prime }( y) =0$, i.e.,
\begin{equation*}
\phi ( G(y)+y) ( G^{\prime }( y) +1) +\phi ( -G(y)) G^{\prime }( y) ~=~0. 
\end{equation*}
By evaluating on $y=\mu /\sigma $ with $( \sigma ,\mu ) \in [ \underline{\sigma },\overline{\sigma }] \times \bar{\mathbb{R}} _{+}$ and solving for $G^{\prime }( \mu /\sigma ) $, \eqref{eq:tildeH1} follows.

Since $G^{\prime }( \tfrac{\mu }{\sigma }) <0$ (by \eqref{eq:tildeH1}), $\Psi \geq 0$, and $\mu >0$, we have that $\Psi -G^{\prime }( \tfrac{\mu }{\sigma })\mu >0$, and so \eqref{eq:tildeH4} is equivalent to 
\begin{equation}
\exp \Big( \tfrac{1}{2}{( \tfrac{\Psi +\mu }{\sigma }+G( \tfrac{\mu }{ \sigma })) ^{2}}-\tfrac{1}{2}{( \tfrac{\Psi }{\sigma }-G( \tfrac{ \mu }{\sigma })) ^{2}}\Big) ~\geq~ \frac{( G^{\prime }( \tfrac{\mu }{\sigma })+1) \mu +\Psi }{\Psi -G^{\prime }( \tfrac{\mu }{\sigma })\mu }.\label{eq:tildeH6}
\end{equation}
By combining with \eqref{eq:tildeH1}, \eqref{eq:tildeH6} is equivalent to
\begin{align}
\Big( \tfrac{\Psi }{\mu }+\tfrac{\phi ( G( \tfrac{\mu }{\sigma }) +\tfrac{\mu }{\sigma }) }{\phi ( G( \tfrac{\mu }{ \sigma }) +\tfrac{\mu }{\sigma }) +\phi ( -G( \tfrac{ \mu }{\sigma })) }\Big) ~ \exp \big( \tfrac{\Psi \mu }{\sigma ^{2}}+\tfrac{\mu ^{2}}{2\sigma ^{2}}+G( \tfrac{\mu }{\sigma }) \tfrac{2\Psi +\mu }{\sigma }\big)  
& \geq ~\tfrac{\phi ( -G( \tfrac{\mu }{\sigma })) }{ \phi ( G( \tfrac{\mu }{\sigma }) +\tfrac{\mu }{\sigma }) +\phi ( -G( \tfrac{\mu }{\sigma })) }+\tfrac{ \Psi }{\mu }.  \label{eq:tildeH7}
\end{align}
To verify \eqref{eq:tildeH7}, it suffices to check that:
\begin{align}
( \Psi /\mu )~  \exp \big( \tfrac{\Psi \mu }{\sigma ^{2}}+\tfrac{\mu ^{2}}{2\sigma ^{2}}+G( \tfrac{\mu }{\sigma }) \tfrac{2\Psi +\mu }{\sigma }\big) & ~\geq~ \Psi /\mu ,\notag\\  
\phi ( G(\mu /\sigma )+\mu /\sigma ) ~ \exp \big( \tfrac{\Psi \mu }{\sigma ^{2}}+\tfrac{\mu ^{2}}{2\sigma ^{2}}+G( \tfrac{\mu }{\sigma }) \tfrac{2\Psi +\mu }{\sigma }\big) & ~\geq~ \phi ( -G(\mu /\sigma )) . \label{eq:tildeH9}
\end{align}
Both of these can be shown by algebra and the fact that $\Psi \geq 0$, $\mu >0$, $G( \tfrac{\mu }{\sigma }) >0$, and $\sigma \geq \underline{\sigma }>0$. 

Finally, we note that $\Psi>0$ further implies that both inequalities in \eqref{eq:tildeH9} hold strictly. By retracing our steps, this implies that  \eqref{eq:tildeH4}, \eqref{eq:tildeH6}, and \eqref{eq:tildeH7} all become strict in this case.
\end{proof}

\begin{lemma}\label{lem:criticalMin}
For any $\alpha <1/2$, $\delta \geq 0$, and $\sigma >0$, consider the problem:
\begin{align}
&\min_{c_{l},c_{u}\in \mathbb{R}}\left( c_{l}+c_{u}\right) \sigma~~~\text{s.t.}\notag\\
&\Phi \left( c_{u}+\delta /\sigma \right) -\Phi \left( -c_{l}\right) \geq 1-\alpha~~\text{and}~~\Phi \left( c_{u}\right) -\Phi \left( -c_{l}-\delta/\sigma \right) \geq 1-\alpha .
\label{eq:criticalMin}
\end{align}
This problem has a unique solution $\left( c_{l}^{\ast },c_{u}^{\ast }\right) =\left( G(\delta /\sigma),G(\delta /\sigma)\right) $, where ${G}( y) :\mathbb{R}_{+} \to \mathbb{R}_{++}$ is defined in Lemma \ref{lem:tildeH}.
\end{lemma}
\begin{proof}
First, we show that $c_{l}^{\ast },c_{u}^{\ast }> 0$. To see why, note that $c_{l}^{\ast }\leq 0$ implies a violation of the first constraint in \eqref{eq:criticalMin} and $c_{u}^{\ast }\leq 0$ would imply a violation of the second constraint in \eqref{eq:criticalMin}. For brevity, we only show the first one. To see this, note that $c_{l}^{\ast }\leq 0$ implies 
\begin{equation*}
\Phi ( c_{u}^{\ast }+\delta /\sigma ) -\Phi ( -c_{l}^{\ast }) ~\leq~ 1-\Phi ( -c_{l}^{\ast })~ =~\Phi ( c_{l}^{\ast }) ~\leq~ \Phi ( 0) ~=~1/2~<~1-\alpha.
\end{equation*}
This is a violation of the first constraint, as desired. 

By $c_{l}^{\ast },c_{u}^{\ast }> 0$ and $\sigma >0$, we have that \eqref{eq:criticalMin} is equivalent to 
\begin{align} 
&\min_{c_{l},c_{u}\in \mathbb{R}}( c_{l}+c_{u})~~~\text{s.t.}~~\notag\\ &\Phi( c_{u}+\delta /\sigma ) +\Phi ( c_{l}) \geq 2-\alpha ~~\text{and}~~\Phi ( c_{u}) +\Phi ( c_{l}+\delta /\sigma ) \geq 2-\alpha . \label{eq:criticalMin3}
\end{align} 
We focus on \eqref{eq:criticalMin3} for the remainder of the proof.

We now show that at any solution $( c_{l}^{\ast },c_{u}^{\ast }) $ one of the two constraints must be binding. To see why, suppose that this is not the case and $( c_{l}^{\ast },c_{u}^{\ast }) $ minimizes $( c_{l}+c_{u}) $ with $\Phi ( c_{u}+\delta /\sigma ) +\Phi ( c_{l}) \geq 2-\alpha $ and $\Phi ( c_{u}) +\Phi ( c_{l}+\delta /\sigma ) \geq 2-\alpha $. Since $\Phi $ is continuous and the constraints are slack, we can then find an alternative candidate solution $( \tilde{c}_{l}^{\ast },\tilde{c}_{u}^{\ast })$ with $0\leq\tilde{c}_{l}^{\ast } <c_{l}^{\ast }$ and $0\leq\tilde{c}_{u}^{\ast } <c_{u}^{\ast }$ that also satisfies the constraints. Since $\tilde{c}_{l}^{\ast }+\tilde{c}_{u}^{\ast }<c_{l}^{\ast }+c_{u}^{\ast }$, this contradicts that $( c_{l}^{\ast },c_{u}^{\ast }) $ was a minimizer.

We now show that the solution is unique. Suppose that this is not the case, and that $( c_{l}^{\ast },c_{u}^{\ast }) $ and $( \tilde{c}_{l}^{\ast },\tilde{c}_{u}^{\ast }) $ are solutions to the problem and $( c_{l}^{\ast },c_{u}^{\ast }) \neq ( \tilde{c}_{l}^{\ast },\tilde{c}_{u}^{\ast }) $. Let $V=c_{l}^{\ast }+c_{u}^{\ast }=\tilde{c}_{l}^{\ast }+\tilde{c}_{u}^{\ast }$ denote their (common) minimized value. 

Then, consider $( ( c_{l}^{\ast }+\tilde{c}_{l}^{\ast }) /2,( c_{u}^{\ast }+\tilde{c}_{u}^{\ast }) /2) $ as another candidate solution. Note that this achieves the same minimized value $( c_{l}^{\ast }+c_{u}^{\ast }) /2+( \tilde{c}_{l}^{\ast }+\tilde{c}_{u}^{\ast }) /2=V$. Then, we have
\begin{align}
\Phi ( \tfrac{c_{u}^{\ast }+\tilde{c}_{u}^{\ast }}{2}+\delta /\sigma ) +\Phi ( \tfrac{c_{l}^{\ast }+\tilde{c}_{l}^{\ast }}{2}) &~\overset{(1)}{>}~\frac{ \Phi ( c_{u}^{\ast }+\delta /\sigma ) +\Phi ( c_{l}^{\ast })   +\Phi ( \tilde{c}_{u}^{\ast }+\delta /\sigma ) +\Phi ( \tilde{c}_{l}^{\ast }) }{2}
 ~\overset{(3)}{\geq}~ 2-\alpha \notag \\
\Phi ( \tfrac{c_{u}^{\ast }+\tilde{c}_{u}^{\ast }}{2}) +\Phi( \tfrac{c_{l}^{\ast }+\tilde{c}_{l}^{\ast }}{2}+\delta /\sigma )  &~\overset{(2)}{>}~\frac{\Phi ( c_{u}^{\ast }) +\Phi ( c_{l}^{\ast }+\delta /\sigma)  +\Phi ( \tilde{c}_{u}^{\ast }) +\Phi ( \tilde{c}_{l}^{\ast}+\delta /\sigma ) 
}{2} ~\overset{(4)}{\geq}~ 2-\alpha ,\label{eq:criticalMin4}
\end{align}
where (1) and (2) hold by $c_{l}^{\ast },c_{u}^{\ast },\tilde{c}_{l}^{\ast },\tilde{c}_{u}^{\ast }> 0$ and $\delta /\sigma \geq 0$, $\Phi ( x) $ is strictly concave for $x>0$, and either $c_{l}^{\ast } \neq \tilde{c}_{l}^{\ast }$ or $c_{u}^{\ast } \neq \tilde{c}_{u}^{\ast }$, and (3) and (4) hold because $( c_{l}^{\ast },c_{u}^{\ast })$ and $( \tilde{c}_{l}^{\ast },\tilde{c}_{u}^{\ast }) $ satisfy the constraints. Note that \eqref{eq:criticalMin4} implies that $( ( c_{l}^{\ast }+\tilde{c}_{l}^{\ast }) /2,( c_{u}^{\ast }+\tilde{c}_{u}^{\ast }) /2) $ is an optimal solution where neither constraint is binding (i.e., an interior solution), which is a contradiction.

Next, we show that the unique solution $( c_{l}^{\ast },c_{u}^{\ast }) $ satisfies $c_{l}^{\ast }=c_{u}^{\ast }$. Suppose this is not the case, i.e., $c_{l}^{\ast } \neq c_{u}^{\ast }$. From here, it is easy to see that $( (c_{l}^{\ast }+c_{u}^{\ast })/2,(c_{l}^{\ast }+c_{u}^{\ast })/2) $ is a different candidate that is feasible and achieves the same minimized value. But this then implies that the solution is not unique, which is a contradiction.

By combining all previous arguments, we get that the unique solution is $( c_{l}^{\ast },c_{u}^{\ast }) =( c,c) $ with $c>0$ that satisfies $\Phi ( c+\delta /\sigma ) -\Phi ( -c) \geq 1-\alpha $ or $\Phi ( c) -\Phi ( -c-\delta /\sigma ) \geq 1-\alpha $, one of the two constraints with equality. We now show that this implies that both hold with equality. To this end, suppose that the first one holds with equality, i.e., $\Phi ( c+\delta /\sigma ) -\Phi ( -c) =1-\alpha $. This is equivalent to $\Phi ( c+\delta /\sigma ) +\Phi ( c) =2-\alpha $ or, equivalently, $\Phi ( c) -\Phi ( -c-\delta /\sigma ) =1-\alpha $, which is the second equality. Finally, we note that the first constraint with equality is exactly how we define $G(\delta /\sigma)$, completing the proof.
\end{proof}

\begin{lemma}\label{lem:C2_is_continuous}
The correspondence $C_2:\bar{\mathbb{R}}_{+}\times [\underline{\sigma },\overline{\sigma }]
\times [\underline{\sigma },\overline{\sigma }]\times [-1,1]
\rightrightarrows \bar{\mathbb{R}}^{2}$ in \eqref{eq:C2_defn} is continuous (i.e., both UHC and LHC) at the following values:
\begin{enumerate}[(a)]
\item $(\tilde{\delta},\tilde{\sigma}_{l},\tilde{\sigma}_{u},\tilde{\rho})
=(\infty,\sigma_l,\sigma_u,\rho)$ with
$(\sigma_l,\sigma_u,\rho)\in [\underline{\sigma},\overline{\sigma}]
\times [\underline{\sigma},\overline{\sigma}]\times [-1,1]$.
\item $(\tilde{\delta},\tilde{\sigma}_{l},\tilde{\sigma}_{u},\tilde{\rho})
=(\delta,\sigma,\sigma,1)$ with
$(\delta,\sigma)\in \bar{\mathbb{R}}_{+}\times [\underline{\sigma},\overline{\sigma}]$.
\end{enumerate}
\end{lemma}
\begin{proof} 
\noindent\underline{Part (a).} For any $(\sigma_l,\sigma_u,\rho)\in [\underline{\sigma},\overline{\sigma}]\times [\underline{\sigma},\overline{\sigma}]\times [-1,1]$, consider $(\tilde{\delta},\tilde{\sigma}_{l},\tilde{\sigma}_{u},\tilde{\rho})=(\infty,\sigma_l,\sigma_u,\rho)$.

We first show UHC. Let $\lim_{M\to \infty} (\delta_M,\sigma_{M,l},\sigma_{M,u},\rho_M)=(\infty,\sigma_l,\sigma_u,\rho)$ and $(c_{M,l},c_{M,u}) \in \\C_2(\delta_M,\sigma_{M,l},\sigma_{M,u},\rho_M)$ with $\lim_{M\to \infty}  (c_{M,l},c_{M,u})= (c_l,c_u)$. By \eqref{eq:C2_defn}, $(c_{M,l},c_{M,u})\in \\C_2(\delta_M,\sigma_{M,l},\sigma_{M,u},\rho_M)$ implies
\begin{align}
1-\alpha~&\leq~P(\{-c_{M,l}\leq z_1\}\cap \{\rho_M z_1\leq c_{M,u}+{\delta_M}/{\sigma_{M,u}}+\sqrt{1-\rho_M^2}z_2\})~\leq~1-\Phi(-c_{M,l})\notag \\
1-\alpha~&\leq~P(\{-c_{M,l}-{\delta_M}/{\sigma_{M,l}}+\sqrt{1-\rho_M^2}z_2\leq \rho_M z_1\}\cap\{z_1\leq c_{M,u}\})~\leq~\Phi(c_{M,u}).\label{eq:C_is_continuous_1}
\end{align}
Taking limits in both equations in \eqref{eq:C_is_continuous_1} as $M\to \infty$ yields $1-\Phi(-c_l)\geq 1-\alpha$ and $\Phi(c_u)\geq 1-\alpha$. By \eqref{eq:C2_defn}, this is equivalent to $(c_l,c_u)\in C_2(\infty,\sigma_l,\sigma_u,\rho)$, proving UHC.

We now show LHC. Let $\lim_{M\to \infty}  (\delta_M,\sigma_{M,l},\sigma_{M,u},\rho_M)=(\infty,\sigma_l,\sigma_u,\rho)$ and fix $(c_l,c_u)\in C_2(\infty,\sigma_l,\sigma_u,\rho)$. For each $k\in\mathbb N$, set $\tilde{c}_{k,l}=c_l+1/k$ and $\tilde{c}_{k,u}=c_u+1/k$. Then $\lim_{k\to \infty}  (\tilde{c}_{k,l},\tilde{c}_{k,u})=(c_l,c_u)$, $1-\Phi(-\tilde{c}_{k,l})>1-\alpha$, and $\Phi(\tilde{c}_{k,u})>1-\alpha$. Moreover, for each fixed $k\in\mathbb N$, $\lim_{M\to \infty}  (\delta_M,\sigma_{M,l},\sigma_{M,u},\rho_M)=(\infty,\sigma_l,\sigma_u,\rho)$ and the dominated convergence theorem give
\begin{equation*}
\begin{aligned}
&\lim_{M \to \infty} P\big(\{-\tilde{c}_{k,l}\leq z_1\}\cap \big\{\rho_M z_1\leq \tilde{c}_{k,u}+{\delta_M}/{\sigma_{M,u}} +\sqrt{1-\rho_M^2}z_2\big\}\big)~=~1-\Phi(-\tilde{c}_{k,l})~>~1-\alpha,\\
&\lim_{M \to \infty}  P\big(\big\{-\tilde{c}_{k,l}-{\delta_M}/{\sigma_{M,l}}+\sqrt{1-\rho_M^2}z_2\leq \rho_M z_1\big\}\cap\{z_1\leq \tilde{c}_{k,u}\}\big)~=~\Phi(\tilde{c}_{k,u})~>~1-\alpha,
\end{aligned}
\end{equation*}
where $(z_1,z_2)\sim\mathcal N({\bf 0}_{2\times 1},{\bf I}_{2\times 2})$. Hence, for each fixed $k\in \mathbb{N}$, there exists $M_k\in \mathbb{N}$ such that $(\tilde{c}_{k,l},\tilde{c}_{k,u})\in C_2(\delta_M,\sigma_{M,l},\sigma_{M,u},\rho_M)$ for all $M\geq M_k$. 
Without loss of generality, we can take $M_k \in \mathbb{N}$ strictly increasing in $k$. Then, for any $M \in \mathbb{N}$, proceed as follows. For $M <M_1$, set $(c_{M,l},c_{M,u})= (\infty,\infty)$. For $M \in [M_k,M_{k+1})$ with $k\in \mathbb{N}$, set $(c_{M,l},c_{M,u})=(\tilde{c}_{k,l},\tilde{c}_{k,u})$. By construction, this means that $(c_{M,l},c_{M,u}) \in C_2(\delta_M,\sigma_{M,l},\sigma_{M,u},\rho_M)$ for all $M \in \mathbb{N}$. Moreover, since $M_k$ is strictly increasing, the index $k$ corresponding to the interval $M\in [M_k,M_{k+1})$ satisfies $k\to\infty$ as $M\to\infty$. Therefore, as $\lim_{M\to \infty} (c_{M,l},c_{M,u})=\lim_{M\to \infty} (\tilde{c}_{k,l},\tilde{c}_{k,u})=(c_l,c_u)$. Thus, we have constructed a sequence $(c_{M,l},c_{M,u})$ satisfying  $(c_{M,l},c_{M,u}) \in C_2(\delta_M,\sigma_{M,l},\sigma_{M,u},\rho_M)$ with $\lim_{M\to \infty} (c_{M,l},c_{M,u})=(c_l,c_u)$, which proves LHC.

\noindent\underline{Part (b).} For any $(\delta,\sigma)\in \bar{\mathbb R}_{+}\times[\underline{\sigma},\overline{\sigma}]$, consider $(\tilde{\delta},\tilde{\sigma}_{l},\tilde{\sigma}_{u},\tilde{\rho})=(\delta,\sigma,\sigma,1)$. If $\delta=\infty$, the result follows from part (a), so suppose that $\delta<\infty$.

We first show UHC. Let $\lim_{M\to \infty}  (\delta_M,\sigma_{M,l},\sigma_{M,u},\rho_M) = (\delta,\sigma,\sigma,1)$ and $(c_{M,l},c_{M,u})\in \\C_2(\delta_M,\sigma_{M,l},\sigma_{M,u},\rho_M)$ with $\lim_{M\to \infty}  (c_{M,l},c_{M,u})=(c_l,c_u)$. Since $\delta<\infty$, we have $\delta_M<\infty$ for all sufficiently large $M$. By \eqref{eq:C2_defn}, $(c_{M,l},c_{M,u})\in C_2(\delta_M,\sigma_{M,l},\sigma_{M,u},\rho_M)$ implies that
\begin{align}
&P\big(\{-c_{M,l}\leq z_1\}\cap \big\{\rho_M z_1\leq c_{M,u}+{\delta_M}/{\sigma_{M,u}}+\sqrt{1-\rho_M^2}z_2\big\}\big)~\geq~1-\alpha,\notag\\
&P\big(\big\{-c_{M,l}-{\delta_M}/{\sigma_{M,l}}+\sqrt{1-\rho_M^2}z_2\leq \rho_M z_1\big\}\cap\{z_1\leq c_{M,u}\}\big)~\geq~1-\alpha .\label{eq:C_is_continuous_10}
\end{align}
Moreover, $\lim_{M\to \infty}  (\delta_M,\sigma_{M,l},\sigma_{M,u},\rho_M)=(\delta,\sigma,\sigma,1)$ and the dominated convergence theorem give
\begin{align}
&\lim_{M \to \infty} P\big(\{-c_{M,l}\leq z_1\}\cap \big\{\rho_M z_1\leq c_{M,u}+{\delta_M}/{\sigma_{M,u}} +\sqrt{1-\rho_M^2}z_2\big\}\big)~=~\Phi(c_u+{\delta}/{\sigma})-\Phi(-c_l),\notag\\
&\lim_{M \to \infty}  P\big(\big\{-c_{M,l}-{\delta_M}/{\sigma_{M,l}}+\sqrt{1-\rho_M^2}z_2\leq \rho_M z_1\big\}\cap\{z_1\leq c_{M,u}\}\big)~=~\Phi(c_u)-\Phi(-c_l-{\delta}/{\sigma}),\label{eq:C_is_continuous_11}
\end{align}
where $(z_1,z_2)\sim\mathcal N({\bf 0}_{2\times 1},{\bf I}_{2\times 2})$. Combining \eqref{eq:C_is_continuous_10} and \eqref{eq:C_is_continuous_11} implies that $\Phi(c_u+{\delta}/{\sigma})-\Phi(-c_l)\geq 1-\alpha$ and $\Phi(c_u)-\Phi(-c_l-{\delta}/{\sigma})\geq 1-\alpha$. By \eqref{eq:C2_defn}, this is equivalent to $(c_l,c_u)\in C_2(\delta,\sigma,\sigma,1)$, proving UHC.

We now show LHC. Let $\lim_{M\to \infty}  (\delta_M,\sigma_{M,l},\sigma_{M,u},\rho_M) = (\delta,\sigma,\sigma,1)$ and fix $(c_l,c_u)\in C_2(\delta,\sigma,\sigma,1)$. For each $k\in\mathbb N$, let $\tilde{c}_{k,l}=c_l+1/k$ and $\tilde{c}_{k,u}=c_u+1/k$. Then $\lim_{k\to \infty}  (\tilde{c}_{k,l},\tilde{c}_{k,u})=(c_l,c_u)$, $\Phi(\tilde{c}_{k,u}+\tfrac{\delta}{\sigma})-\Phi(-\tilde{c}_{k,l})>1-\alpha$, and $\Phi(\tilde{c}_{k,u})-\Phi(-\tilde{c}_{k,l}-\tfrac{\delta}{\sigma})>1-\alpha$. Moreover, for each fixed $k\in\mathbb N$, $\lim_{M\to \infty}  (\delta_M,\sigma_{M,l},\sigma_{M,u},\rho_M)= (\delta,\sigma,\sigma,1)$ and the dominated convergence theorem give
\begin{equation*}
\begin{aligned}
&\lim_{M \to \infty} P\big(\{-\tilde{c}_{k,l}\leq z_1\}\cap \big\{\rho_M z_1\leq \tilde{c}_{k,u}+\tfrac{\delta_M}{\sigma_{M,u}} +\sqrt{1-\rho_M^2}z_2\big\}\big)~=~\Phi(\tilde{c}_{k,u}+\tfrac{\delta}{\sigma})-\Phi(-\tilde{c}_{k,l})~>1-\alpha,\\
&\lim_{M \to \infty}  P\big(\big\{-\tilde{c}_{k,l}-\tfrac{\delta_M}{\sigma_{M,l}}+\sqrt{1-\rho_M^2}z_2\leq \rho_M z_1\big\}\cap\{z_1\leq \tilde{c}_{k,u}\}\big)~=~\Phi(\tilde{c}_{k,u})-\Phi(-\tilde{c}_{k,l}-\tfrac{\delta}{\sigma})~>1-\alpha,
\end{aligned}
\end{equation*}
where $(z_1,z_2)\sim\mathcal N({\bf 0}_{2\times 1},{\bf I}_{2\times 2})$. Hence, for each fixed $k\in \mathbb{N}$, there exists $M_k\in \mathbb{N}$ such that $(\tilde{c}_{k,l},\tilde{c}_{k,u})\in C_2(\delta_M,\sigma_{M,l},\sigma_{M,u},\rho_M)$ for all $M\geq M_k$. Without loss of generality, we can take $M_k \in \mathbb{N}$ strictly increasing in $k$. Then, proceed as follows for every $M \in \mathbb{N}$. For $M <M_1$, set $(c_{M,l},c_{M,u})= (\infty,\infty)$. For $M \in [M_k,M_{k+1})$ with $k\in \mathbb{N}$, set $(c_{M,l},c_{M,u})=(\tilde{c}_{k,l},\tilde{c}_{k,u})$. By construction, this means that $(c_{M,l},c_{M,u}) \in C_2(\delta_M,\sigma_{M,l},\sigma_{M,u},\rho_M)$ for all $M \in \mathbb{N}$. Moreover, since $M_k$ is strictly increasing, the index $k$ corresponding to the interval $M\in [M_k,M_{k+1})$ satisfies $k\to\infty$ as $M\to\infty$. Therefore, $\lim_{M\to \infty} (c_{M,l},c_{M,u})= \lim_{M\to \infty} (\tilde{c}_{M_k,l},\tilde{c}_{M_k,u})=(c_l,c_u)$. Thus, we have constructed a sequence $(c_{M,l},c_{M,u})\in C_2(\delta_M,\sigma_{M,l},\sigma_{M,u},\rho_M)$ with $\lim_{M\to \infty}  (c_{M,l},c_{M,u})=(c_l,c_u)$, which proves LHC.
\end{proof}

\begin{lemma}\label{lem:C_is_UHC_3} 
The correspondence $C_3:\bar{\mathbb{R}}_{+}\times \bar{ \mathbb{R}}\times [\underline{\sigma },\overline{ \sigma }]\times [ \underline{\sigma },\overline{\sigma }]\times [-1,1]  \rightrightarrows {\bar{\mathbb{R}} ^{2}}$ in \eqref{eq:C3_defn} is continuous (i.e., UHC and LHC) at the following values:
\begin{enumerate}[(a)]
\item $(\tilde{\delta}_{1},\tilde{\delta}_{2},\tilde{\sigma}_{l},\tilde{ \sigma}_{u},\tilde{\rho})=(\infty ,\eta ,\sigma_l ,\sigma_u ,\rho)$ with $(\eta ,\sigma_l, \sigma_u, \rho )\in (\bar{\mathbb{R}}\backslash \{0\})\times [\underline{\sigma }, \overline{\sigma }] \times [\underline{\sigma }, \overline{\sigma }] \times [-1, 1].$
\item $(\tilde{\delta}_{1},\tilde{\delta}_{2},\tilde{\sigma}_{l},\tilde{ \sigma}_{u},\tilde{\rho})=(\delta ,-\infty ,\sigma ,\sigma ,1)$ with $ (\delta ,\sigma )\in \bar{\mathbb{R}}_{+}\times [\underline{\sigma }, \overline{\sigma }].$
\end{enumerate}
\end{lemma}
\begin{proof}
By \eqref{eq:C2_defn} and \eqref{eq:C3_defn}, $C_3(\delta_1,\delta_2,\sigma_l,\sigma_u,\rho)
= C_2(\delta_1 1[\delta_2>0],\sigma_l,\sigma_u,\rho)$ for all $(\delta_1,\delta_2,\sigma_l,\sigma_u,\rho)$. We divide the rest of the argument into parts.

\noindent\underline{Part (a).} Consider
$(\tilde{\delta}_{1},\tilde{\delta}_{2},\tilde{\sigma}_{l},\tilde{\sigma}_{u},\tilde{\rho})=(\infty,\eta,\sigma_l,\sigma_u,\rho)$ with
$(\eta,\sigma_l,\sigma_u,\rho)\in(\bar{\mathbb R}\setminus\{0\})\times[\underline{\sigma},\overline{\sigma}]\times[\underline{\sigma},\overline{\sigma}]\times[-1,1]$.
Let $\{(\delta_{M,1},\delta_{M,2},\sigma_{M,l},\sigma_{M,u},\rho_M)\}_{M \in  \mathbb{N}}$ be a convergence sequence that satisfies $ (\delta_{M,1},\delta_{M,2},\sigma_{M,l},\sigma_{M,u},\rho_M)
\to (\infty,\eta,\sigma_l,\sigma_u,\rho)$ as $M\to\infty$.

First, suppose that $\eta>0$. Then $\delta_{M,2}>0$ for all sufficiently large $M$, and so $\delta_{M,1}1[\delta_{M,2}>0]=\delta_{M,1}\to\infty $. 
Therefore, for all sufficiently large $M$, $C_3(\delta_{M,1},\delta_{M,2},\sigma_{M,l},\sigma_{M,u},\rho_M)=  C_2(\delta_{M,1},\sigma_{M,l},\sigma_{M,u},\rho_M)$,
and $C_3(\infty,\eta,\sigma_l,\sigma_u,\rho)= C_2(\infty,\sigma_l,\sigma_u,\rho)$. 
Thus, the UHC and LHC of $C_3$ follow from Part (a) of Lemma \ref{lem:C2_is_continuous}.

Now suppose that $\eta<0$. Then $\delta_{M,2}\leq 0$ for all sufficiently large $M$, and so $\delta_{M,1}1[\delta_{M,2}>0]=0$. Hence, for all sufficiently large $M$, $C_3(\delta_{M,1},\delta_{M,2},\sigma_{M,l},\sigma_{M,u},\rho_M)=C_2(0,\sigma_{M,l},\sigma_{M,u},\rho_M)$ and $C_3(\infty,\eta,\sigma_l,\sigma_u,\rho)
=C_2(0,\sigma_l,\sigma_u,\rho).$

We first show UHC in this case. Let
$(c_{M,l},c_{M,u})\in C_3(\delta_{M,1},\delta_{M,2},\sigma_{M,l},\sigma_{M,u},\rho_M)$ with
\\$\lim_{M\to \infty}  (c_{M,l},c_{M,u}) = (c_l,c_u)$. By \eqref{eq:C3_defn}, membership implies that, for all sufficiently large $M$,
\begin{align}
&P\big(\{-c_{M,l}\leq z_1\}\cap\{\rho_M z_1\leq c_{M,u}+\sqrt{1-\rho_M^2}z_2\}\big)~\geq~1-\alpha,\notag\\
&P\big(\{-c_{M,l}+\sqrt{1-\rho_M^2}z_2\leq \rho_M z_1\}\cap\{z_1\leq c_{M,u}\}\big)~\geq~1-\alpha .\label{eq:C3_cont_eta_neg_1}
\end{align}
Moreover, the dominated convergence theorem gives
\begin{align}
&\lim_{M\to\infty}P\big(\{-c_{M,l}\leq z_1\}\cap\{\rho_M z_1\leq c_{M,u}+\sqrt{1-\rho_M^2}z_2\}\big)~=~P\big(\{-c_l\leq z_1\}\cap\{\rho z_1\leq c_u+\sqrt{1-\rho^2}z_2\}\big),\notag\\
&\lim_{M\to\infty}P\big(\{-c_{M,l}+\sqrt{1-\rho_M^2}z_2\leq \rho_M z_1\}\cap\{z_1\leq c_{M,u}\}\big)~=~P\big(\{-c_l+\sqrt{1-\rho^2}z_2\leq \rho z_1\}\cap\{z_1\leq c_u\}\big),\label{eq:C3_cont_eta_neg_2}
\end{align}
where $(z_1,z_2)\sim\mathcal N({\bf 0}_{2\times 1},{\bf I}_{2\times 2})$. Combining \eqref{eq:C3_cont_eta_neg_1} and \eqref{eq:C3_cont_eta_neg_2} implies that
$(c_l,c_u)\in C_2(0,\sigma_l,\sigma_u,\rho)$. Since
$C_3(\infty,\eta,\sigma_l,\sigma_u,\rho)=C_2(0,\sigma_l,\sigma_u,\rho)$, this proves UHC.

We now show LHC in the case with $\eta<0$. Fix
$(c_l,c_u)\in C_3(\infty,\eta,\sigma_l,\sigma_u,\rho)=C_2(0,\sigma_l,\sigma_u,\rho)$. For each $k\in\mathbb N$, set
$\tilde c_{k,l}=c_l+1/k$ and $\tilde c_{k,u}=c_u+1/k$. Then
$\lim_{k\to \infty}  (\tilde c_{k,l},\tilde c_{k,u})=(c_l,c_u)$, and the two defining inequalities for $C_2(0,\sigma_l,\sigma_u,\rho)$ hold strictly at $(\tilde c_{k,l},\tilde c_{k,u})$. Moreover, for each fixed $k\in\mathbb N$, the dominated convergence theorem gives
\begin{align*}
&\lim_{M\to\infty}P\big(\{-\tilde c_{k,l}\leq z_1\}\cap\{\rho_M z_1\leq \tilde c_{k,u}+\sqrt{1-\rho_M^2}z_2\}\big)\\
&~=~P\big(\{-\tilde c_{k,l}\leq z_1\}\cap\{\rho z_1\leq \tilde c_{k,u}+\sqrt{1-\rho^2}z_2\}\big)~>~1-\alpha,
\end{align*}
and 
\begin{align*}
&\lim_{M\to\infty}P\big(\{-\tilde c_{k,l}+\sqrt{1-\rho_M^2}z_2\leq \rho_M z_1\}\cap\{z_1\leq \tilde c_{k,u}\}\big)\\
&~=~P\big(\{-\tilde c_{k,l}+\sqrt{1-\rho^2}z_2\leq \rho z_1\}\cap\{z_1\leq \tilde c_{k,u}\}\big)~>~1-\alpha.
\end{align*}
Hence, for each fixed $k\in\mathbb N$, there exists $M_k\in\mathbb N$ such that
$(\tilde c_{k,l},\tilde c_{k,u})\in C_3(\delta_{M,1},\delta_{M,2},\sigma_{M,l},\sigma_{M,u},\rho_M)$ for all $M\geq M_k$. Without loss of generality, we can take $M_k\in\mathbb N$ strictly increasing in $k$. For $M<M_1$, set $(c_{M,l},c_{M,u})=(\infty,\infty)$. For $M\in[M_k,M_{k+1})$ with $k\in\mathbb N$, set
$(c_{M,l},c_{M,u})=(\tilde c_{k,l},\tilde c_{k,u})$. By construction, $(c_{M,l},c_{M,u})\in C_3(\delta_{M,1},\delta_{M,2},\sigma_{M,l},\sigma_{M,u},\rho_M)$ for all $M\in\mathbb N$. Moreover, since $M_k$ is strictly increasing, the index $k$ corresponding to the interval $M\in[M_k,M_{k+1})$ satisfies $k\to\infty$ as $M\to\infty$. Therefore, $\lim_{M\to \infty} (c_{M,l},c_{M,u})=(\tilde c_{M_k,l},\tilde c_{M_k,u}) = (c_l,c_u)$. This proves LHC, completing the proof of part (a).

\noindent\underline{Part (b).} Consider
$(\tilde{\delta}_{1},\tilde{\delta}_{2},\tilde{\sigma}_{l},\tilde{\sigma}_{u},\tilde{\rho})=(\delta,-\infty,\sigma,\sigma,1)$ with
$(\delta,\sigma)\in\bar{\mathbb R}_{+}\times[\underline{\sigma},\overline{\sigma}]$.
Let \\
$\lim_{M\to \infty}  (\delta_{M,1},\delta_{M,2},\sigma_{M,l},\sigma_{M,u},\rho_M) = (\delta,-\infty,\sigma,\sigma,1)$. Since $\delta_{M,2}\leq 0$ for all sufficiently large $M$, $\delta_{M,1}1[\delta_{M,2}>0]=0$ for all sufficiently large $M$. Therefore, for all sufficiently large $M$,
$C_3(\delta_{M,1},\delta_{M,2},\sigma_{M,l},\sigma_{M,u},\rho_M)=C_2(0,\sigma_{M,l},\sigma_{M,u},\rho_M)$ and $C_3(\delta,-\infty,\sigma,\sigma,1)=C_2(0,\sigma,\sigma,1)$. The UHC and LHC of $C_3$ at this point therefore follow directly from Part (b) of Lemma \ref{lem:C2_is_continuous}, applied with $\delta=0$.
\end{proof}

\begin{lemma}\label{lem:Berge}
If the correspondence $C_2:\bar{\mathbb{R}}_{+}\times [\underline{\sigma },\overline{ \sigma }]\times [\underline{\sigma },\overline{\sigma }]\times [-1,1] \rightrightarrows {\bar{\mathbb{R}}^{2}}$ in \eqref{eq:C2_defn} is continuous at $(\tilde{\delta},\tilde{\sigma}_{l},\tilde{ \sigma}_{u},\tilde{\rho})$ and the correspondence $S_2:\bar{\mathbb{R}}_{+}\times [ \underline{\sigma },\overline{\sigma }]\times [  \underline{\sigma },\overline{\sigma }]\times [ -1,1]  \rightrightarrows {\bar{\mathbb{R}}^{2}}$ in \eqref{eq:S2_defn} is a singleton at $(\tilde{\delta}, \tilde{\sigma}_{l},\tilde{\sigma}_{u},\tilde{\rho})$, the function $F_2:\bar{\mathbb{R}}_{+}\times [ \underline{\sigma },\overline{\sigma }]\times [ \underline{\sigma },\overline{\sigma }]\times [ -1,1]  \to {\bar{\mathbb{R}}^{2}}$ that arbitrarily selects a solution from $S_2$ is continuous at $(\tilde{\delta}, \tilde{\sigma}_{l},\tilde{\sigma}_{u},\tilde{\rho})$.
\end{lemma}
\begin{proof}
This result follows by restricting the second coordinate of the argument in $C_3$, $S_3$, and $F_3$ in Lemma \ref{lem:Berge2} to $(0,\infty]$.
\end{proof}

\begin{lemma}\label{lem:Berge2}
If the correspondence $C_3:\bar{\mathbb{R}}_{+}\times \bar{ \mathbb{R}}\times [\underline{\sigma },\overline{ \sigma }]\times [ \underline{\sigma },\overline{\sigma }]\times [-1,1]  \rightrightarrows {\bar{\mathbb{R}} ^{2}}$ in \eqref{eq:C3_defn} is continuous at $(\tilde{\delta}_{1},\tilde{\delta}_{2},\tilde{\sigma}_{l},\tilde{\sigma}_{u},\tilde{\rho})$ and the correspondence $S_3:\bar{\mathbb{R}}_{+}\times \bar{\mathbb{R}}\times [ \underline{\sigma },\overline{\sigma }]\times [ \underline{\sigma },\overline{\sigma }]\times [-1,1] \rightrightarrows \bar{\mathbb{R}}^{2}$ in \eqref{eq:S3_defn} is a singleton at $(\tilde{\delta}_{1},\tilde{\delta}_{2},\tilde{\sigma}_{l},\tilde{\sigma}_{u},\tilde{\rho})$, the function $F_3:\bar{\mathbb{R}}_{+}\times \bar{\mathbb{R}}\times [ \underline{\sigma },\overline{\sigma }]\times [ \underline{\sigma },\overline{\sigma }]\times [-1,1]  \to \bar{\mathbb{R}}^{2}$  that arbitrarily selects a solution from $S_3$ is continuous at $(\tilde{\delta}_{1},\tilde{\delta}_{2},\tilde{\sigma}_{l},\tilde{\sigma}_{u},\tilde{\rho})$.
\end{lemma}
\begin{proof}
We divide the proof into two parts.

\noindent \underline{Part 1.} Show that $S_3$ is UHC at $(\tilde{\delta}_{1},\tilde{\delta}_{2},\tilde{\sigma}_{l},\tilde{ \sigma}_{u},\tilde{\rho})$.

We prove this by contradiction: we assume that there is $\{(\delta _{M,1},\delta _{M,2},\sigma _{M,l},\sigma _{M,u},\rho _{M})\}_{M \in \mathbb{N}}$ that satisfies $( \delta _{M,1}, \delta _{M,2}, \sigma _{M,l}, \sigma _{M,u}, \rho _{M}) \to (\tilde{\delta}_{1},\tilde{\delta}_{2},\tilde{\sigma}_{l},\tilde{\sigma}_{u},\tilde{\rho})$ as $M\to\infty$, and $\{(c_{M,l},c_{M,u})\}_{M \in \mathbb{N}}$ with $(c_{M,l},c_{M,u})\in S_3(\delta _{M,1},\delta _{M,2},\sigma _{M,l},\sigma _{M,u},\rho _{M})$ for all $M\in \mathbb{N}$ with $(c_{M,l},c_{M,u}) \to (c_{l}^{\ast },c_{u}^{\ast })$ as $M\to\infty$, and $(c_{l}^{\ast },c_{u}^{\ast })\not\in S_3(\tilde{\delta}_{1},\tilde{\delta}_{2},\tilde{\sigma}_{l},\tilde{\sigma}_{u},\tilde{\rho})$. Since $S_3(\tilde{\delta}_{1},\tilde{\delta}_{2},\tilde{\sigma}_{l},\tilde{\sigma}_{u},\tilde{\rho})$ is a singleton, we have $(c_{l}^{\ast },c_{u}^{\ast })\neq (\tilde{c}_{l},\tilde{c}_{u})$. Since $S_3(\tilde{\delta}_{1},\tilde{\delta}_{2},\tilde{\sigma}_{l},\tilde{\sigma}_{u},\tilde{\rho})=\{(\tilde{c}_{l},\tilde{c}_{u})\}$, we have that $(\tilde{c}_{l},\tilde{c}_{u})\in C_3(\tilde{\delta}_{1},\tilde{\delta}_{2},\tilde{\sigma}_{l},\tilde{\sigma}_{u},\tilde{\rho})$. Since $C_3$ is LHC at $(\tilde{\delta}_{1},\tilde{\delta}_{2},\tilde{\sigma}_{l},\tilde{\sigma}_{u},\tilde{\rho})$, we can find a subsequence $\{k_{M}\}_{M \in \mathbb{N}}$ s.t.\ $\{(\delta _{k_{M},1},\delta _{k_{M},2},\sigma _{k_{M},l},\sigma _{k_{M},u},\rho _{k_{M}})\}_{M\geq 1}$ with $ (\delta _{k_{M},1},\delta _{k_{M},2},\sigma _{k_{M},l},\sigma _{k_{M},u},\rho _{k_{M}}) \to(\tilde{\delta}_{1},\tilde{\delta}_{2},\tilde{\sigma}_{l},\tilde{\sigma}_{u},\tilde{\rho})$ as $M\to\infty$, and $\{(\tilde{c}_{k_{M},l},\tilde{c}_{k_{M},u})\}_{M \in \mathbb{N}}$ with $(\tilde{c}_{k_{M},l},\tilde{c}_{k_{M},u})\in C_3(\delta _{k_{M},1},\delta _{k_{M},2},\sigma _{k_{M},l},\sigma _{k_{M},u},\rho _{k_{M}})$ for all $M\in \mathbb{N}$ with $ (\tilde{c}_{k_{M},l},\tilde{c}_{k_{M},u}) \to (\tilde{c}_{l},\tilde{c}_{u})$ as $M\to\infty$. We focus on this subsequence for the remainder of the part. Along this subsequence, we have
\begin{equation}
\lim_{M\to\infty} \sigma _{k_{M},l}c_{k_{M},l}+\sigma _{k_{M},u}c_{k_{M},u}~  = ~\tilde{\sigma}_{l}c_{l}^{\ast }+\tilde{\sigma}_{u}c_{u}^{\ast }.
\label{eq:Berge1}
\end{equation}
Since $(c_{k_{M},l},c_{k_{M},u})\in S_3(\delta _{k_{M},1},\delta _{k_{M},2},\sigma _{k_{M},l},\sigma _{k_{M},u},\rho _{k_{M}})\subseteq C_3(\delta _{k_{M},1},\delta _{k_{M},2},\sigma _{k_{M},l},\sigma _{k_{M},u},\rho _{k_{M}})$ for all $M\in \mathbb{N}$, $(c_{k_{M},l},c_{k_{M},u})\in C_3(\delta _{k_{M},1},\delta _{k_{M},2},\sigma _{k_{M},l},\sigma _{k_{M},u},\rho _{k_{M}})$ for all $M\in \mathbb{N}$. Thus, we have that $\{(\delta _{k_{M},1},\delta _{k_{M},2},\sigma _{k_{M},l},\sigma _{k_{M},u},\rho _{k_{M}})\}_{M\geq 1}$ with $\lim_{M\to\infty} (\delta _{k_{M},1},\delta _{k_{M},2},\sigma _{k_{M},l},\sigma _{k_{M},u},\rho _{k_{M}})  = (\tilde{\delta}_{1},\tilde{\delta}_{2},\tilde{\sigma}_{l},\tilde{\sigma}_{u},\tilde{\rho})$ and $\{(c_{k_{M},l},c_{k_{M},u})\}_{M\geq 1}$ with $(c_{k_{M},l},c_{k_{M},u})\in C_3(\delta _{k_{M},1},\delta _{k_{M},2},\sigma _{k_{M},l},\sigma _{k_{M},u},\rho _{k_{M}})$ for all $M\in \mathbb{N}$ with $\lim_{M\to\infty} (c_{k_{M},l},c_{k_{M},u}) =(c_{l}^{\ast },c_{u}^{\ast })$. This implies that 
\begin{equation}
\lim_{M\to\infty}\sigma _{k_{M},l}\tilde{c}_{k_{M},l}+\sigma _{k_{M},u}\tilde{c} _{k_{M},u}~ =~\tilde{\sigma}_{l}\tilde{c}_{l}+\tilde{\sigma}_{u} \tilde{c}_{u}.  \label{eq:Berge2}
\end{equation}
Since $C_3$ is UHC at $(\tilde{\delta}_{1},\tilde{\delta}_{2},\tilde{\sigma}_{l},\tilde{\sigma}_{u},\tilde{\rho})$, $(c_{l}^{\ast },c_{u}^{\ast })\in C_3(\tilde{\delta}_{1},\tilde{\delta}_{2},\tilde{\sigma}_{l},\tilde{\sigma}_{u},\tilde{\rho})$. By $(c_{l}^{\ast },c_{u}^{\ast })\in C_3(\tilde{\delta}_{1},\tilde{\delta}_{2},\tilde{\sigma}_{l},\tilde{\sigma}_{u},\tilde{\rho})$ and $(c_{l}^{\ast },c_{u}^{\ast })\not\in S_3(\tilde{\delta}_{1},\tilde{\delta}_{2},\tilde{\sigma}_{l},\tilde{\sigma}_{u},\tilde{\rho})=\{(\tilde{c}_{l},\tilde{c}_{u})\}$, we deduce that
\begin{equation}
\tilde{\sigma}_{l}c_{l}^{\ast }+\tilde{\sigma}_{u}c_{u}^{\ast }~>~\tilde{\sigma}_{l}\tilde{c}_{l}+\tilde{\sigma}_{u}\tilde{c}_{u}.  \label{eq:Berge3}
\end{equation}
By \eqref{eq:Berge1}, \eqref{eq:Berge2}, and \eqref{eq:Berge3}, we deduce that $\exists \bar{M}$ s.t.\ 
\begin{equation}
\sigma _{k_{M},l}c_{k_{M},l}+\sigma _{k_{M},u}c_{k_{M},u}~>~\sigma _{k_{M},l}\tilde{c}_{k_{M},l}+\sigma _{k_{M},u}\tilde{c}_{k_{M},u}~~~\text{for all}~M\geq \bar{M}.  \label{eq:Berge4}
\end{equation}
However, \eqref{eq:Berge4} contradicts $(c_{k_{M},l},c_{k_{M},u}) \in S_3(\delta _{k_{M},1},\delta _{k_{M},2},\sigma _{k_{M},l},\sigma _{k_{M},u},\rho _{k_{M}})$ and $(\tilde{c}_{k_{M},l},\tilde{c}_{k_{M},u})\in C_3(\delta _{k_{M},1},\delta _{k_{M},2},\sigma _{k_{M},l},\sigma _{k_{M},u},\rho _{k_{M}})$ for all $M\in \mathbb{N}$. 

\noindent \underline{Part 2.} Show that $F_3:\bar{\mathbb{R}}_{+}\times \bar{\mathbb{R}} \times [ \underline{\sigma },\overline{\sigma }]\times [ \underline{ \sigma },\overline{\sigma }]\times [ -1,1]  \to {\bar{\mathbb{R}} ^{2}}$ is continuous at $(\tilde{\delta}_{1},\tilde{\delta}_{2},\tilde{\sigma }_{l},\tilde{\sigma}_{u},\tilde{\rho})$.

We prove this by contradiction: we assume $\{(\delta _{k_{M},1},\delta _{k_{M},2},\sigma _{k_{M},l},\sigma _{k_{M},u},\rho _{k_{M}})\}_{M\in \mathbb{N}}$ satisfies $ (\delta _{k_{M},1},\delta _{k_{M},2},\sigma _{k_{M},l},\sigma _{k_{M},u},\rho _{k_{M}}) \to (\tilde{\delta}_{1},\tilde{\delta}_{2},\tilde{\sigma}_{l},\tilde{\sigma}_{u},\tilde{\rho})$ as $M\to\infty$, but the sequence $\{F_3(\delta _{k_{M},1},\delta _{k_{M},2},\sigma _{k_{M},l},\sigma _{k_{M},u},\rho _{k_{M}})\}$ does not converge to $F_3(\tilde{\delta}_{1},\tilde{\delta}_{2},\tilde{\sigma}_{l},\tilde{\sigma}_{u},\tilde{\rho})=(\tilde c_l,\tilde c_u)$ as $M\to\infty$. By this and the compactness of $\bar{\mathbb R}^{2}$, it follows that there is a subsequence $\{a_M\}_{M\in\mathbb{N}}$ of $\{k_M\}_{M\in\mathbb{N}}$ s.t.\ $F_3(\delta _{a_{M},1},\delta _{a_{M},2},\sigma _{a_{M},l},\sigma _{a_{M},u},\rho _{a_{M}}) \to (c_l^*,c_u^*)\neq (\tilde c_l,\tilde c_u)$ as $M\to\infty$. 
Since $F_3$ selects from $S_3$, $(c_l^*,c_u^*)$ is the limit of a sequence with each element in $S_3(\delta _{a_{M},1},\delta _{a_{M},2},\sigma _{a_{M},l},\sigma _{a_{M},u},\rho _{a_{M}})$.

By part 1, $S_3$ is UHC at $(\tilde{\delta}_{1},\tilde{\delta}_{2},\tilde{\sigma}_{l},\tilde{\sigma}_{u},\tilde{\rho})$. This implies that $(c_{l}^{\ast },c_{u}^{\ast })\in S_3(\tilde{\delta}_{1},\tilde{\delta}_{2},\tilde{\sigma}_{l},\tilde{\sigma}_{u},\tilde{\rho})$. Since $S_3(\tilde{\delta}_{1},\tilde{\delta}_{2},\tilde{\sigma}_{l},\tilde{\sigma}_{u},\tilde{\rho})=\{(\tilde{c}_{l},\tilde{c}_{u})\}$, we conclude that $(c_{l}^{\ast },c_{u}^{\ast })=(\tilde{c}_{l},\tilde{c}_{u})$, which is a contradiction.
\end{proof}

\begin{lemma}\label{lem:C3NotEmpty}
If $(\hat{\theta}_{l},\hat{\theta} _{u},\hat{\sigma}_{l},\hat{\sigma}_{u},\hat{\rho})$ satisfies OBS and $\alpha  \in (0, 0.5)$,
\begin{equation}
P\big( ~CI_{\alpha }^{3}=[ \hat{\theta}_{l}-{\hat{\sigma}_{l}c_{l}^{3}}/{\sqrt{N}},~\hat{\theta}_{u}+{\hat{\sigma}_{u}c_{u}^{3}}/{\sqrt{N}}] ~\big) ~=~1.\label{eq:nonEmpty}
\end{equation}
\end{lemma}
\begin{proof}
We first show that $(c_{l}^{3},c_{u}^{3})$ satisfies $c_{l}^{3},c_{u}^{3}\geq 0$. We only show $c_{l}^{3}\geq 0$, as the proof of $c_{u}^{3}\geq 0$ is analogous. Suppose that $c_{l}^{3}\geq 0$ fails, i.e., $c_{l}^{3}<0$. Then,
\begin{align}
0.5 &~\overset{(1)}{>}~P( -c_{l}^{3}\leq z_{1})   \notag \\
&~\geq ~ P\left( \{ -c_{l}^{3}\leq z_{1}\} \cap \left\{ \hat{\rho}z_{1}\leq c_{u}^{3}+\tfrac{\sqrt{N}(\hat{\theta}_{u}-\hat{\theta}_{l})1[(\hat{\theta}_{u}-\hat{\theta}_{l})>b_{N}]}{\hat{\sigma}_{u}}+\sqrt{1-\hat{\rho}^{2}}z_{2}\right\} \left|\hat{\theta}_{l},\hat{\theta}_{u},\hat{\sigma}_{l},\hat{\sigma}_{u},\hat{\rho}\right)\right.   \notag \\
&~\overset{(2)}{\geq}~ 1-\alpha ,\label{eq:nonEmpty2}
\end{align}
where (1) holds by $c_{l}^{3}<0$ and $z_{1}\sim \mathcal{N}(0,1)$, and (2) by \eqref{eq:CI3_problem}. Note that \eqref{eq:nonEmpty2} contradicts $\alpha <0.5$, which proves the desired result.

Since $(\hat{\theta}_{l},\hat{\theta} _{u},\hat{\sigma}_{l},\hat{\sigma}_{u},\hat{\rho})$ satisfies OBS, we have that $P(\hat{\theta}_{l}\leq \hat{\theta}_{u})=1$ and $\hat{\sigma}_{l},\hat{\sigma}_{u} \geq 0$. If we combine these with $c_{l}^{3},c_{u}^{3}\geq 0$, we get
\begin{equation}
P( \hat{\theta}_{l}-{\hat{\sigma}_{l}c_{l}^{3}}/{\sqrt{N}}\leq \hat{\theta}_{u}+{\hat{\sigma}_{u}c_{u}^{3}}/{\sqrt{N}}) ~=~1.\label{eq:nonEmpty3}
\end{equation}
To conclude the proof, note that \eqref{eq:CI3} and \eqref{eq:nonEmpty3} imply \eqref{eq:nonEmpty}, as desired.
\end{proof}

\subsection{An example that does not satisfy OBS}\label{sec:OBSfails}

\begin{example} 
\label{ex:NoOBS_Ex}
Consider a linear regression model with missing outcome data. In this example, $\{ ( Y_{i},Z_{i},X_{i}) \} _{i=1}^{N}$ are i.i.d.\ from a distribution $P$, where $Y_i\geq 0$ is a nonnegative outcome variable, $X_i$ is a scalar regressor with support $S_X=\{-1,1\}$, and $Z_i\in\{0,1\}$ is a binary variable that indicates whether $Y_i$ is observed ($Z_i=1$) or not ($Z_i=0$). Assume that these variables satisfy the following linear regression model:
\begin{equation}
Y_i~=~1+X_i\theta+\varepsilon_i~~\text{with}~E_P[\varepsilon_i|X_i]~=~0,
\label{eq:NoOBS_Ex0}
\end{equation}
where $\theta\in\Theta=\mathbb R$ is the parameter of interest and $\varepsilon_i\in\mathbb R$ is the unobserved regression residual. The data are not assumed to be missing at random, i.e., $E_P[\varepsilon_i|X_i,Z_i]$ is not necessarily zero.

In this framework, $Y_i$ is only observed when $Z_i=1$. Thus, the available data are the i.i.d.\ sample $\{(Y_iZ_i,Z_i,X_i)\}_{i=1}^N$. The econometric model implies the following moment inequalities:
\begin{equation}
1+x\theta
~\overset{(1)}{=}~
E_P[Y_iZ_i+(1-Z_i)Y_i | X_i=x]
~\overset{(2)}{\geq}~
E_P[Y_iZ_i | X_i=x]
\quad\text{for }x\in S_X,
\label{eq:NoOBS_Ex1}
\end{equation}
where (1) holds by \eqref{eq:NoOBS_Ex0}, and (2) by $Y_i\geq0$. By evaluating \eqref{eq:NoOBS_Ex1} at $x\in S_X=\{-1,1\}$, we deduce that the identified set for $\theta$ is given by $\Theta_I(P)=[\theta_l(P),\theta_u(P)],$
where $\theta_l(P)=E_P[Y_iZ_i| X_i=1]-1$ and $\theta_u(P)=1-E_P[Y_iZ_i| X_i=-1]$. We also assume that $P(Z_i=1,X_i=x)>0$ and $V_P[Y_iZ_i|X_i=x]\in(0,\infty)$ for all $x\in\{-1,1\}$.

In this context, it is natural to estimate the bounds using sample analogs:
\begin{equation*}
(\hat\theta_l,\hat\theta_u)
~=~
\left(
\tfrac{\sum_{i=1}^N Y_iZ_i I\{X_i=1\}}{\sum_{i=1}^N I\{X_i=1\}}-1,
~
1-\tfrac{\sum_{i=1}^N Y_iZ_i I\{X_i=-1\}}{\sum_{i=1}^N I\{X_i=-1\}}
\right).
\end{equation*}
Standard asymptotic arguments imply
\begin{align}
&\sqrt{N}(\hat\theta_l-\theta_l(P),\hat\theta_u-\theta_u(P)) \notag \\
&~=~
\left(
\tfrac{\frac{1}{\sqrt N}\sum_{i=1}^N(Y_iZ_i-E_P[Y_iZ_i \mid X_i=1])I\{X_i=1\}}{\frac{1}{N}\sum_{i=1}^N I\{X_i=1\}},
~
-\tfrac{\frac{1}{\sqrt N}\sum_{i=1}^N(Y_iZ_i-E_P[Y_iZ_i \mid X_i=-1])I\{X_i=-1\}}{\frac{1}{N}\sum_{i=1}^N I\{X_i=-1\}}
\right) \notag \\
&~\overset{d}{\to}~
N\left(
\mathbf 0_{2\times1},
\left(
\begin{array}{cc}
V_P[Y_iZ_i \mid X_i=1]/P(X_i=1) & 0 \\
0 & V_P[Y_iZ_i \mid X_i=-1]/P(X_i=-1)
\end{array}
\right)
\right).
\label{eq:NoOBS_Ex2}
\end{align}

It is not hard to verify that OBS can fail in this econometric model. For example, consider the case with $E_P[Y_iZ_i | X_i=1]-1=1-E_P[Y_iZ_i | X_i=-1]$, 
or, equivalently, $\theta_l(P)=\theta_u(P)$, so the model is point identified. Then,
\begin{align}
P(\hat\theta_l\leq\hat\theta_u)~=~P(\sqrt N(\hat\theta_l-\theta_l(P))\leq\sqrt N(\hat\theta_u-\theta_u(P))) ~\overset{(1)}{\to}~\frac{1}{2},
\label{eq:NoOBS_Ex3}
\end{align}
where (1) holds by \eqref{eq:NoOBS_Ex2}. Note that \eqref{eq:NoOBS_Ex3} implies that $P(\hat\theta_l\leq\hat\theta_u)<1$ for all $N$ large enough.
\end{example}

\section{Online supplement}

\subsection{Proof of Theorems}

\begin{proof}[Proof of part (c) of Theorem \ref{thm:CIcomparison}] We prove \eqref{eq:CI4main4} by contradiction. That is, suppose that ${\lim \inf}_{N\to\infty } (P_{N}(\theta
_{N}\in CI_{\alpha }^{4})-P_{N}(\theta _{N}\in CI_{\alpha }^{j}))<0$ for some $j=1,2,3$. We can then find a subsequence $\{k_{N}\}_{N\in \mathbb{N}}$ s.t. 
\begin{equation}
\lim_{N \to \infty} (P_{k_{N}}(\theta _{k_{N}}\in CI_{\alpha }^{j})-P_{k_{N}}(\theta _{k_{N}}\in CI_{\alpha }^{4})) ~ >~0~~\text{ for some }j=1,2,3.
\label{eq:CIcomparison_8}
\end{equation}
The proof is completed by showing that \eqref{eq:CIcomparison_8} cannot hold.

By possibly taking a further subsequence, 
\begin{align}
&\Bigg( 
\begin{array}{c}
\theta _{l}( P_{k_{N}}) ,\theta _{u}( P_{k_{N}}) ,\sigma _{l}( P_{k_{N}}) ,\sigma _{u}( P_{k_{N}}) ,\rho ( P_{k_{N}}) , \sqrt{k_{N}}( \theta _{u}( P_{k_{N}}) -\theta _{l}( P_{k_{N}}) ), \\ 
\sqrt{k_{N}}(\theta _{u}(P_{k_{N}})-\theta _{l}(P_{k_{N}})-b_{k_{N}}) ,\sqrt{ k_{N} }( \theta _{l}( P_{k_{N}}) -\theta _{k_{N}}) ,\sqrt{k_{N}}( \theta _{k_{N}}-\theta _{u}( P_{k_{N}}) )
\end{array}
\Bigg)  \notag \\
&~\to~ ( \theta _{l},\theta _{u},\sigma _{l},\sigma _{u},\rho ,\mu ,\eta ,\Psi _{l},\Psi _{u}) .  \label{eq:CIcomparison_9}
\end{align}
where $\{b_{N}\}_{N \geq 1}$ is the tuning parameter sequence used to implement $CI_{\alpha }^{3}$.

We then divide the argument into six exhaustive cases, depending on possible values of $( \mu ,\Psi _{l},\Psi _{u}) $. In this regard, note that $\theta _{N}\in \Theta _{I}(P_{N})^{c}$ implies that either (i) $\sqrt{k_{N}}(\theta _{l}(P_{k_{N}})-\theta _{k_{N}})>0$ or (ii) $\sqrt{k_{N}}(\theta _{k_{N}}-\theta _{u}(P_{k_{N}}))>0$. By taking limits, we conclude that either (i) $\Psi _{l}\geq 0$ or (ii) $\Psi _{u}\geq 0$. The result is completed by showing that none of the following exhaustive cases satisfy \eqref{eq:CIcomparison_8}.

\noindent {Case 1:} $\mu =\infty $, $\Psi _{l}\geq 0$, and $\eta \in (0,\infty ]$. Then, for any $j=1,2,3,$ consider the following derivation.
\begin{align}
P_{k_{N}}(\theta _{k_{N}}\in CI_{\alpha }^{4})& ~\overset{(1)}{\geq }~P_{k_{N}}(\theta _{k_{N}}\in CI_{\alpha }^{4,a})  \notag \\
& ~=~P_{k_{N}}(\hat{\theta}_{l}-{\hat{\sigma}_{l}c(\hat{\rho})}/\sqrt{k_{N}}\leq \theta _{k_{N}}\leq \hat{\theta}_{u}+{\hat{\sigma}_{u}c(\hat{\rho})}/\sqrt{k_{N}})  \notag \\
& ~=~P_{k_{N}}\left( 
\begin{array}{c}
\{\sqrt{k_{N}}(\hat{\theta}_{l}-\theta _{l}(P_{k_{N}}))/\hat{\sigma}_{l}-c(\hat{\rho})\leq -\sqrt{k_{N}}(\theta _{l}(P_{k_{N}})-\theta _{k_{N}})/\hat{\sigma}_{l}\}\cap  \\ 
\left\{ 
\begin{array}{c}
-\sqrt{k_{N}}(\theta _{l}(P_{k_{N}})-\theta _{k_{N}})/\hat{\sigma}_{u}-\sqrt{k_{N}}(\theta _{u}(P_{k_{N}})-\theta _{l}(P_{k_{N}}))/\hat{\sigma}_{u}\leq \\ 
\sqrt{k_{N}}(\hat{\theta}_{u}-\theta _{u}(P_{k_{N}}))/\hat{\sigma}_{u}+c(\hat{\rho})
\end{array}
\right\} 
\end{array}
\right)   \notag \\
& ~\overset{(2)}{ \to }~P(\{z_{1}-c(\rho )\leq -\Psi _{l}/\sigma_{l}\}\cap \{-(\Psi _{l}+\mu )/\sigma _{u}\leq z_{2}\sqrt{1-\rho ^{2}}+z_{1}\rho +c(\rho )\})  \notag \\
& ~\overset{(3)}{=}~P(\{z_{1}-c(\rho )\leq -\Psi _{l}/\sigma _{l}\})  \notag\\
& ~=~\Phi (c(\rho )-\Psi _{l}/\sigma _{l})  \notag \\
& ~\overset{(4)}{\geq }~\Phi (\Phi ^{-1}(1-\alpha )-\Psi _{l}/\sigma _{l}) \notag \\
&~\overset{(5)}{=}~\lim P_{k_{N}}(\theta _{k_{N}} \in CI_{\alpha }^{j}), \label{eq:CIcomparison_10}
\end{align}
where (1) holds by $CI_{\alpha }^{4,a}\subseteq CI_{\alpha }^{4}$ with $ CI_{\alpha }^{4,a}=[\hat{\theta}_{l}-{\hat{\sigma}_{l}c(\hat{\rho})}/\sqrt{ k_{N}},\hat{\theta}_{u}+{\hat{\sigma}_{u}c(\hat{\rho})}/\sqrt{k_{N}}]$ with $ c(\cdot )$ as in Definition \ref{def:c_rho}, (2) by \eqref{eq:CIcomparison_9}, Lemma \ref{lem:cp_CTS}, and OBS (Definition \ref{def:setup}), (3) by $\mu
=\infty $, (4) by Lemma \ref{lem:cp_geq_PhiInv1malpha}, and (5) by part (a) of Lemma \ref{lem:C1_limit}, $F_{1}( \infty ,\sigma _{l},\sigma _{u}) =\Phi ^{-1}(1-\alpha )$, part (a) of Lemma \ref{lem:C2_limit}, and part (c) of Lemma \ref{lem:C3_limit}. Note that \eqref{eq:CIcomparison_10} implies that \eqref{eq:CIcomparison_8} fails.

\noindent {Case 2:} $\mu =\infty $, $\Psi _{l}\geq 0$, and $\eta \in [ -\infty ,0]$. By this and Lemma \ref{lem:near1}, we deduce that $\rho =1$ and $\sigma =\sigma _{l}=\sigma _{u}$. Then, for any $j=1,2,3,$ consider the following derivation. 
\begin{align}
P_{k_{N}}(\theta _{k_{N}}\in CI_{\alpha }^{4})& ~\overset{(1)}{\geq } ~P_{k_{N}}(\theta _{k_{N}}\in CI_{\alpha }^{4,a})  \notag \\
& ~=~P_{k_{N}}\left( 
\begin{array}{c}
\{\sqrt{k_{N}}(\hat{\theta}_{l}-\theta _{l}(P_{k_{N}}))/\hat{\sigma}_{l}-c( \hat{\rho})\leq -\sqrt{k_{N}}(\theta _{l}(P_{k_{N}})-\theta _{k_{N}})/\hat{ \sigma}_{l}\}\cap  \\ 
\left\{ 
\begin{array}{c}
-\sqrt{k_{N}}(\theta _{l}(P_{k_{N}})-\theta _{k_{N}})/\hat{\sigma}_{u}-\sqrt{ k_{N}}(\theta _{u}(P_{k_{N}})-\theta _{l}(P_{k_{N}}))/\hat{\sigma}_{u} \\ 
\leq \sqrt{k_{N}}(\hat{\theta}_{u}-\theta _{u}(P_{k_{N}}))/\hat{\sigma} _{u}+c(\hat{\rho})
\end{array}
\right\} 
\end{array}
\right)   \notag \\
& ~\overset{(2)}{ \to }~P(\{z_{1}-c(1)\leq -\Psi _{l}/\sigma \}) \notag \\
& ~\overset{(3)}{=}~P(\{z_{1}-\Phi ^{-1}(1-\alpha /2)\leq -\Psi _{l}/\sigma \})  \notag \\
& ~=~\Phi (\Phi ^{-1}(1-\alpha /2)-\Psi _{l}/\sigma )  \notag \\
& ~\overset{(4)}{\geq }~\limsup_{N\to\infty} P_{k_{N}}(\theta _{k_{N}} \in CI_{\alpha }^{j}),\label{eq:CIcomparison_11}
\end{align}
where (1) holds by $CI_{\alpha }^{4,a}\subseteq CI_{\alpha }^{4}$ with $ CI_{\alpha }^{4,a}=[\hat{\theta}_{l}-{\hat{\sigma}_{l}c(\hat{\rho})}/\sqrt{k_{N}},\hat{\theta}_{u}+{\hat{\sigma}_{u}c(\hat{\rho})}/\sqrt{k_{N}}]$ with $c(\cdot )$ as in Definition \ref{def:c_rho}, (2) by \eqref{eq:CIcomparison_9}, Lemma \ref{lem:cp_CTS}, and OBS (Definition \ref{def:setup}), (3) by Lemma  \ref{lem:c1}, and (4) by $\Phi ^{-1}(1-\alpha /2)\geq \Phi ^{-1}(1-\alpha )$, Lemma \ref{lem:C1_limit}, $F_{1}( \infty ,\sigma,\sigma ) =\Phi ^{-1}(1-\alpha )$, part (a) of Lemmas \ref{lem:C1_limit} and \ref{lem:C2_limit}, and parts (b) and (d) of Lemma \ref{lem:C3_limit}. Note that \eqref{eq:CIcomparison_11} implies that \eqref{eq:CIcomparison_8} fails.

\noindent {Case 3:} $\mu <\infty $ and $\Psi _{l}\geq 0$. By $\mu \in \mathbb{R}_{+}$, it follows that $\theta _{u}(P_{k_{N}})-\theta _{l}(P_{k_{N}})\to 0$. By this and Lemma \ref{lem:near1}, we deduce that $\rho =1$ and $\sigma =\sigma _{l}=\sigma _{u}$. Also, by repeating the argument in \eqref{eq:CIcomparison_4}, we deduce that $\Phi ^{-1}(1-\alpha /2)\geq G(\mu /\sigma )$. Then, for any $j=1,2,3,$ consider the following derivation.
\begin{align}
P_{k_{N}}(\theta _{k_{N}}\in CI_{\alpha }^{4})& ~\overset{(1)}{\geq }~P_{k_{N}}(\theta _{k_{N}}\in CI_{\alpha }^{4,a})  \notag \\
& ~=~P_{k_{N}}\left( 
\begin{array}{c}
\{\sqrt{k_{N}}(\hat{\theta}_{l}-\theta _{l}(P_{k_{N}}))/\hat{\sigma}_{l}-c(\hat{\rho})\leq -\sqrt{k_{N}}(\theta _{l}(P_{k_{N}})-\theta _{k_{N}})/\hat{\sigma}_{l}\}\cap  \\ 
\left\{ 
\begin{array}{c}
-\sqrt{k_{N}}(\theta _{l}(P_{k_{N}})-\theta _{k_{N}})/\hat{\sigma}_{u}-\sqrt{k_{N}}(\theta _{u}(P_{k_{N}})-\theta _{l}(P_{k_{N}}))/\hat{\sigma}_{u} \\ 
\leq \sqrt{k_{N}}(\hat{\theta}_{u}-\theta _{u}(P_{k_{N}}))/\hat{\sigma}_{u}+c(\hat{\rho})
\end{array}
\right\} 
\end{array}
\right)   \notag \\
& ~\overset{(2)}{ \to }~P(\{z_{1}-c(1)\leq -\Psi _{l}/\sigma \}\cap \{-(\Psi _{l}+\mu )/\sigma \leq z_{1}+c(1)\})  \notag \\
& ~\overset{(3)}{=}~\Phi ((\Psi _{l}+\mu )/\sigma +\Phi ^{-1}(1-\alpha /2))-\Phi (\Psi _{l}/\sigma -\Phi ^{-1}(1-\alpha /2))  \notag \\
&~\overset{(4)}{\geq}~\lim_{N\to\infty} P_{k_{N}}(\theta _{k_{N}} \in CI_{\alpha }^{j}),\label{eq:CIcomparison_12}
\end{align}
where (1) holds by $CI_{\alpha }^{4,a}\subseteq CI_{\alpha }^{4}$ with $ CI_{\alpha }^{4,a}=[\hat{\theta}_{l}-{\hat{\sigma}_{l}c(\hat{\rho})}/\sqrt{ k_{N}},\hat{\theta}_{u}+{\hat{\sigma}_{u}c(\hat{\rho})}/\sqrt{k_{N}}]$ with $ c(\cdot )$ as in Definition \ref{def:c_rho}, (2) by \eqref{eq:CIcomparison_9}, Lemma \ref{lem:cp_CTS}, and OBS (Definition \ref{def:setup}), (3) by Lemma \ref{lem:c1}, and (4) by $\Phi ^{-1}(1-\alpha /2)\geq G(\mu /\sigma )$, Lemma \ref{lem:C1_limit}, $F_{1}( \mu ,\sigma ,\sigma ) =G(\mu /\sigma )$, part (a) of Lemma \ref{lem:C2_limit}, and part (e) of Lemma \ref{lem:C3_limit}. Note that \eqref{eq:CIcomparison_12} implies that \eqref{eq:CIcomparison_8} fails.

\noindent {Case 4:} $\mu =\infty $, $\Psi _{u}\geq 0$, and $\eta \in (0,\infty ]$. This case is analogous to case 1 and, therefore, omitted.

\noindent {Case 5:} $\mu =\infty $, $\Psi _{u}\geq 0$, and $\eta \in [ -\infty ,0]$. This case is analogous to case 2 and, therefore, omitted.

\noindent {Case 6:} $\mu <\infty $ and $\Psi _{u}\geq 0$. This case is analogous to case 3 and, therefore, omitted.

To conclude the proof, it suffices to verify \eqref{eq:CI4main4} holds strictly for suitably chosen sequences $\{(P_{N},\theta _{N})\in \mathcal{P}\times \Theta _{I}(P_{N})^{c}\}_{N\in \mathbb{N}}$. We use Example \ref{ex:OBS_Ex} with $\underline{Y}=0$, $\overline{Y}=1$, $\{Y_{i}|Z_{i}=1\}\sim Be(1/2)$, and $Z_{i}\sim Be(1-\pi (P_{N}))$, where $\pi (P_{N})~=~(\ln N)/\sqrt{N} + a_1/\sqrt{N}~\downarrow ~0$ for some $a_1 >0$. We consider coverage of $\theta _{N}=\theta _{l}(P_{N})-a_2/ \sqrt{N}\in\Theta _{I}(P_{N})^{c}$ for some $a_2>0$ and $CI_{\alpha }^{3}$ implemented with $b_{N}=(\ln N)/\sqrt{N}$ and the full sample the data (i.e., $\lambda = 1$).

By repeating arguments presented in the proof of Theorem \ref{thm:CI1}, we deduce that the vector $(\theta _{l}, \theta _{u}, \sigma _{l}, \sigma _{u}, \rho, \mu, \eta ,\Psi _{l})= (1/2,1/2,1/2,1/2,1,\infty ,a_1, a_2)$, and part (c) of Lemma \ref{lem:C3_limit} implies that
\begin{align}
\lim_ {N \to \infty} P_{N}(\theta _{N}\in CI_{\alpha }^{3})~=~\Phi (\Phi ^{-1}(1-\alpha )-2\sqrt{\lambda } a_2).\label{eq:CIcomparison_14}
\end{align}
In turn, Lemmas \ref{lem:C1_limit}, \ref{lem:C2_limit}, \ref{lem:c1}, and \ref{lem:limits} imply that 
\begin{align}
\lim_{N\to\infty} P_{N}(\theta _{N}\in CI_{\alpha }^{j})~&=~ \Phi (\Phi ^{-1}(1-\alpha )-2\sqrt{\lambda } a_2)~~~~\text{ for }j=1,2,\notag\\
\lim_{N\to\infty} P_{N}(\theta _{N}\in CI_{\alpha }^{4})~&=~\Phi (\Phi ^{-1}(1-\alpha /2)-2\sqrt{\lambda } a_2).\label{eq:CIcomparison_15}
\end{align}
The desired result follows from \eqref{eq:CIcomparison_14}, \eqref{eq:CIcomparison_15}, and $\Phi ^{-1}(1-\alpha )<\Phi ^{-1}(1-\alpha/2)$. For instance, using $a_1 =1.5$, $a_2  = 0.01$, and  $\alpha = 0.05$ yields $\lim_{N\to\infty}P_N(\theta_N\in CI_{\alpha}^{4}) = 0.974 > 0.948 = \lim_{N\to\infty}P_N(\theta_N\in CI_{\alpha}^{1}) =\lim_{N\to\infty}P_N(\theta_N\in CI_{\alpha}^{2})=\lim_{N\to\infty}P_N(\theta_N\in CI_{\alpha}^{3})$.
\end{proof}

\begin{proof}[Proof of Theorem \ref{thm:CI4}] To show this result, we construct specific sequences where \eqref{eq:orderbounds1} and \eqref{eq:orderbounds2} can occur. We focus on sequences $\{(P_{N},\theta _{N})\in \mathcal{P}\times \Theta _{I}(P_{N})^{c}\}_{N\in \mathbb{N}}$  s.t.\  
\begin{align}
&\Bigg( 
\begin{array}{c}
\theta _{l}( P_{N}) ,\theta _{u}( P_{N}) ,\sigma _{l}^{E}( P_{N}) ,\sigma _{u}^{E}( P_{N}) ,\rho ^{E}( P_{N}) ,\sigma _{l}^{I}( P_{N}) ,\sigma _{u}^{I}( P_{N}) ,\rho ^{I}( P_{N})  \\ 
\sqrt{N}( \theta _{u}( P_{N}) -\theta _{l}( P_{N}) ) ,\sqrt{N}( \theta _{l}( P_{N}) -\theta _{N}) ,\sqrt{N}( \theta _{N}-\theta _{u}( P_{N}) ) 
\end{array}
\Bigg)  \notag \\
&\to ( \theta _{l},\theta _{u},\sigma _{l}^{E},\sigma _{u}^{E},\rho ^{E},\sigma _{l}^{I},\sigma _{u}^{I},\rho ^{I}, \mu ,\Psi _{l},\Psi _{u}) .
\label{eq:CI4_comp}
\end{align}
with $\Psi _{l}\geq 0$ and $\mu \in \mathbb{R}_{+}$. By Lemma \ref{lem:near1}, $\rho^{E}=\rho^{I} =1$ and $\sigma _{l}^E=\sigma _{u}^E$, and $\sigma _{l}^{I}=\sigma _{u}^{I}$.

We can construct concrete sequences in the context of Example \ref{ex:OBS_Ex}. In particular, we use Example \ref{ex:OBS_Ex} with $\underline{Y}=0$, $\overline{Y}=1$, $\{Y_{i}|Z_{i}=1\}\sim Be(1/2)$, and $Z_{i}\sim Be(1-\pi (P_{N}))$, where $\pi (P_{N})~=~a_1/\sqrt{N}~\downarrow ~0$ for some $a_1 >0$. We consider coverage of $\theta _{N}=\theta _{l}(P_{N}) - a_2/\sqrt{N}\in\Theta _{I}(P_{N})^{c}$ for some $a_2>0$ and $CI_{\alpha }^{4}$ implemented with the subsample of the data $ \{(Y_{i},Z_{i})\}_{i=1}^{\lfloor \lambda N\rfloor }$ for $\lambda \in (0,1]$. In this context, we consider two estimators for the bounds. The efficient estimator uses the full sample, while the inefficient estimator uses only a fraction $\lambda = a_3 \in (0,1)$ of the sample. By repeating arguments in Theorem \ref{thm:CI1}, we deduce that \eqref{eq:CI4_comp} holds with
\begin{equation*}
     ( \theta _{l},\theta _{u},\sigma _{l}^{E},\sigma _{u}^{E},\rho ^{E},\sigma _{l}^{I},\sigma _{u}^{I},\rho ^{I}, \mu ,\Psi _{l}, \Psi_{u}) ~=~ \Big(\frac{1}{2},\frac{1}{2}, \frac{1}{2},\frac{1}{2},1, \frac{1}{2\sqrt{a_3}},\frac{1}{2\sqrt{a_3}},1,a_1,a_2,-(a_1+a_2)\Big).
\end{equation*}

Part (b) of Lemma \ref{lem:limits} then yields
\begin{align*}
\lim_{N \to \infty} P_{N}(\theta _{N} \in CI_{\alpha }^{4,E})& ~=~\Phi \big( 2( a_2+a_1)+\Phi ^{-1}( 1-\alpha /2) \big)-
\Phi \big( 2a_2-\Phi ^{-1}( 1-\alpha /2) \big), \\
\lim_{N \to \infty} P_{N}(\theta _{N} \in CI_{\alpha }^{4,I})&~=~
\Phi \big( 2\sqrt{a_3}(a_2+a_1 ) +\Phi ^{-1}( 1-\alpha /2) \big)-\Phi \big( 2\sqrt{a_3}a_2-\Phi ^{-1}( 1-\alpha /2) \big).
\end{align*}
We now verify the strict inequalities. To obtain \eqref{eq:orderbounds1}, set $a_1=a_2=1$, $a_3=0.3$, and $\alpha=0.05$, which give
\begin{align*}
\lim_{N\to \infty} P_{N}(\theta _{N} \in CI_{\alpha }^{4,E})~=~
0.48~<~0.81~=~\lim_{N\to \infty} P_{N}(\theta _{N} \in CI_{\alpha }^{4,I}).
\end{align*}
To obtain \eqref{eq:orderbounds2}, set $a_1=0.15$, $a_2=0.01$, $a_3=0.3$, and $\alpha=0.05$, which give
\begin{align*}
\lim_{N\to \infty} P_{N}(\theta _{N} \in CI_{\alpha }^{4,E})~=~0.963 ~>~ 0.958~=~\lim_{N\to \infty} P_{N}(\theta _{N} \in CI_{\alpha }^{4,I}).
\end{align*}
This completes the proof.
\end{proof}

\subsection{Auxiliary material}

This section collects auxiliary definitions and intermediate results that are useful for proving parts (e)-(f) of Theorem \ref{thm:CIcomparison} and Theorem \ref{thm:CI4}.

\begin{definition}\label{def:c_rho}
For any $\rho \in[ -1,1] $, $c( \rho )$ is the unique $c$ that solves 
\begin{equation*}
\inf_{\Delta \geq 0}P\Bigg( 
\begin{array}{c}
\{ z_{1}-\Delta -c\leq 0\leq \rho z_{1}+z_{2}\sqrt{1-\rho ^{2}}+c\}~ \cup  \\ 
\{ \vert ( 1+\rho ) z_{1}+z_{2}\sqrt{1-\rho ^{2}}-\Delta \vert \leq \sqrt{2+2\rho }\Phi ^{-1}( 1-\alpha /2) \} 
\end{array}
\Bigg) ~=~1-\alpha ,
\end{equation*}
where $z=( z_{1},z_{2}) \sim \mathcal{N}( \mathbf{0}_{2\times 1},\mathbf{I}_{2\times 2}) $.
\end{definition}

\begin{remark}
By definition, \cite{stoye:2020} uses $c^{4}=c(\hat{\rho})$. 
To see this, observe that $z=(z_{1},z_{2})\sim \mathcal{N}(\mathbf{0}_{2\times 1},\mathbf{I}_{2\times 2})$ is equivalent to $\tilde{z}( \rho )= (\tilde z_{1}(\rho),\tilde z_{2}(\rho)):=(z_1, \rho  z_{1}+z_{2}\sqrt{1-\rho ^{2}}) \sim \mathcal{N}( \mathbf{0}_{2\times 1},[ 1,\rho ;\rho ,1] ) $.
Therefore, $c(\rho)$ can equivalently be characterized as the unique $c$ solving
\begin{equation*}
\inf_{\Delta \geq 0}P\Bigg( 
\begin{array}{c}
\{ \tilde z_{1}(\rho)-\Delta-c\leq 0\leq \tilde z_{2}(\rho)+c\} \cup  \\
\{ |\tilde z_{1}(\rho)+\tilde z_{2}(\rho)-\Delta| \leq \sqrt{2+2\rho}\,\Phi^{-1}(1-\alpha/2) \}
\end{array}
\Bigg)~=~1-\alpha .
\end{equation*}
If we set $\rho =\hat \rho$ and condition on $\hat \rho$, the above expression coincides with the one in  \cite{stoye:2020}.
\end{remark}

\begin{lemma}\label{lem:cp_CTS}
$c( \rho ) $ in Definition \ref{def:c_rho} is a continuous function of $\rho \in [-1,1]$.
\end{lemma}
\begin{proof}
Recall $\bar{\mathbb{R}}_{+} = [0,\infty]$. Define the function $G:[-1,1]\times \bar{\mathbb{R}}_{+}\times \bar{\mathbb{R}}_{+} \to [0,1]$ as follows:
\begin{equation*}
G( \rho ,\Delta ,c)~=~P\left(
\begin{array}{c}
\{ z_{1}-\Delta -c\leq 0\leq \rho z_{1}+z_{2}\sqrt{1-\rho ^{2}}+c\}\cup\\
\{\vert( 1+\rho ) z_{1}+z_{2}\sqrt{1-\rho ^{2}}-\Delta\vert\leq \sqrt{2+2\rho }\Phi ^{-1}( 1-\alpha /2)\}
\end{array}
\right)~,
\end{equation*}
where $G(\rho,\Delta,\infty) = 1$ and $G(\rho,\infty,c) = \Phi(c)$. Define also $\widetilde{G}(\rho,c) = \inf_{\Delta \ge 0} G(\rho, \Delta, c)$. It is easy to see that $\widetilde{G}$ is continuous in all its arguments.\footnote{For any sequence $(\rho_n,c_n) \to (\rho,c)$, $ \exists \Delta_n, \Delta_0  \in \bar{\mathbb{R}}_{+}$ s.t. $\widetilde{G}(\rho_n,c_n) = G(\rho_n,\Delta_n,c_n)$ and $\widetilde{G}(\rho,c) = G(\rho,\Delta_0,c)$. By taking subsequence if necessary we can assume that $\Delta_n \to \bar{\Delta}$ for some $\bar{\Delta}$. When $\rho >-1$, the continuity of $G$ on  $(-1,1] \times  \bar{\mathbb{R}}_{+}\times \bar{\mathbb{R}}_{+}$ and definition of $\widetilde{G}$ imply $\widetilde{G}(\rho,c) \le  G(\rho,\bar{\Delta},c) = \lim_{n \to \infty} \widetilde{G}(\rho_n,c_n) \le \lim_{n \to \infty} G(\rho_n,\Delta_0,c_n) = \widetilde{G}(\rho,c)$; hence $\widetilde{G}(\rho,c) = \lim_{n \to \infty} \widetilde{G}(\rho_n,c_n)$. When $\rho=-1$, the result follows by similar arguments, using $\widetilde G(-1,c)=\Phi(c)$, the continuity of $G$ on $[-1,1]\times(0,\infty]\times\bar{\mathbb R}_+$ for the upper bound, and the fact that, uniformly over $\Delta\ge0$, the first event contains ${z_1\le c,\ \rho z_1+z_2\sqrt{1-\rho^2}\ge-c}$ for the lower bound.} By definition,
\[
c( \rho )~=~\Big\{c: \widetilde{G}(\rho,c)=1-\alpha\Big\},
\]
where the uniqueness of the solution follows from \cite{stoye:2020}.

Suppose $c( \rho ) $ is not a continuous function of $\rho \in [-1,1]$, i.e., $\{\rho_m\}_{m\in\mathbb{N}}$ with $\rho_m\to \rho\in[-1,1]$ and $c(\rho_m)\not\to c(\rho)$. By possibly taking a subsequence, we have that $c(\rho_{m})\to \bar{c} \neq c(\rho)$ as $m\to \infty$, where  $\bar{c} \in [0,\infty]$. By definition of $c(\rho_{m})$,
\begin{equation}
\widetilde{G}( \rho _{m},c( \rho _{m}))~=~1-\alpha.
\label{eq:c_defn_proof}
\end{equation}
Since $c(\rho)$ is unique and $\bar{c} \neq c(\rho)$, we have $\widetilde{G}( \rho,\bar{c})\neq 1-\alpha$. There are two cases: $\widetilde{G}( \rho,\bar{c})>1-\alpha$ or $\widetilde{G}( \rho  ,\bar{c})<1-\alpha$. To complete the proof, it suffices to show that both cases are contradictory.

\noindent {Case 1:} $\widetilde{G}( \rho, \bar{c})>1-\alpha$. Then, $\exists \varepsilon>0$ s.t.\ $\widetilde{G}( \rho, \bar{c})\ge 1-\alpha+\varepsilon.$ By continuity of $\widetilde{G}$, $\lim_{m\to \infty }\widetilde{G}( \rho _{m}, c( \rho _{m}))=\widetilde{G}( \rho ,\bar{c})\ge 1-\alpha+\varepsilon$. Therefore, $\exists M$  s.t.\  for all $m\ge M$, $\widetilde{G}( \rho _{m} ,c( \rho _{m}))\ge 1-\alpha+\varepsilon/2$, which is a contradiction to \eqref{eq:c_defn_proof}.

\noindent {Case 2:} $\widetilde{G}( \rho  ,\bar{c})<1-\alpha$. Then, $\exists \varepsilon>0$  s.t.\ $\widetilde{G} (\rho, \bar{c})  \le 1-\alpha-\varepsilon.$ By continuity of $\widetilde{G}$, $\lim_{m\to \infty }\widetilde{G}( \rho _{m}, c( \rho _{m}))=\widetilde{G}( \rho  ,\bar{c})\le 1-\alpha-\varepsilon/2$. Therefore, $\exists M$  s.t.\  for all $m\ge M$, $\widetilde{G}( \rho _{m},c( \rho _{m}))\le 1-\alpha-\varepsilon/4$, which is a contradiction to \eqref{eq:c_defn_proof}.
\end{proof}

\begin{lemma}\label{lem:cp_geq_PhiInv1malpha}
For any $\rho \in [-1,1]$ and $\alpha \in (0,1)$, $c( \rho ) $ in Definition \ref{def:c_rho} satisfies $c( \rho )\ge \Phi^{-1}(1-\alpha)$.
\end{lemma}
\begin{proof}
Fix $\rho \in [-1,1]$ arbitrarily. Assume the result is false, i.e., $c( \rho )<\Phi^{-1}(1-\alpha)$. Then, $\exists \varepsilon>0$ s.t. $c( \rho )\le\Phi^{-1}(1-\alpha) - \varepsilon$. Define
\begin{equation*}
\Pi(c)
~=~\inf_{\Delta \geq 0}
P\left(
\begin{array}{c}
\{ z_{1}-\Delta -c\leq 0\leq \rho z_{1}+z_{2}\sqrt{1-\rho ^{2}}+c\}
\cup
\\
\left\{
\begin{array}{c}
 \Delta -\sqrt{2+2\rho }\Phi ^{-1}( 1-\alpha /2) \leq \\
 z_{1}( 1+\rho ) +z_{2}\sqrt{1-\rho ^{2}}\leq \Delta +\sqrt{2+2\rho }\Phi ^{-1}( 1-\alpha /2)
\end{array}
\right\}
\end{array}
\right).
\end{equation*}
It is easy to see that $\Pi(c)$ is nondecreasing in $c$.

By $\Pi(c(\rho))=1-\alpha$ and $c( \rho )\le\Phi^{-1}(1-\alpha) - \varepsilon$, the monotonicity of $\Pi(c)$ implies that $\Pi(\Phi^{-1}(1-\alpha)-\varepsilon) \ge 1-\alpha$. 
Then, for any $\Delta \ge 0$, 
$$\Pi(\Phi^{-1}(1-\alpha)-\varepsilon)$$ is less than or equal to 
\begin{align*} 
&P\left(
\begin{array}{c}
\{ z_{1}-\Delta -\Phi ^{-1}( 1-\alpha ) + \varepsilon \leq 0\leq \rho z_{1}+z_{2}\sqrt{1-\rho ^{2}}+\Phi ^{-1}( 1-\alpha ) - \varepsilon \}
\cup
\\
\{ \Delta -\sqrt{2+2\rho }\Phi ^{-1}( 1-\alpha /2) \leq z_{1}( 1+\rho ) +z_{2}\sqrt{1-\rho ^{2}}\leq \Delta +\sqrt{2+2\rho }\Phi ^{-1}( 1-\alpha /2) \}
\end{array}
\right)~.
\end{align*}
Now evaluate this inequality for a sequence $\Delta_m\to \infty$ to obtain
\[
\Pi(\Phi^{-1}(1-\alpha)-\varepsilon) \le P(0\le \rho z_{1}+z_{2}\sqrt{1-\rho ^{2}}+\Phi ^{-1}(1-\alpha) -\varepsilon)  = \Phi\left( \Phi^{-1}(1-\alpha)-\varepsilon\right)~,
\]
where we used $\tilde z=\rho z_{1}+z_{2}\sqrt{1-\rho ^{2}}\sim N(0,1)$ in the last equality. Finally, we obtain
\begin{align*}
\Phi\left( \Phi^{-1}(1-\alpha) \right)~\leq~ \Pi(\Phi^{-1}(1-\alpha)-\varepsilon) ~\le~\Phi\left( \Phi^{-1}(1-\alpha)-\varepsilon\right),
\end{align*}
which is a contradiction.
\end{proof}

\begin{lemma}\label{lem:c1}
For any $\alpha \in (0,1)$, $c( \rho ) $ in Definition \ref{def:c_rho} satisfies $c(1)=\Phi^{-1}(1-\alpha/2)$.
\end{lemma}
\begin{proof}
By definition,
\begin{align*}
c(1)&~=~\left\{c:\inf_{\Delta \geq 0}P\left(
\begin{array}{c}
\{-c\leq z\leq \Delta+c\}~~~ \cup \\
\{\Delta/2-\Phi^{-1}(1-\alpha/2)\leq z\leq \Phi^{-1}(1-\alpha/2)+\Delta/2\}
\end{array}
\right)=1-\alpha\right\}.
\end{align*}

Given any $c \in  \mathbb{R}$, consider the minimization of
\begin{align*}
G(\Delta)&=~\Phi(\max\{\Phi^{-1}(1-\alpha/2)+\Delta/2,\Delta+c\})-\Phi(\min\{\Delta/2-\Phi^{-1}(1-\alpha/2),-c\})
\end{align*}
with respect to $\Delta\ge 0$. We have two cases.

\noindent {Case 1:} $2\Phi^{-1}(1-\alpha/2)-2c \ge 0$. If $\Delta \geq 2\Phi^{-1}(1-\alpha/2)-2c $, then, $\Delta+c\geq \Phi^{-1}(1-\alpha/2)+\Delta/2$ and $-c\leq \Delta/2-\Phi^{-1}(1-\alpha/2)$, and so $G(\Delta)=\Phi(\Delta+c)-\Phi(-c).$ Hence $G$ is increasing in $\Delta$ on $[2\Phi^{-1}(1-\alpha/2)-2c,\infty)$. If  $\Delta <2\Phi^{-1}(1-\alpha/2)-2c$, then, $G(\Delta)=\Phi(\Phi^{-1}(1-\alpha/2)+\Delta/2)-\Phi(\Delta/2-\Phi^{-1}(1-\alpha/2))$. The function $G$ is decreasing in $\Delta$ on $[0,2\Phi^{-1}(1-\alpha/2)-2c]$. Therefore, in this case, $G$ is minimized by setting $\Delta=2\Phi^{-1}(1-\alpha/2)-2c$.

\noindent {Case 2:} $2\Phi^{-1}(1-\alpha/2)-2c<0$. Since $\Delta \ge 0$, we have $\Delta \geq 2\Phi^{-1}(1-\alpha/2)-2c$.  The same argument with the constraint $\Delta\ge 0$ yields $\Delta=0$.

Combining all cases,
\begin{align}
\inf_{\Delta \geq 0}G(\Delta)~=~\left\{
\begin{array}{c}
I\{\Phi^{-1}(1-\alpha/2)\ge c\}(\Phi(2\Phi^{-1}(1-\alpha/2)-c)-\Phi(-c)) \\
+ I\{\Phi^{-1}(1-\alpha/2)<c\}(\Phi(c)-\Phi(-c))
\end{array}
\right\}.
\label{eq:G_RHS}
\end{align}
Now, we solve for $c$ s.t.\ \eqref{eq:G_RHS} equals $1-\alpha$ which, by definition, equals $c(1)$. 

We now show that $\Phi^{-1}(1-\alpha/2)\geq c$. Suppose otherwise that $\Phi^{-1}(1-\alpha/2)<c$. Then, the right-hand side of \eqref{eq:G_RHS} equals $\Phi(c)-\Phi(-c)$. Then, the desired solution $c$ satisfies $\Phi(c)-\Phi(-c)=1-\alpha$, which yields $\Phi(c)=1-\alpha/2$, and, so, $c=\Phi^{-1}(1-\alpha/2)$, which contradicts the premise that $\Phi^{-1}(1-\alpha/2)<c$.

Since $\Phi^{-1}(1-\alpha/2)\geq c$, the right-hand side of \eqref{eq:G_RHS} equals $\Phi(2\Phi^{-1}(1-\alpha/2)-c)-\Phi(-c)$. Then, the desired solution $c$ satisfies
\begin{equation}
\Phi(2\Phi^{-1}(1-\alpha/2)-c)-\Phi(-c)~=~1-\alpha.
\label{eq:G_RHS2}
\end{equation}
Note that the left-hand side of \eqref{eq:G_RHS2} is strictly increasing in $c$ whenever $\Phi^{-1}(1-\alpha/2)\geq c$ holds.  
Therefore, the unique solution is $c=\Phi^{-1}(1-\alpha/2)$, as desired.
\end{proof}

\begin{lemma}
\label{lem:limits}
Let $\alpha\in(0,0.5)$ and Assumption \ref{ass:1} hold. Let $\{(P_{N},\theta_{N})\in \mathcal{P}\times \Theta_{I}(P_{N})^{c}\}_{N\in\mathbb{N}}$ be a sequence  s.t.\ 
\begin{align}
&\bigg(
\begin{array}{c}
\theta_{l}(P_{N}),\theta_{u}(P_{N}),\sigma_{l}(P_{N}),\sigma_{u}(P_{N}),\rho(P_{N}), \\
\sqrt{N}(\theta_{u}(P_{N})-\theta_{l}(P_{N})),\sqrt{N}(\theta_{l}(P_{N})-\theta_{N}),\sqrt{N}(\theta_{N}-\theta_{u}(P_{N}))
\end{array}
\bigg)\notag\\
&\to~ (\theta_{l},\theta_{u},\sigma_{l},\sigma_{u},\rho,\mu,\Psi_{l},\Psi_{u}).\label{eq:sequences2}
\end{align}
Moreover, assume that $\Psi_{l}\ge 0$. Then,
\begin{enumerate}[(a)]
\item If $\mu=\infty$, $P_{N}(\theta_{N}\in CI_{\alpha}^{4})\to \Phi(c(\rho)-\Psi_{l}/\sigma_{l}).$
\item If $\mu\in\mathbb{R}_{+}$, then $\rho=1$, $\sigma_{l}=\sigma_{u}$, and, if we denote $\sigma=\sigma_{l}=\sigma_{u}$, we get
\begin{equation*}
P_{N}(\theta_{N}\in CI_{\alpha}^{4})~\to~ \Phi((\Psi_{l}+\mu)/\sigma+\Phi^{-1}(1-\alpha/2)) - \Phi(\Psi_{l}/\sigma-\Phi^{-1}(1-\alpha/2)).
\end{equation*}
\end{enumerate}
\end{lemma}
\begin{proof}
As a preliminary result, note that
\begin{align}
\Psi_{u}~=~-\lim (\sqrt{N}(\theta_{u}(P_{N})-\theta_{l}(P_{N}))+\sqrt{N}(\theta_{l}(P_{N})-\theta_{N})) 
~=~-\mu-\Psi_{l}. 
\label{eq:Psi_connect}
\end{align}
We divide the proof into two parts.

\noindent \underline{Part (a)}: In this case, $\mu=\infty$. Then $\Psi_{l}\ge 0$ and, by \eqref{eq:Psi_connect}, $\Psi_{u}=-\infty$. Hence,
\begin{align*}
&P_{N}(\theta_{N}\in CI_{\alpha}^{4})\\
&=~P_{N}(\theta_{N}\in CI_{\alpha}^{4,a}\cup CI_{\alpha}^{4,b}) \\
&=~ P_{N}\left(
\begin{array}{c}
\left\{
\begin{array}{c}
\{\sqrt{N}(\hat{\theta}_{l}-\theta_{l}(P_{N}))/\hat{\sigma}_{l}-c(\hat{\rho})\le -\sqrt{N}(\theta_{l}(P_{N})-\theta_{N})/\hat{\sigma}_{l}\}
\cap\\
\left\{
\begin{array}{c}
-\sqrt{N}(\theta_{l}(P_{N})-\theta_{N})/\hat{\sigma}_{u}-\sqrt{N}(\theta_{u}(P_{N})-\theta_{l}(P_{N}))/\hat{\sigma}_{u}
\\
\le \sqrt{N}(\hat{\theta}_{u}-\theta_{u}(P_{N}))/\hat{\sigma}_{u}+c(\hat{\rho})
\end{array}
\right\}
\end{array}
\right\}\\
\cup\left\{
\begin{array}{c}
\sqrt{N}(\theta_{N}-\theta_{u}(P_{N}))/\hat{\sigma}_{u}-\sqrt{N}(\theta_{l}(P_{N})-\theta_{N})/\hat{\sigma}_{l}-\sqrt{2+2\hat{\rho}}\Phi^{-1}(1-\alpha/2)\\
\le \sqrt{N}(\hat{\theta}_{l}-\theta_{l}(P_{N}))/\hat{\sigma}_{l}+\sqrt{N}(\hat{\theta}_{u}-\theta_{u}(P_{N}))/\hat{\sigma}_{u}\le\\
\sqrt{N}(\theta_{N}-\theta_{u}(P_{N}))/\hat{\sigma}_{u}-\sqrt{N}(\theta_{l}(P_{N})-\theta_{N})/\hat{\sigma}_{l}+\sqrt{2+2\hat{\rho}}\Phi^{-1}(1-\alpha/2)
\end{array}
\right\}
\end{array}
\right)\\
&~\overset{(1)}{\to}~
P\left(
\begin{array}{c}
\{\{z_{1}-c(\rho)\le -\Psi_{l}/\sigma_{l}\}\cap\{-\Psi_{l}/\sigma_{u}-\mu/\sigma_{u}\le \rho z_{1}+z_{2}\sqrt{1-\rho^{2}}+c(\rho)\}\}\cup\\
\left\{
\begin{array}{c}
\Psi_{u}/\sigma_{u}-\Psi_{l}/\sigma_{l}-\sqrt{2+2\rho}\Phi^{-1}(1-\alpha/2) \le (1+\rho)z_{1}+z_{2}\sqrt{1-\rho^{2}}\\
\le \Psi_{u}/\sigma_{u}-\Psi_{l}/\sigma_{l}+\sqrt{2+2\rho}\Phi^{-1}(1-\alpha/2)
\end{array}
\right\}
\end{array}
\right)
\\
&~\overset{(2)}{=}~ P(z_{1}-c(\rho)\le -\Psi_{l}/\sigma_{l})~=~\Phi(c(\rho)-\Psi_{l}/\sigma_{l}),
\end{align*}
as desired, where (1) holds by \eqref{eq:sequences2}, Lemma \ref{lem:cp_CTS}, and OBS (Definition \ref{def:setup}), and (2) by $\mu=\infty$, $\Psi_{l}\geq 0$, and \eqref{eq:Psi_connect}, which implies that $\Psi_{u}=-\infty$.

\noindent \underline{Part (b)}: In this case,  $\mu\in\mathbb{R}_{+}$. Lemma \ref{lem:near1} then implies that $\rho=1$ and $\sigma_{l}=\sigma_{u}$. Let $\sigma=\sigma_{l}=\sigma_{u}$. By a similar derivation as in part (a),
\begin{align*}
P_{N}(\theta_{N}\in CI_{\alpha}^{4})
&~\overset{(1)}{\to}~
P\left(
\begin{array}{c}
\{\{z_{1}-c(1)\le -\Psi_{l}/\sigma\}\cap \{-(\Psi_{l}+\mu)/\sigma\le z_{1}+c(1)\}\}\\
\cup\left\{
\begin{array}{c}
(\Psi_{u}-\Psi_{l})/\sigma-2\Phi^{-1}(1-\alpha/2)\le 2z_{1}\\
\le (\Psi_{u}-\Psi_{l})/\sigma+2\Phi^{-1}(1-\alpha/2)
\end{array}
\right\}
\end{array}
\right)\\
&~\overset{(2)}{=}~P(-(\Psi_{l}+\mu)/\sigma-\Phi^{-1}(1-\alpha/2)\le z_{1}\le -\Psi_{l}/\sigma+\Phi^{-1}(1-\alpha/2))\\
&~=~\Phi((\Psi_{l}+\mu)/\sigma+\Phi^{-1}(1-\alpha/2)) - \Phi(\Psi_{l}/\sigma-\Phi^{-1}(1-\alpha/2)),
\end{align*}
where (1) holds by \eqref{eq:sequences2}, Lemma \ref{lem:cp_CTS}, and OBS (Definition \ref{def:setup}), and (2) by Lemma \ref{lem:c1}, $\mu \geq 0$, $\Psi_{l}\geq 0$, and \eqref{eq:Psi_connect}, which implies that $\Psi_{u}=-\mu-\Psi_{l}$.
\end{proof}

\end{appendix}

\bibliography{BIBLIOGRAPHY}

\end{document}